\documentclass[twoside,leqno,twocolumn]{article}

\usepackage[letterpaper]{geometry}

\usepackage{ltexpprt}

\usepackage{url}
\usepackage{multicol}
\usepackage{multirow}
\usepackage{tikz}
\usetikzlibrary{calc}
\usepackage{amsmath,amsfonts,amssymb}
\usepackage{algorithmicx, algpseudocode}
\usepackage{subcaption}
\usepackage{makecell}

\begin{document}

\title{\Large On Optimal Partitioning For Sparse Matrices In Variable Block Row Format\thanks{ %
This work was supported by a Department of Energy Computational Science %
Graduate Fellowship, DE-FG02-97ER25308. %
This work also funded by the Department of Energy's Exascale Computing %
Program (ECP). Sandia National Laboratories is a multimission laboratory %
managed and operated by National Technology and Engineering Solutions of %
Sandia, LLC., a wholly owned subsidiary of Honeywell International, Inc., for %
the U.S. Department of Energy's National Nuclear Security Administration %
under contract DE-NA-0003525. This research used resources of the National %
Energy Research Scientific Computing Center (NERSC), a U.S. Department of %
Energy Office of Science User Facility operated under Contract No. %
DE-AC02-05CH11231. %
}}

\author{Peter Ahrens\thanks{Massachusetts Institute of Technology and Sandia National Laboratories, pahrens@mit.edu.}%
\and Erik G. Boman\thanks{Sandia National Laboratories, egboman@sandia.gov.}}

\date{}

\maketitle

% Copyright Statement
% When submitting your final paper to a SIAM proceedings, it is requested that you include 
% the appropriate copyright in the footer of the paper.  The copyright added should be 
% consistent with the copyright selected on the copyright form submitted with the paper.
% Please note that "20XX" should be changed to the year of the meeting.

% Default Copyright Statement TODO
%\fancyfoot[R]{\scriptsize{Copyright \textcopyright\ 20XX by SIAM\\
%Unauthorized reproduction of this article is prohibited}}

% Depending on which copyright you agree to when you sign the copyright form, the copyright 
% can be changed to one of the following after commenting out the default copyright statement
% above.

%\fancyfoot[R]{\scriptsize{Copyright \textcopyright\ 20XX\\
%Copyright for this paper is retained by authors}}

%\fancyfoot[R]{\scriptsize{Copyright \textcopyright\ 20XX\\
%Copyright retained by principal author's organization}}

%\pagenumbering{arabic}
%\setcounter{page}{1}%Leave this line commented out.

\begin{abstract} \small\baselineskip=9pt
  The Variable Block Row (VBR) format is an influential blocked sparse matrix
  format designed for matrices with shared sparsity structure between
  adjacent rows and columns. VBR groups adjacent rows and columns, storing
  the resulting blocks that contain nonzeros in a dense format. This reduces
  the memory footprint and enables optimizations such as register blocking
  and instruction-level parallelism. Existing approaches use heuristics to
  determine which rows and columns should be grouped together. We show that
  finding the optimal grouping of rows and columns for VBR is NP-hard under
  several reasonable cost models. In light of this finding, we propose a
  1-dimensional variant of VBR, called 1D-VBR, which achieves better
  performance than VBR by only grouping rows. We describe detailed cost
  models for runtime and memory consumption. Then, we describe a linear time
  dynamic programming solution for optimally grouping the rows for 1D-VBR
  format. We extend our algorithm to produce a heuristic VBR partitioner
  which alternates between optimally partitioning rows and columns, assuming
  the columns or rows to be fixed, respectively. Our alternating heuristic
  produces VBR matrices with the smallest memory footprint of any partitioner
  we tested.

%  We adapt and
%  optimize a dynamic programming algorithm for sequential hypergraph
%  partitioning to produce a linear time algorithm which can determine the
%  optimal partition of rows under an expressive cost model, assuming the
%  column partition remains fixed. Furthermore, we show that the problem of
%  determining an optimal partition for the rows and columns simultaneously is
%  NP-Hard under a simple linear cost model.

%  To evaluate our algorithm empirically against existing heuristics, we
%  introduce the 1D-VBR format, a specialization of VBR format where columns
%  are left ungrouped. We evaluate our algorithms on all
%  \input{results/number_of_matrices}real-valued matrices in the SuiteSparse
%  Matrix Collection. When asked to
%  minimize an empirically derived cost model for a sparse matrix-vector
%  multiplication kernel, our algorithm produced partitions whose 1D-VBR
%  realizations achieve a speedup of at least
%  \input{results/quartile_minimizecompute_speedup}over an unblocked kernel on
%  25\% of the matrices, and a speedup of at least
%  \input{results/octile_minimizecompute_speedup}on 12.5\% of the matrices.
%  The 1D-VBR representation produced by our algorithm had faster SpMVs than
%  the 1D-VBR representations produced by any existing heuristics on
%  \input{results/optimal_vs_heuristic_percent}of the test matrices.
\end{abstract}

\section{Introduction}

Matrices that occur in practice are often \textbf{sparse}, meaning that most of
their entries are zero, and it is faster to process only the nonzero
entries\cite{saad_iterative_2003,vuduc_oski:_2005}. Some applications produce
matrices where nonzeros occur close together. In these cases, we can reduce the
complexity and storage requirements of processing and locating individual
nonzeros by storing the nonzeros in dense blocks. We need only store the size
and location of the block, and can employ dense performance engineering
techniques like register blocking and instruction-level parallelism.

Blocked formats are most commonly used to accelerate multiplication between a
sparse matrix and a dense vector (\textbf{SpMV}). SpMV is often used as a
subroutine in iterative solvers; the same sparse matrix is multiplied
hundreds of times before a solution is found. A practical use case is to
block the matrix once, then offset the cost of finding the blocks and
converting the matrix format with the savings obtained after multiplying the
matrix many times in an iterative solver.

Dense blocks have also been used in supernodal sparse factorizations
\cite{davis_survey_2016, demmel_supernodal_1999}, incomplete factorizations
\cite{saad_iterative_2003,saad_finding_2003} and in sparse triangular solves
\cite{yamazaki_performance_2020}. Originally, only rows/columns with
identical sparsity patterns were merged, but the approach can be relaxed to
merge rows/columns with merely similar patterns \cite{ashcraft_influence_1989}. Note that for
preconditioning, block methods are mathematically different and may affect
the convergence rate, typically making methods more robust. Kim et. al.
\cite{kim_task_2016} extended the ``algorithms by block'' concept from dense
to sparse linear algebra and showed it is useful for task parallel systems on
modern architectures.

One of the first blocked formats to receive considerable study was the
Variable Block Row (\textbf{VBR}) format, where similar adjacent rows and
columns are grouped together \cite{vuduc_fast_2005,
karakasis_perfomance_2009, karakasis_comparative_2009}. VBR is described in
the SPARSKIT library \cite{saad_sparskit_1990,saad_sparskit_1994}, the
SparseBLAS specification \cite{remington_nist_1996}, and the OSKI Sparse
Kernel Interface \cite{vuduc_automatic_2004, vuduc_oski:_2005}, and is used
internally by the MKL Paradiso solver \cite{noauthor_developer_2020}. Unlike
many formats which use fixed-size blocks, the number of rows or columns that
may be grouped together is allowed to vary along each dimension, producing
variably sized blocks. Since blocks are produced by merging entire rows or
columns, the blocks are aligned, allowing implementations to reuse elements
of other arguments along the direction of alignment. In general, producing
bigger blocks means that less location information is needed, but as blocks
get bigger, they may cover and store more zeros explicitly in
dense storage.

While the VBR format was motivated by scientific applications that produce
matrices with perfect block structure (where nearby rows with identical
patterns might represent different partial derivatives of the same variable
or different variables of the same multiphysics mesh point), we investigate
the application of these techniques to sparse matrices with imperfect block
structure (where nearby rows with correlated patterns might represent
neighboring mesh points or similar mathematical programming constraints).

Block partitioning algorithms sometimes reorder rows to group similar
rows together \cite{shantharam_exploiting_2011, saad_finding_2003,
pinar_improving_1999}. In this work, as in the definition of VBR, we consider
only contiguous partitions (splitting without reordering). While
noncontiguous partitions allow for more expressive blocks, permuting a matrix
or vector may be an expensive memory-intensive procedure. In some situations,
the matrix may have already been reordered for numerical reasons, and the
user might need to operate on the matrix without changing the row ordering.
Furthermore, there are several matrices which do not need reordering to
utilize similarities among adjacent rows. Our cost models apply to
contiguous or noncontiguous partitions alike, and research into contiguous
partitioning may inform general partitioning approaches.

Although the contiguous case may appear simpler, in
Appendix~\ref{app:vbrblockingnphard}, we prove that determining optimal
groupings of adjacent rows and columns for VBR format is NP-Hard under two
simple cost models by reduction from the Maximum Cut problem
\cite{karp_reducibility_1972,papadimitriou_optimization_1991}. The problem is
still NP-Hard even when the row and column partitions are constrained to be
the same, a symmetric constraint required by some block factorizations. In
light of this observation, we invent a specialization of the VBR sparse
matrix format for the case where the columns are simply ungrouped, making
optimal partitioning tractable. We refer to this new, simpler, format as
1D-VBR. 1D-VBR enjoys many of the same benefits as VBR, obtaining better
performance at the cost of slightly more memory usage.

At the time of writing, only heuristic algorithms have been given to
determine which rows or columns should be grouped together. We first propose
detailed cost models which describe the number of blocks, the memory
footprint, or the expected SpMV runtime of the resulting VBR or 1D-VBR
format, inspired by \cite{buttari_performance_2007,
karakasis_perfomance_2009}. We then describe a linear time, single pass,
algorithm to determine the optimal contiguous row groups under a fixed column
grouping and general cost model, inspired by
\cite{grandjean_optimal_2012, jackson_algorithm_2005, alpert_multiway_1995,
ziantz_run-time_1994, kernighan_optimal_1971}. Our algorithm is optimal for
1D-VBR since ungrouped columns are a fixed column grouping. We can also build
a VBR heuristic by alternately partitioning just the rows, then columns, then
rows again, similar to
\cite{kolda_partitioning_1998,hendrickson_graph_2000,yasar_heuristics_2019}.
Our algorithm runs in time $O(R \cdot (u_{\max} \cdot m + N) + n)$, where $R$ is the
rank of the cost function (a small constant), $u_{\max}$ is the maximum block
height (a small constant \cite{vuduc_performance_2002}), and $m$, $n$, and $N$ are the number of rows,
columns, and nonzeros, respectively. Our algorithm requires only one pass directly on the
sparse matrix in CSR format.

%
%In Section \ref{sec:heuristics}, we propose a
%novel modification to the implementation of the commonly used ``overlap
%similarity'' heuristic to ensure that it reads each element of the matrix
%only once.

We test our optimal algorithm against existing heuristics on a test set of
20\! real-valued sparse matrices with
interesting, imperfect, block structure (10 of which are the test set of
\cite{vuduc_fast_2005}), in both 64 and 32 bit precision. Using our heuristic
to reduce the VBR memory footprint resulted in the the best compression.
Using our algorithm to minimize our empirical cost model of 1D-VBR SpMV
runtime resulted in the best performance, achieving a median speedup of
$2.22$$\times$
\! over the reference CSR
implementation. On at least half of the matrices, the overhead of
partitioning and conversion to 1D-VBR format was justified within
$16.9$
\! multiplications, demonstrating
the practicality of our techniques.

\section{Partitioning}

Let $A$ be an $m \times n$ \textbf{matrix} with $N$ nonzeros, where $A_{i,j}$
corresponds to the value in the $i^{\text{th}}$ row and $j^{\text{th}}$ column.
We use $i{:}i'$ to represent the integer sequence $i, i + 1, ..., i'$, and
describe submatrices of $A$ as $A_{i : i',j : j'}$.  When the argument to a
function is clear from context, we will omit the argument for brevity. 

In practice, adjacent rows (and columns) often contain similar patterns of
nonzero locations. To capitalize on this observation, we will group similar
rows together, forming a \textbf{row part}. A \textbf{$K$-partition} $\Pi$ of
rows of $A$ assigns each row $i$ to one of $K$ \textbf{parts} $\Pi_k$. In
this work, we insist that our partitions stay \textbf{contiguous}, meaning
that for $k < k'$, $i \in \Pi_k$, $i' \in \Pi_{k'}$, we have $i < i'$. We
use $\Pi^{-1}$ to refer to the length $n$ vector of part assignments, so that
when $i \in \Pi_k$, $\Pi^{-1}_i = k$. When $\Pi$ is contiguous, we can
represent it with a vector $spl_\Pi$ of $K + 1$ split points, so that $\Pi_k
= spl_{\Pi, k} : spl_{\Pi, k + 1} - 1$. A partition is \textbf{trivial} if it
assigns each row to a distinct part.

%$\vec{1}$ to represent the \textbf{unit partition}, which assigns each row to the
%same part. Currently unused

We may also impose an $L$-partition $\Phi$ of columns. The partitions $\Pi$
and $\Phi$ tile our matrix with $K \times L$ contiguous, non-overlapping,
rectangular \textbf{blocks}. The block $(k, l)$ is of size $u_k \times w_l =
|\Pi_{k}| \times |\Phi_{l}|$. Blocked formats store only \textbf{nonzero
blocks}, or blocks that contain at least one nonzero of $A$. If we partition
wisely, many blocks will be zero and not need to be stored.

We use $v_i(A, \Phi) = \{l | A[i, \Phi_l] \neq 0\}$ to refer to the set of column parts containing nonzeros in the
$i^{th}$ row of $A$, and $\gamma_k(A, \Pi, \Phi) = \bigcup_{i \in \Pi_k} v_i$ to refer to the set of column parts
containing nonzeros in the $k^{th}$ row part of $A$.

\section{Sparse Formats}

Sparse matrices are commonly stored in Compressed Sparse Row (\textbf{CSR}) format
\cite{saad_iterative_2003}, which consists of three vectors $pos$, $idx$, and
$val$. The length $m + 1$ vector $pos$ stores the regions of $idx$ and $val$
corresponding to each row. The vectors $idx$ and $val$ (each of length $N$)
store the sorted nonzero column locations and corresponding values,
respectively. Storing $A$ in CSR format uses
\begin{equation}\label{eq:csrmemory}
    s_{\text{CSR}}(A) = (m + 1)s_{\text{index}} + N \cdot s_{\text{index}} + N \cdot s_{\text{value}}
\end{equation}
bits, where $s_{\text{index}}$ and $s_{\text{value}}$ are the sizes of the index
and value types, in bits.

%Assume we were to store $A$ in CSR format and
%wanted to determine the value $A_{i, j}$. We start by finding the unique index
%$l$ such that $A.pos[i] \leq l < A.pos[i + 1]$ and $A.idx[l] = j$. If no such
%$l$ exists, $A_{i, j} = 0$ and is not stored explicitly in CSR format.
%Otherwise, $A_{i, j} = A.val[l]$. Note that $A.pos[1]$ will always be $1$ and
%$A.pos[m + 1]$ will always be $N + 1$.

The Variable Block Row (\textbf{VBR}) format imposes a contiguous
$K$-partition $\Pi$ of rows and a contiguous $L$-partition $\Phi$ of columns
\cite{saad_sparskit_1990, saad_sparskit_1994, remington_nist_1996,
vuduc_automatic_2004, vuduc_oski:_2005}. It is illustrated in Figure \ref{fig:vbr-example}. It is convenient to store the
length $K + 1$ and $L + 1$ split vectors $spl_{\Pi}$ and $spl_{\Phi}$,
respectively. Instead of storing individual nonzero locations, the VBR format
saves memory using the $idx$ vector to store block indices (the indices
record the parts corresponding to each block). The positions of the
variably-sized blocks in $val$ are not aligned with the positions of the
block indices in the $idx$ array. Therefore, we use a vector $ofs$ of block
locations to encode the starting index of each block row in $val$. Assume we
were to store $A$ in VBR format and wanted to determine the value of the
entry $A_{i, j}$. Let $k = \Pi^{-1}_i$ and $l = \Phi^{-1}_j$ (when
partitions are contiguous, we can compute this with binary search on the
split vectors). If we cannot find $q$ such that $A.pos[k] \leq q < A.pos[k +
1]$ and $A.idx[q] = l$, then $A_{i, j} = 0$ because the block $(k, l)$ is
entirely zero and is not stored explicitly in VBR format. Otherwise, our
block contains at least one nonzero and starts at position $p = A.ofs[k] +
\sum_{l' \in \gamma_k | l' < l} u_k * w_{l'}$ in the $val$ array.
%Our block has dimensions $u_k \times w_l = (A.spl_\Pi[k + 1] - A.spl_\Pi[k]) \times
%(A.spl_\Phi[l + 1] - A.spl_\Phi[l])$.
Because VBR stores nonzero blocks in a dense, column-major format, $A_{i, j}
= A.val[p + (j - A.spl_\Phi[l]) \cdot u_k + (i - A.spl_\Pi[k])]$.

Let $N_{\text{index}}(A, \Pi, \Phi)$ be the number of nonzero blocks induced by $\Pi$ and $\Phi$, so that
\begin{equation}\label{eq:numberofblocks}
    %N_{\text{index}}(A, \Pi, \Phi) = |\{(k, l) | A_{\Pi_k, \Phi_l} \neq 0\}|.
    N_{\text{index}}(A, \Pi, \Phi) = \sum\limits_{k = 1}^K |\gamma_k|.
\end{equation}
Let $N_{\text{value}}(A, \Pi, \Phi)$ be the number of entries contained in all nonzero
blocks induced by $\Pi$ and $\Phi$, such that
\begin{equation}\label{eq:numberofvalues}
    N_{\text{value}}(A, \Pi, \Phi) = \sum\limits_{k = 1}^K\sum\limits_{l \in \gamma_k} u_k \cdot w_l.
\end{equation}
The VBR format uses six arrays, $spl_\Pi$, $spl_\Phi$, $pos$, $ofs$, $idx$,
and $val$. Storing $A$ in VBR format uses $s_{\text{VBR}}$ bits,
\begin{multline}\label{eq:vbrmemory}
    s_{\text{VBR}}(A, \Pi, \Phi) = \\ (3(K + 1) + (L + 1) + N_{\text{index}})s_{\text{index}} + N_{\text{value}}s_{\text{value}}.
\end{multline}

In this work, we introduce a novel specialization of VBR format where the column
partition is trivial, meaning that the columns are not grouped together and
we concern ourselves only with row partitioning. We call this special format
\textbf{1D-VBR}. An example is illustrated in Figure \ref{fig:1dvbr-example}.
Because $\Phi$ is trivial, we do not need to store it.
Additionally, because blocks have only one column, block sizes are constant
within each row part and the stride between blocks is constant within each
row part, simplifying the implementation of conversion and multiplication
routines. Assume we were to store $A$
in 1D-VBR format and wanted to determine the value of the entry $A_{i, j}$.
Let $k = \Pi^{-1}[i]$. Notice that $\Phi^{-1}[j] = j$, since $\Phi$ is
trivial. If we cannot find $q$ such that $A.pos[k] \leq q < A.pos[k + 1]$ and
$A.idx[q] = j$, then $A[i, j] = 0$ because the block $(k, j)$ is entirely
zero and is not stored explicitly in 1D-VBR format. Otherwise, our block
contains at least one nonzero and starts at position $p = A.ofs[k] +
(q - A.pos[k]) \cdot u_k$ in the $val$ array. Thus, $A_{i, j} = A.val[p + (i
- A.spl_\Pi[k])]$. The 1D-VBR format uses five arrays, $spl_\Pi$, $pos$, $idx$,
$ofs$, and $val$. Storing $A$ in 1D-VBR format uses $s_{\text{1D-VBR}}$ bits,
\begin{multline}\label{eq:1dvbrmemory}
    s_{\text{1D-VBR}}(A, \Pi) = \\(3(K + 1) + N_{\text{index}})s_{\text{index}} +  N_{\text{value}}s_{\text{value}}.
\end{multline}

\subsection{Related Sparse Formats}

Blocked sparse formats have enjoyed a long history of study. In lieu of
providing an exhaustive overview of existing formats, we refer the reader to
works such as \cite{razzaq_dynb_2017, karakasis_comparative_2009,
vuduc_automatic_2004} which provide summaries of several sparse blocking
techniques. We focus only on the most relevant formats here. Figure
\ref{fig:examples} illustrates some examples of relevant formats.

The BCSR format tiles the matrix with fixed-size dense format blocks, storing
nonzero block locations in CSR format \cite{im_optimizing_2000,
im_optimizing_2001, im_sparsity:_2004, eberhardt_optimization_2016,
choi_model-driven_2010}. BCSR is referred to as BSR in the
Intel\textsuperscript{\textregistered} Math Kernel Library \cite{noauthor_developer_2020}. Cost
models developed for BCSR depend on the number of nonzero blocks, leading to
the development of row-wise sampling algorithms to estimate the number of
nonzero blocks \cite{im_optimizing_2000,
im_sparsity:_2004, vuduc_automatic_2004, lee_performance_2004, vuduc_oski:_2005,
buttari_performance_2007, karakasis_perfomance_2009}. These row-wise sampling
algorithms were improved on by a constant time nonzero-wise sampling
algorithm \cite{ahrens_fill_2018,xu_fill_2018,ahrens_parallel_2019}.

Generalizing to less constrained block decompositions, unaligned block
formats continue to use fixed-size blocks, but relax alignment requirements.
The SPARSKIT implementation of BCSR relaxes the column alignment of blocks,
allowing blocks to shift along the block rows \cite{saad_sparskit_1994}. One
could imagine a format which groups adjacent blocks in 1D-VBR block rows to
achieve a similar format. The UBCSR format uses a number of fixed block sizes
that can start at any entry in the matrix \cite{vuduc_fast_2005}. An
intriguing approximation algorithm has been described for the related NP-hard
problem of finding good fixed-sized, unaligned, nonoverlapping sparse matrix
block decompositions (the UBCSR format) \cite{vassilevska_finding_2004}. The
CSR-SIMD format produces dense blocks inside the rows, putting successive
groups of nonzeros into SIMD-register sized blocks for instruction level
parallelism \cite{chen_efficient_2018}. Note that SpMVs on CSR-SIMD formatted
matrices cannot reuse loads from the input vector, whereas 1D-VBR uses only
one load from the input for each block, no matter how large the block is. The
1D-Variable Block Length (1D-VBL) format, originally proposed in
\cite{pinar_improving_1999} and referred to as 1D-VBL in
\cite{karakasis_comparative_2009}, relaxes the constraint that the blocks
inside rows must be of fixed length. Both 1D-VBL and CSR-SIMD can reduce the
size of the matrix when nonzeros occur next to each other in the same row.
The Variable Blocked-$\sigma$-SIMD Format (VBSF) is similar to CSR-SIMD, but
allows the blocks to be merged across multiple rows, so the blocks have a
fixed width but variable height. The DynB format relaxes all alignment and
size constraints, allowing variably sized blocks to start at any entry of the
matrix \cite{razzaq_dynb_2017}. Algorithms for producing CSR-SIMD, VBSR, and
DynB formats create their blocks with greedy algorithms that add adjacent
elements into the block up to a density-related threshold. Because these
formats make decisions on a block-by-block basis, it makes sense to convert
the matrix to blocked format at the same time as the block decomposition is
determined \cite{chen_efficient_2018,razzaq_dynb_2017}. 

\begin{figure}
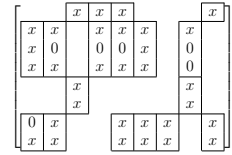
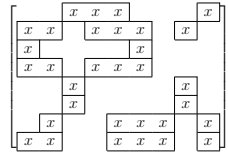
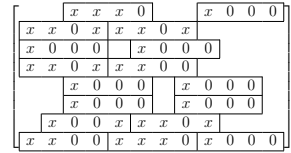
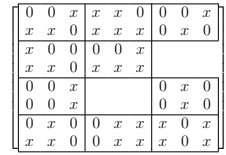
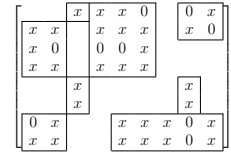

\centering

\begin{minipage}[b]{.46\linewidth}\centering\large
\input{vbr_example}
\end{minipage}\hspace{0.06\linewidth} %
\begin{minipage}[b]{.46\linewidth}\centering\large
\input{1dvbr_example}
\end{minipage}

\leavevmode\newline

\begin{minipage}[b]{.46\linewidth}\centering\large
\input{vbl_example}
\end{minipage}\hspace{0.06\linewidth} %
\begin{minipage}[b]{.46\linewidth}\centering\large
\input{csrsimd_example}
\end{minipage}

\leavevmode\newline

\begin{minipage}[b]{.46\linewidth}\centering\large
\input{bcsr_example}
\end{minipage}\hspace{0.06\linewidth} %
\begin{minipage}[b]{.46\linewidth}\centering\large
\input{ubcsr_example}
\end{minipage}

\caption{Various blocked sparse representations of a sample
matrix $A$. Here, $x$ represents a nonzero, $0$ represents an explicitly
stored zero, and each box represents a distinct stored block. Implicit zeros
are left blank. Most formats store nonzero blocks in row-major order
analogous to how CSR stores nonzero entries.}\label{fig:examples}
\end{figure}

\section{Blocked SpMV}

Algorithm \ref{alg:spmv_vbr} shows an example SpMV kernel for a matrix stored
in VBR format. Processing each stored element of $A$ requires a load from
$A.val$, but we only need to load from $A.idx$ and $x$ once for each block
and column in the block row, respectively. This data reuse is a benefit of
producing aligned blocks, and a key property enjoyed by VBR and 1D-VBR but
not by CSR-SIMD. Of course, computations and sequential loads are now
processed with vector instructions. If our vector size does not divide our
block size, we simply pad our vectors as they are loaded from memory, without
needing to pad the stored blocks. For example, if our blocks are of size 3,
we can process them using vectors of size 4, letting the fourth entry of our $y$
vector register be undefined. While this does not affect the number of blocks
or the memory usage, it does have an effect on the empirical runtime.

\begin{algorithm}\label{alg:spmv_vbr}
    Given $m \times n$ matrix $A$ in VBR format and a length $n$ vector
    $x$, add $A \cdot x$ to the length $m$ vector $y$, in-place. 
    \begin{algorithmic}[1]
        \small
        \Function{SpMV-VBR}{$y$, $A$, $x$}
            \State{$p \gets 1$}
            \For{$k \gets 1 \textbf{ to } K$}
                \State{$yy \gets y[A.spl_\Pi[k]:A.spl_\Pi[k + 1] - 1]$}
                \State{$u \gets A.spl_\Pi[k + 1] - A.spl_\Pi[k]$}
                \For{$q \gets A.pos[k] \textbf{ to } A.pos[k + 1] - 1$}\label{alg:spmv_vbr:1D-VBR}
                    \State{$l \gets A.idx[q]$}
                    \For{$j \gets A.spl_\Phi[l] \textbf{ to } A.spl_\Phi[l + 1] - 1$}\label{alg:spmv_vbr:inner}
                        \State{$yy \gets yy + A.val[p:p + u - 1] \cdot x[j]$}
                        \State{$p \gets p + u$}
                    \EndFor
                \EndFor
                \State{$y[A.spl_\Pi[k]:A.spl_\Pi[k + 1] - 1] \gets yy$}
            \EndFor
        \EndFunction
    \end{algorithmic}
\end{algorithm}
To modify \textproc{SpMV-VBR} for 1D-VBR, we need only replace the loop on line 
\ref{alg:spmv_vbr:1D-VBR} with the simpler inner loop:
\begin{algorithmic}[1]
    \small
    \For{$q \gets A.pos[k] \textbf{ to } A.pos[k + 1] - 1$}
        \State{$j \gets A.idx[q]$}
        \State{$yy \gets yy + A.val[p:p + u - 1] \cdot x[j]$}
        \State{$p \gets p + u$}
    \EndFor
\end{algorithmic}

We have designed our implementation of Algorithm \ref{alg:spmv_vbr} so that
we can optimize the code to use a computed jump instruction to select
between dedicated unrolled loop bodies for each block size $u$ and $w$. This
allows us to pad the vertical (SIMD) dimension to the nearest vector width
and unroll the horizontal dimension with minimal overhead. Note that in the
1D-VBR algorithm, this jump occurs once per row block, since all blocks in
the row block have the same height.

\section{Partitioning Problem Statement}\label{sec:problem}

Because the blocks in a VBR format are stored in a dense format, we must
trade off between a partition that uses larger blocks (and stores more
explicit zeros) and a partition that uses smaller blocks (and stores more
block locations). Practitioners often use cost models to measure the effect
of performance parameters like block sizes. Several diverse cost models have
been proposed for blocked sparse matrix formats
\cite{im_sparsity:_2004,vuduc_automatic_2004,nishtala_when_2007}. While
many of these models apply to VBR SpMV \cite{karakasis_perfomance_2009},
we are not aware of any work which takes the next step to use the cost model
to optimize a VBR partition. 

To simplify the presentation of our algorithms, we keep our three cost models
simple. Our first model is simply the number of blocks,
\eqref{eq:numberofblocks}. This model should perform well on matrices which
fit in fast memory, when the cost of computing a block is only weakly
dependent on its size. The second model assumes that runtime will be directly
proportional to the memory footprint \eqref{eq:vbrmemory} or
\eqref{eq:1dvbrmemory}. Because SpMV is a memory-bound kernel
\cite{yzelman_generalised_2015} and many sparse matrices do not fit in fast
memory, we expect this model to work well for large matrices.

Our third, more general, cost model is inspired by
\cite[(2)]{karakasis_perfomance_2009} and
\cite[(3)]{buttari_performance_2007}, which both model the time taken to
compute a row part $k$ of height $u_k$ as an affine function in the number of
elements in the part. Cost models with similar forms have been proposed for
similar blocked formats \cite{im_optimizing_2000, im_optimizing_2001,
im_sparsity:_2004, vuduc_automatic_2004}.
\pagebreak[3]
\begin{multline}\label{eq:vbrcompute}
    t_{VBR}(A, \Pi, \Phi) = \\ \sum\limits_{k = 1}^K \alpha_{\text{row}, u_k} + \sum\limits_{l = 1}^L \alpha_{\text{col}, w_l} + \sum\limits_{k = 1}^K \sum\limits_{l \in \gamma_k} \beta_{u_k, w_l}.
\end{multline}

The vectors $\alpha$ represent the costs associated with row or column parts,
such as loading elements from $x$, $y$, or $A.\Pi$, etc. The coefficient matrix
$\beta$ represents the cost of each block. The runtime for each block size is
represented individually because the relationship between block size and
performance is architecturally dependent and not easily characterized,
especially since we pad the block to the next available vector size. In
practice, however, $\beta$ is well approximated by a low rank matrix.  Equation
\eqref{eq:vbrcompute} applies to 1D-VBR, but is much simplified since $\Phi$ is
trivial, so $w_l = 1$ and $\beta$ becomes a vector.

We parameterize \eqref{eq:vbrcompute} with empirical measurements. For each
block size, we measure the time to multiply the smallest square matrix with an
average of $8$ blocks per block row such that the problem exceeds the L2 cache
size. We then benchmark the same problem with twice the blocks, or twice the
rows, or twice the columns. Finally, we use a least-squares fit (normalized to
minimize relative error of each datapoint), and then approximate $\beta$ with a
singular value decomposition of rank $3$, which we found kept the relative error
within $2\%$. Measurement noise sometimes led the cost function to incentivize
larger blocks; we encourage monotonicity by using the prefix maximum of the
measured cost. Taking these empirical measurements takes a few hours, but only
needs to be performed once per architecture. Because our empirical cost model
uses real measurements, it can account for factors like memory bandwidth or
padding to fit in SIMD registers, or potentially unanticipated decisions that
other implementers may make.

Our main problem can be stated as follows:

\begin{Definition}[Block Partitioning]\label{prob:vbrblocking}
Given an $m \times n$ matrix $A$ and block size limit $u_{\max} \times
w_{\max}$, find the contiguous $K$-partition $\Pi$ and $L$-partition $\Phi$
minimizing a cost function of the form
\[
    \sum\limits_{k = 1}^K \alpha_{\text{row}, u_k} + \sum\limits_{w = 1}^L \alpha_{\text{col}, w_l} + \sum\limits_{k = 1}^K \sum\limits_{l \in \gamma_k} \sum\limits_{r = 1}^R \beta_{\text{row}, u_k, r} * \beta_{\text{col}, w_l, r},
\]
where $R$ is the rank of the block cost matrix.
\end{Definition}

All previously described cost functions \eqref{eq:numberofblocks},
\eqref{eq:1dvbrmemory}, \eqref{eq:vbrmemory}, and \eqref{eq:vbrcompute} are
expressible in the above form. The block count is rank 1, and the memory
usage is rank 2.

We show in Appendix \ref{app:vbrblockingnphard} that Problem
\ref{prob:vbrblocking} is NP-Hard for both the very simple cost model
\begin{equation}\label{eq:nphardcost}
    f(A, \Pi, \Phi) = s\cdot N_{\text{index}} + N_{\text{value}},
\end{equation}
where $s \geq 1$ is small constant and $u_{\max} \geq 2$, and for
\begin{equation}\label{eq:nphardcost2}
    f(A, \Pi, \Phi) = N_{\text{index}}.
\end{equation}
These cost models approximately minimize the memory usage
\eqref{eq:1dvbrmemory} or \eqref{eq:vbrmemory}, or simply the number of
blocks \eqref{eq:numberofblocks}, and are special cases of the fully generic
form of Definition \ref{prob:vbrblocking}. A corollary will show the problem
is still NP-hard even when the row and column partitions are constrained to be
the same, the symmetric case. Our proof reduces from the Maximum Cut problem
\cite{karp_reducibility_1972, papadimitriou_optimization_1991}.  We represent
vertices of the graph with block rows in a matrix, and edges as block columns,
inserting gadgets at the endpoints of each edge. Then we show that we can
construct a maximum cut from the optimal VBR blocking for the gadgets.

\subsection{1D Partitioning and Alternation}
The remainder of the paper will be focused on situations where $\Phi$ is
considered fixed. We will propose an optimal, linear-time algorithm for this restricted
problem. Because we can convert any row partitioning algorithm to a column
partitioning algorithm by simply transposing our matrix first
\cite{gustavson_two_1978}, without loss of generality we consider only the
row partitioning case.

In the case of 1D-VBR, $\Phi$ is fixed to be the trivial partition,
and our restricted solution can optimize our cost models exactly.

The restricted solution also gives rise to an alternating heuristic for the
original VBR problem where we iteratively partition the rows first, then
partition the columns under the new row partition, and so on. In each
iteration, the previously fixed partition provides an upper bound on the
optimal value, so the objective always decreases and the process eventually
converges. In the symmetric case, when the column and row partitions must be
the same, we could set the column partition equal to the row partition after
each iteration, but could no longer make similar guarantees on convergence.
Alternating heuristics have been applied to problems like graph partitioning
and load balancing
\cite{kolda_partitioning_1998,hendrickson_graph_2000,yasar_heuristics_2019}.
Existing VBR heuristics partition rows and columns separately from each
other, without incorporating information from one when partitioning the
other.
%TODO could cut above sentence if needed

\section{Heuristics}\label{sec:heuristics}

Existing VBR implementations use heuristics instead of directly optimizing
partitions to minimize cost models. The heuristics used for rows are ignorant
of column partitions and vice versa.

The VBR implementation in SPARSKIT uses a heuristic we refer to as the
\textproc{StrictPartitioner}, that groups identical rows and columns
\cite{saad_sparskit_1990, saad_sparskit_1994}. One can easily determine that
two consecutive rows have identical sparsity patterns by coiterating over
their patterns. This straightforward heuristic can be quite effective when
the block structure is obvious, which is sometimes the case for FEM matrices
or supernodal structures in direct factorizations. The
\textproc{StrictPartitioner} reads each nonzero at most twice, and the reads
are sequential. Producing the VBR or 1D-VBR matrix from the computed partition
is also easier with the added assumption that sparsity patterns are identical
within blocks, since we can easily compute the block sparsity pattern from
the nonzero sparsity pattern.

The VBR implementations of OSKI and MKL PARADISO relax the
\textproc{StrictPartitioner} in favor of the \textproc{OverlapPartitioner}, a
heuristic which groups rows that satisfy some similarity requirement
\cite{vuduc_oski:_2005, noauthor_developer_2020}. The rows are initially
ungrouped, and each row $i'$ is processed in turn from top to bottom. Let $i$
be the first row in the group immediately preceding $i'$. The overlap
heuristic adds $i'$ to $i$'s group if the height would not exceed $u_{\max}$
and
\[
    \frac{|v_i \cap v_{i'}|}{\min(|v_i|, |v_i'|)} \geq \rho,
\]
otherwise, we start a new group with row $i'$. This process is repeated for
the columns, producing $\Pi$ and $\Phi$. The above similarity metric is known
as the \textbf{overlap similarity}, although the cosine similarity is also
sometimes used for greedy noncontiguous partitioning
\cite{saad_finding_2003}.

The OSKI code base uses a binary vector $h$ of length $n$ as a perfect hash
table to calculate the size of the intersection, first setting $h_j$ to true
for each $j \in v_i$, then iterating over elements of $v_{i'}$, checking to
see if corresponding locations in $h$ have been set to true. When we start a
new group, we must iterate through $v_i$ again to reinitialize $h$ to false.
Because $h$ will be used at most $m$ times, if we instead change $h$ to be a
integer vector and store the value $i$ at each location of $h$ when
calculating intersections with $v_i$, we need only iterate over $v_i$ once.
Since we will build on this concept when introducing our optimal algorithm,
we relegate the pseudocode for our improved implementation of the overlap
heuristic to Appendix \ref{app:overlappartitioner}.

\section{Optimal Algorithm}\label{sec:optimal}

Recall that our restricted Problem \ref{prob:vbrblocking} asks us to compute
an optimal row partition under some fixed column partition $\Phi$. In
situations where the runtime of the partitioner is justified with respect to
the runtime savings of the target kernel, efficient algorithms that operate
directly on the input matrix are desirable. Thus, we seek a linear time
algorithm that can partition the matrix with only one pass over the nonzeros.
Our problem has optimal substructure, and we use dynamic programming on the rows
of the matrix $A$. Given that the optimal partition of rows $i'$ through $m$ has
cost $c_{i'}$, then the cost of the optimal partition of rows $i$ through $m$
can be computed as
\[
    \max\limits_{1 \leq u \leq u_{max}} \alpha_{\text{row}, u} + \sum\limits_{r = 1}^R \sum\limits_{l \in \gamma'} \beta_{\text{row}, u, r} * \beta_{\text{col}, w_l, r} + c_{i'},
\]
where $i' = i + u$ and $\gamma' = v_{i} \cup ... \cup v_{i' - 1}$.  We record
each of the $R$ partial sums $\sum_{l \in \gamma'} \beta_{\text{col},w_{l},r}$
in the variables $d_r$. We use a vector $h$ to remember the most recent row
in which we saw each nonzero column part, which allows us to efficiently update
the vectors $\Delta_r$, the changes in $d_r$ as we increment $i'$. We can later
multiply each $d_r$ by $\beta_{\text{row}, u,r}$ and sum to produce the total
cost of each candidate row part, in turn. Each best block size is recorded in a
vector $spl_{\Pi}$, and once the vector is full, we follow these pointers to
construct a partition in-place, completing our algorithm. Since $\Phi$ is fixed,
we can safely ignore $\alpha_{\text{col}}$.

Our approach is shown in Algorithm \ref{alg:optimalpartitioner}. Recall that
CSR format provides convenient iteration over $v_i$ in sorted order.

\begin{algorithm}\label{alg:optimalpartitioner}
    Given an $m \times n$ sparse matrix $A$, a column partition $\Phi$, a
    maximum block height $u_{\max}$, compute a row partition $\Pi$ minimizing the
    cost function
    \[
        \sum\limits_{k = 1}^K \alpha_{\text{row}, u_k} + \sum\limits_{l = 1}^L \alpha_{\text{col}, w_l} + \sum\limits_{k = 1}^K \sum\limits_{r = 1}^R \sum\limits_{l \in \gamma_k} \beta_{\text{row}, u_k, r} * \beta_{\text{col}, w_l, r},
    \]
    
    \begin{algorithmic}[1]
        \small
        \Function{OptimalPartitioner}{$A$, $\Phi$, $\alpha$, $\beta$}
            \State{\textrm{Allocate length-$(m + 1)$ vectors $c$, $spl_\Pi$, $\Delta_1, ..., \Delta_R$}}
            \State{$c[m + 1], \Delta_1, ..., \Delta_R \gets 0$}
            \State{\textrm{Allocate length-$L$ vector $h$ initialized to $m + 1$}}
            \State{\textrm{Compute $\Phi^{-1}$}}

            \For{$i \gets m \textup{ to } 1$} \label{alg:optimalpartitioner:main} \Comment{Iterating Backwards!}
                \For{$l \in v_i$} \label{alg:optimalpartitioner:state}
                    \State{$w \gets |\Phi[l]|$}
                    \For{$r \gets 1 \textup{ to } R$}
                        \State{$\Delta_{r}[i] \gets \Delta_{r}[i] + \beta_{\text{col}, w, r}$}
                        \State{$\Delta_{r}[h[l]] \gets \Delta_{r}[h[l]] - \beta_{\text{col}, w, r}$}
                    \EndFor
                    \State{$h[l] \gets i$}
                \EndFor
                \State $(d_1, ..., d_R) \gets 0$
                \State{$c[i] \gets \infty$}
                \For{$i' \gets i + 1 \textup{ to } \min(i + u_{max}, m + 1)$}\label{alg:optimalpartitioner:optimize}
                    \State{$u \gets i' - i$}
                    \State{$c' \gets \alpha_{\text{row}, u}$}
                    \For{$r \gets 1 \textup{ to } R$}
                        \State{$d_{r} \gets d_{r} + \Delta_{r}[i' - 1]$}
                        \State{$c' \gets c' + d_{r} * \beta_{\text{row}, u, r}$}
                    \EndFor
                    \If{$c' + c[i'] < c[i]$}
                        \State{$c[i] \gets c' + c[i']$}
                        \State{$spl_\Pi[i] \gets i'$}
                    \EndIf
                \EndFor
            \EndFor
            \State $K \gets 0$
            \State $i \gets 1$
            \While{$i \neq m + 1$}\label{alg:optimalpartitioner:pointers}
                \State{$i' \gets spl_\Pi[i]$}
                \State{$K \gets K + 1$}
                \State{$spl_\Pi[K] \gets i$}
                \State{$i \gets i'$}
            \EndWhile
            \State{$spl_\Pi[K + 1] \gets m + 1$}
            \State{\Return $spl_\Pi[1:K + 1]$}
        \EndFunction
    \end{algorithmic}
\end{algorithm}

Our algorithm owes much of its structure to related algorithms
\cite{aydin_distributed_2019, grandjean_optimal_2012, jackson_algorithm_2005,
ziantz_run-time_1994, alpert_multiway_1995, ziantz_run-time_1994,
kernighan_optimal_1971}. For example, an optimal algorithm for the related
problem of ``restricted hypergraph partitioning'' (producing contiguous
partitions that reduce communication in parallel SpMV) is described by
Grandjean et. al. \cite{grandjean_optimal_2012}. However, this algorithm is
described for simpler cost functions which do not apply directly to Problem
\ref{prob:vbrblocking}. Furthermore, this algorithm does not consider a
column partition. Since the algorithm is given as a reduction from the
hypergraph problem to a graph problem, it requires multiple passes over the
input. While all of the cited algorithms are similar, none of them apply
directly.

\subsection{Runtime, Optimizations, and Extensions}
The body of the loop at line \ref{alg:optimalpartitioner:state} can be
executed in $O(R)$ time and will be repeated at most once for each nonzero in
$A$. The body of the loop at line \ref{alg:optimalpartitioner:optimize} can
be executed in $O(R)$ time and will be repeated at most $u_{\max}$ times for
each row in $A$. Initialization takes $O(R \cdot m + n)$ time. The cleanup loop is
accomplished in $O(m)$ time. Thus, Algorithm \ref{alg:optimalpartitioner}
runs in $O(R \cdot (u_{\max} \cdot m + N) + n)$ time. In this work, we
consider cost functions of rank at most 3. The number of blocks
\eqref{eq:numberofblocks} is a rank 1 cost, the memory usages
\eqref{eq:1dvbrmemory} and \eqref{eq:vbrmemory} are rank 2 costs, and we
approximate our empirical model \eqref{eq:vbrcompute} to rank 3.

In practice, diminishing returns are observed for max block sizes $u_{\max}$
or $w_{\max}$ beyond approximately $12$ \cite{vuduc_performance_2002}. Even
in theory, increasing the block size to $u_{\max}'$ will only further
amortize the index storage over more rows, so the additional compression is
bounded by a factor of $(s_{\text{index}}/u_{\max} +
s_{\text{value}})/(s_{\text{index}}/u_{\max}' + s_{\text{value}})$. When
$s_{\text{index}} = s_{\text{value}}$ and $u_{\max} = 8$, doubling $u_{\max}$
will further compress by at most a factor of $18/17$.

Because the outer dynamic program loop on line
\ref{alg:optimalpartitioner:optimize} works backwards, all of the innermost
loops access memory in storage order. This also allows us to construct
$spl_\Pi$ in place. To improve the empirical running time of our partitioner,
we implemented a specialization for the case where $\Phi$ is the trivial
partition. In this case, $\Phi^{-1}[j] = j$, $w$ is always 1 and the costs
are all rank 1, so much of the algorithm can be simplified.

SpMV is often parallelized. If a coarse-grained partition is applied to the
rows or columns so that each part executes on one of $P$ processors, then any
of the algorithms or heuristics presented can be applied to the local regions
to create fine-grained block subpartitions. If, however, one wishes to
compute the block decomposition before the processor decomposition
(partitioning, for example, block rows instead of rows), the previously
described approach is still a practical option, but concatenating the
resulting local partitions is not guaranteed to be optimal, since it imposes
$P$ artificial split points on the processor boundaries. Instead, one might
choose to compute the best partition for all $u_{\max}^2$ combinations of
start and end points within boundary regions of size $u_{\max}$ (multiplying the
serial workload by $u_{\max}$), and then stitch these optimal solutions
together using dynamic programming over each region in $O(P
* u_{\max})$ time. This strategy should be considered when $u_{\max}$ is small
in comparison to $P$, but we expect our first suggestion to be sufficient in
practice.

If the blocks in our format are themselves sparse \cite{nishtala_when_2007, buluc_representation_2008, buluc_parallel_2009, zhang_making_2017, hong_adaptive_2019, namashivayam_variable-sized_2021},
we may be interested in a cost function which models both the number of
nonzero blocks and the number of reads from $x$ required to process the block
row. Note that existing contiguous cache blocking heuristics use aggregate or
probabilistic models to find splits, as opposed to calculating the actual
reuse. If a cost model depends on more than one simultaneous column
partition, we suggest using more than one copy of $h$, $\Delta$, and $d$ to
calculate costs.

\section{Conversion}

After producing a partition, we need to convert the matrix from CSR to VBR or
1D-VBR format. We use a hash table to compute the size of $pos$ and $ofs$ in
one pass over the matrix, an algorithm quite similar to that of Algorithms
\ref{alg:optimalpartitioner} and \ref{alg:overlappartitioner}. It is possible
to fuse computation of $pos$ and $ofs$ with the partitioning itself, but we
did not notice enough of a speedup to justify the added complexity. Our
conversion algorithms are similarly expensive to our partitioners.

%If we need to parallelize the conversion
%algorithm, each process can start at the locations defined by $pos$ and
%$ofs$.

If all the rows in each group are identical, as is the case with partitions
produced by the \textproc{StrictPartitioner}, the nonzero patterns from the
CSR representation can be copied directly from any row in each part to form
the $idx$ array in 1D-VBR format. We can pack $val$ through simultaneous
iteration over all the rows in the part, since we don't need to fill.
%TODO for final submission
%If we also deduplicate adjacent column parts, we can convert to VBR format
%with the \textproc{StrictPartitioner} in a similarly efficient manner.

If, however, all the rows in each group are not identical, then we must merge
the nonzero patterns of each row in a part to produce the nonzero block
patterns. If each CSR row contains a sorted list of the elements in $v_i$,
then our goal is to form a sorted list of the elements in $\bigcup_{i \in
\Pi_k} v_i$. This is a similar problem to merging $w$ sorted
lists. Algorithms exist to solve such a problem in $O(\log(w)(|\bigcup_{i \in
\Pi_k} v_i|) + \sum|v_i|)$ time \cite[``HeapMerge'']{greene_k-way_1991}.
Since we also need to fill all $u\cdot(|\bigcup_{i \in \Pi_k} v_i|)$ entries
of the $val$ array with either nonzeros or explicit zeros, we instead use a
linear search over the rows to find the minimum index, then iterate over the
rows to fill the corresponding elements of $idx$ and $val$
\cite[``LinearSearchMerge'']{greene_k-way_1991}. The direct merge algorithm
is the simpler choice when $u$ is small, which we have assumed it is. The
algorithm for producing block rows in VBR or 1D-VBR format is similar enough
to \cite[``LinearSearchMerge'']{greene_k-way_1991} that we omit it. It is
enough to know that the number of operations performed by the conversion
algorithm is proportional to the size of the resulting format.

\section{Results}

We ran our programs on the ``Haswell'' partition of the ``Cori'' NERSC
Supercomputer. We used a single core of a 16-core Intel\textsuperscript{\textregistered}
Xeon\textsuperscript{\textregistered} Processor E5-2698 v3 running at 2.3 GHz with 32 KB of L1
cache per core, 256 KB of L2 cache per core, 41 MB of shared L3 cache, and
128 GB of memory. This CPU supports the AVX2 instruction set, meaning that it
supports SIMD processing with 256 bit vector lanes.

All kernels were implemented\footnote{Code is available at %
\url{https://github.com/peterahrens/SparseMatrix1DVBCs.jl/releases/tag/2005.12414v2} %
and %
\url{https://github.com/peterahrens/ChainPartitioners.jl/releases/tag/2005.12414v2} %
} in Julia 1.5.3
\cite{bezanson_julia:_2017}. Because Julia is compiled just-in-time, it
enjoys powerful metaprogramming capabilities. This allowed us to create a
custom SpMV subkernel for each block size in our VBR and 1D-VBR SpMV kernel.
Hard-coding block sizes allows the compiler to perform important
optimizations like loop unrolling. Since our matrices were real-valued, the
value datatypes were floating point numbers \cite{noauthor_ieee_2019}, and
the index datatypes were 64 bit signed integers. We represented our matrices
with both 64 bit and 32 bit precision, so each SIMD vector fit 4 or 8
elements, respectively. In our tests, our maximum block size was $u_{\max} =
w_{\max} = 8$ for Float64 and $u_{\max} = w_{\max} = 16$ for Float32 because
we found that further increasing $u_{\max}$ did not increase performance by
much. If the block size of a row part was $1$, we used scalar instructions.
Otherwise, we used one or two vectors to process the row part. We used the
SIMD.jl library to emit explicit LLVM vector instructions for each block size
\cite{erik_schnetter_eschnettsimdjl_2019}. We benchmark with a warm cache,
meaning that we run the kernel once before beginning to measure it. We run
the Julia garbage collector before taking each benchmark. After warming up
the cache, each kernel is sampled one million times or until 10 seconds of
measurement time is exceeded (we allow the kernel to complete before
stopping), whichever happens first. All benchmarks use the minimum sampled
time.

Our test set includes all the matrices of \cite{vuduc_fast_2005}, which have
clear block structure. To diversify our test collection, we also include
several matrices with imperfect block structure from the SuiteSparse Matrix
Collection \cite{davis_university_2011}. This includes matrices with large
dense structures that must be balanced against sparse structures elsewhere,
such as ``exdata\_1,'' or ``TSOPF\_RS\_b678\_c1,'' matrices with large, dense
blocks that may be interrupted by isolated sparse rows and columns, like
``Goodwin\_071,'' ``lpi\_gran,'' or ``heart3,'' and matrices where nonzeros
are merely clustered, rather than appearing in clear blocks, such as in
``ACTIVSg70K,'' ``scircuit,'' or ``rajat26.'' The ``Janna/*'' matrices have
clear blocks, but benefit from merging blocks together. The matrices are
described in Table~\ref{tbl:matrices}.

\begin{table}[!ht]
    \caption{Test matrices used in addition to those of
    \cite{vuduc_fast_2005}. ``Spy'' is the sparsity pattern of the matrix
    $A$ (pixels represent logarithmic density of a region), and ``Zoomed Spy'' is the pattern of a representative $128 \times 128$
    region (pixels represent individual matrix entries). Blue is zero, and yellow is nonzero.}\label{tbl:matrices}
    \setlength{\tabcolsep}{2pt}\begin{tabular}{c@{\hspace{1pt}}ccc}
Spy & \makecell{Zoomed \\ Spy} & \makecell{Group / Name \\ $m \times n$ ($N$)} \\ \hline \\
\raisebox{-.5\height}{\includegraphics[width=1.7cm]{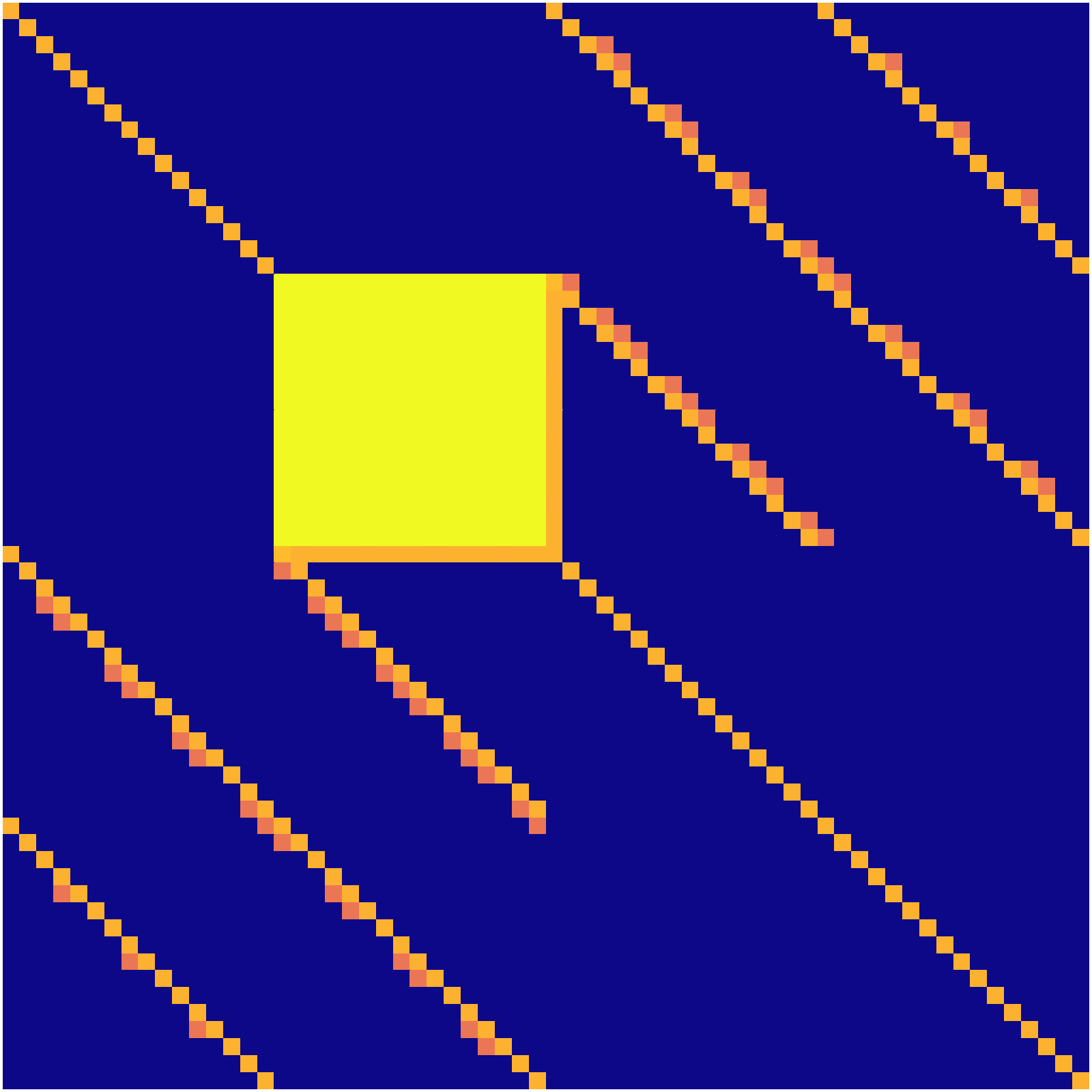}}\vspace{1pt} & \raisebox{-.5\height}{\includegraphics[width=1.7cm]{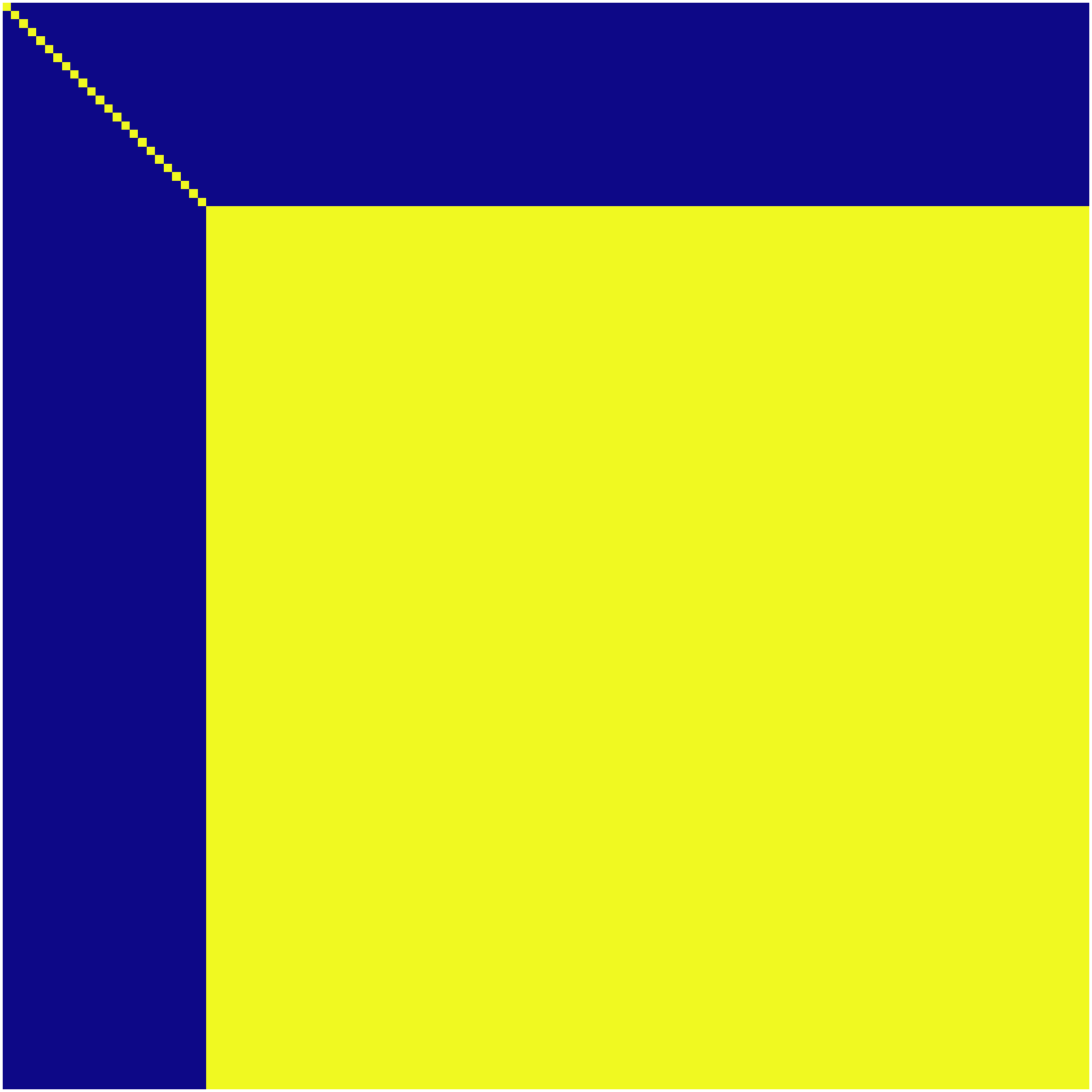}} & \makecell{ GHS\_indef/ \\ \quad exdata\_1 \\ $6001 \times 6001$ ($2269501$)} \\
\raisebox{-.5\height}{\includegraphics[width=1.7cm]{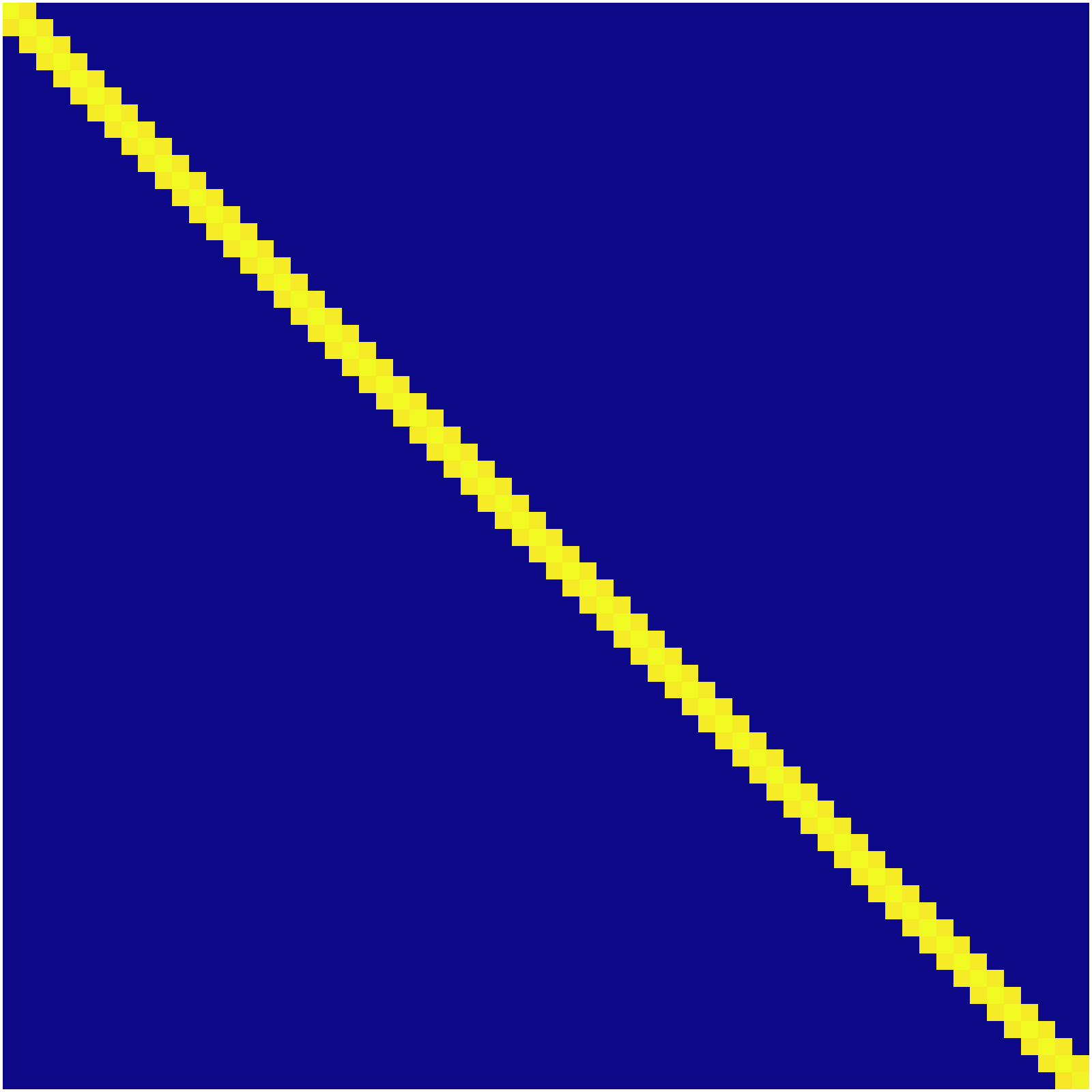}}\vspace{1pt} & \raisebox{-.5\height}{\includegraphics[width=1.7cm]{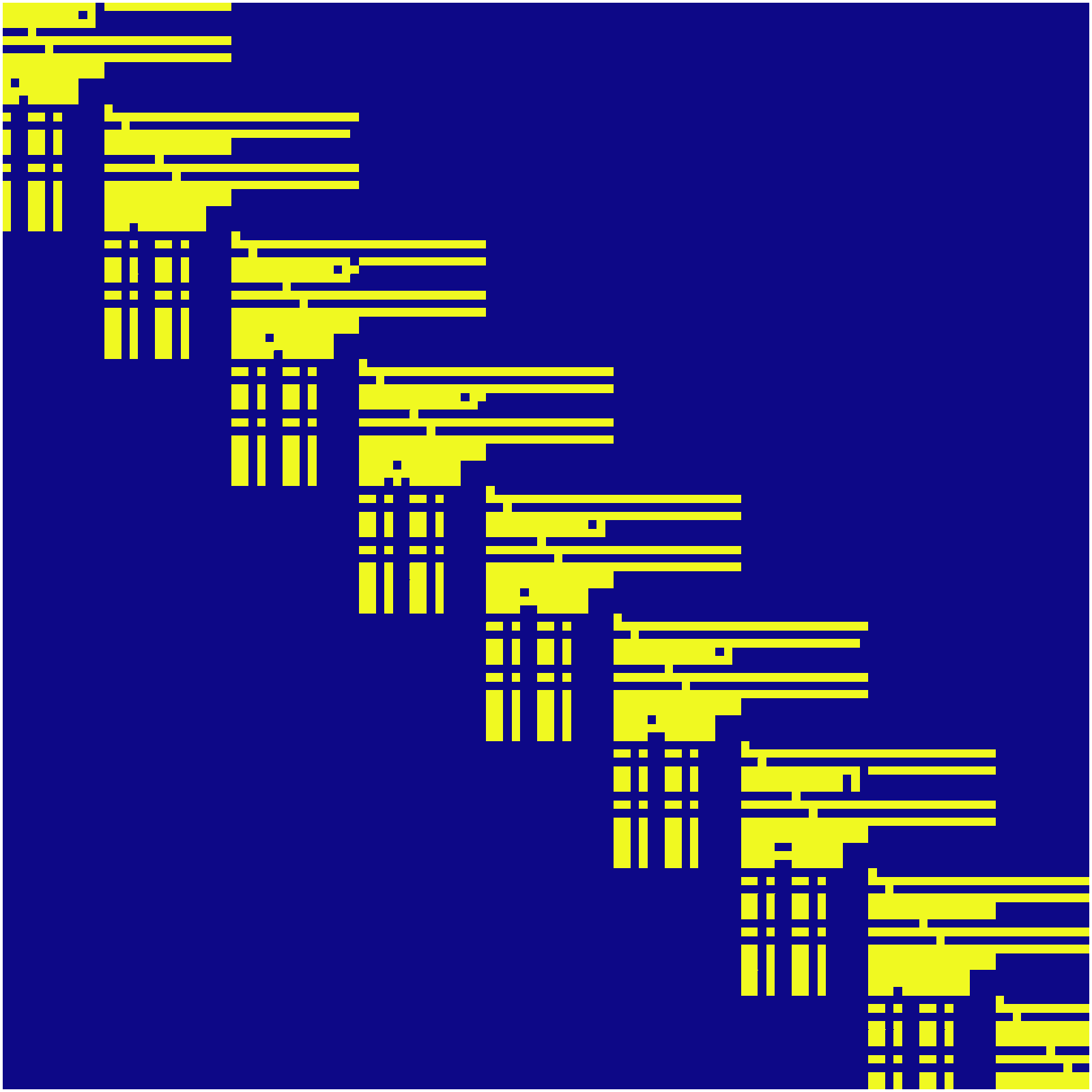}} & \makecell{ Goodwin/ \\ \quad Goodwin\_071 \\ $56021 \times 56021$ ($1797934$)} \\
\raisebox{-.5\height}{\includegraphics[width=1.7cm]{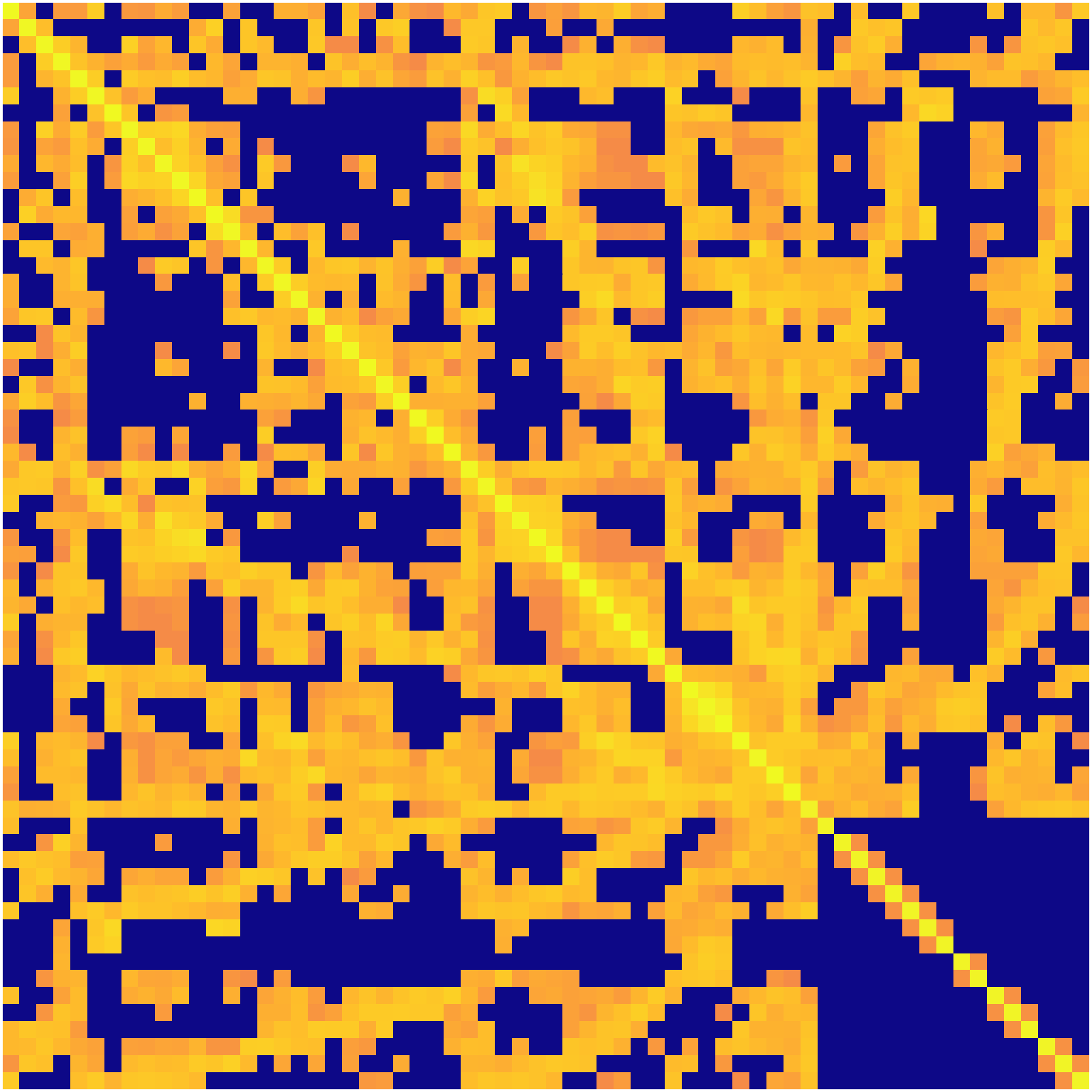}}\vspace{1pt} & \raisebox{-.5\height}{\includegraphics[width=1.7cm]{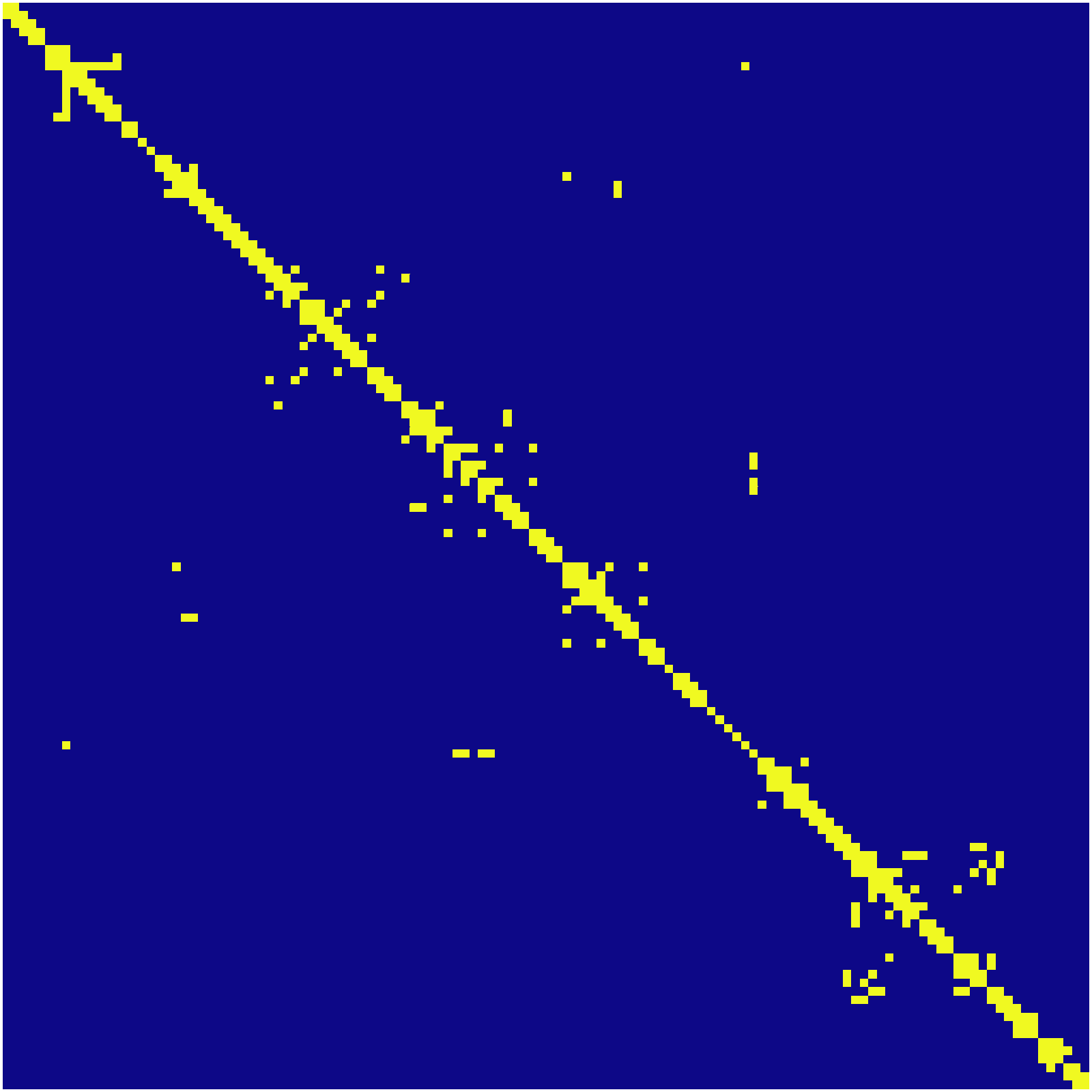}} & \makecell{ Hamm/ \\ \quad scircuit \\ $170998 \times 170998$ ($958936$)} \\
\raisebox{-.5\height}{\includegraphics[width=1.7cm]{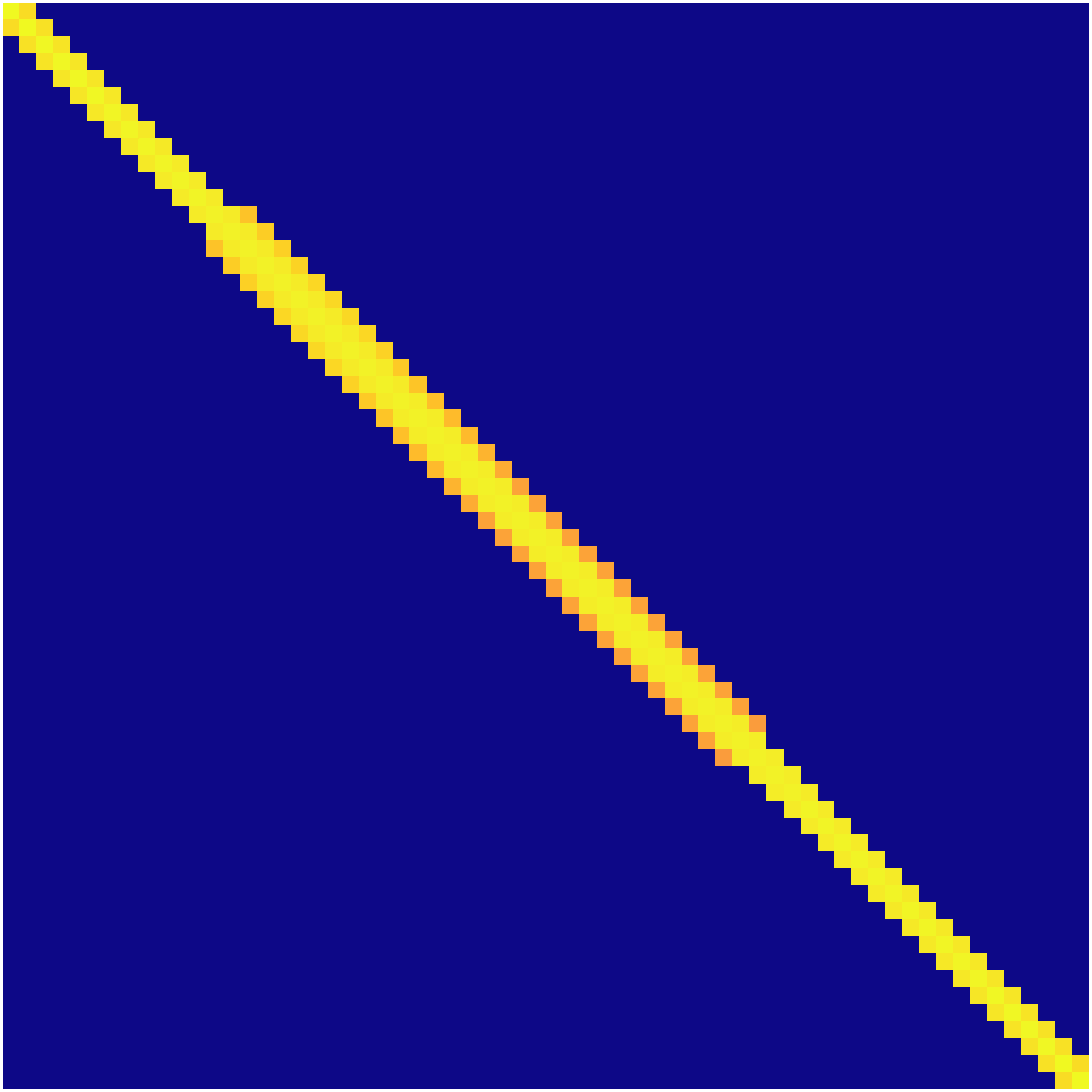}}\vspace{1pt} & \raisebox{-.5\height}{\includegraphics[width=1.7cm]{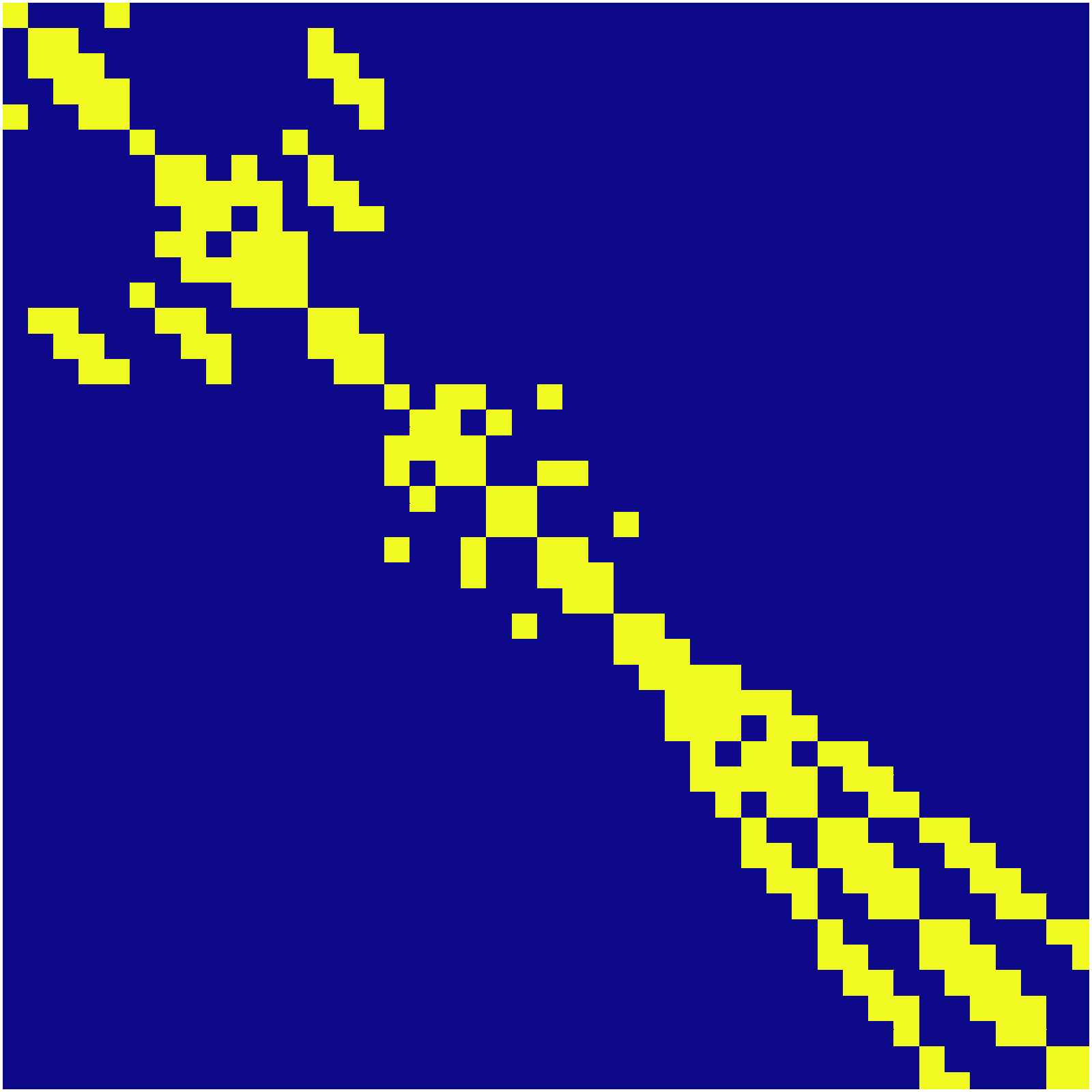}} & \makecell{ Janna/ \\ \quad Emilia\_923 \\ $923136 \times 923136$ ($41005206$)} \\
\raisebox{-.5\height}{\includegraphics[width=1.7cm]{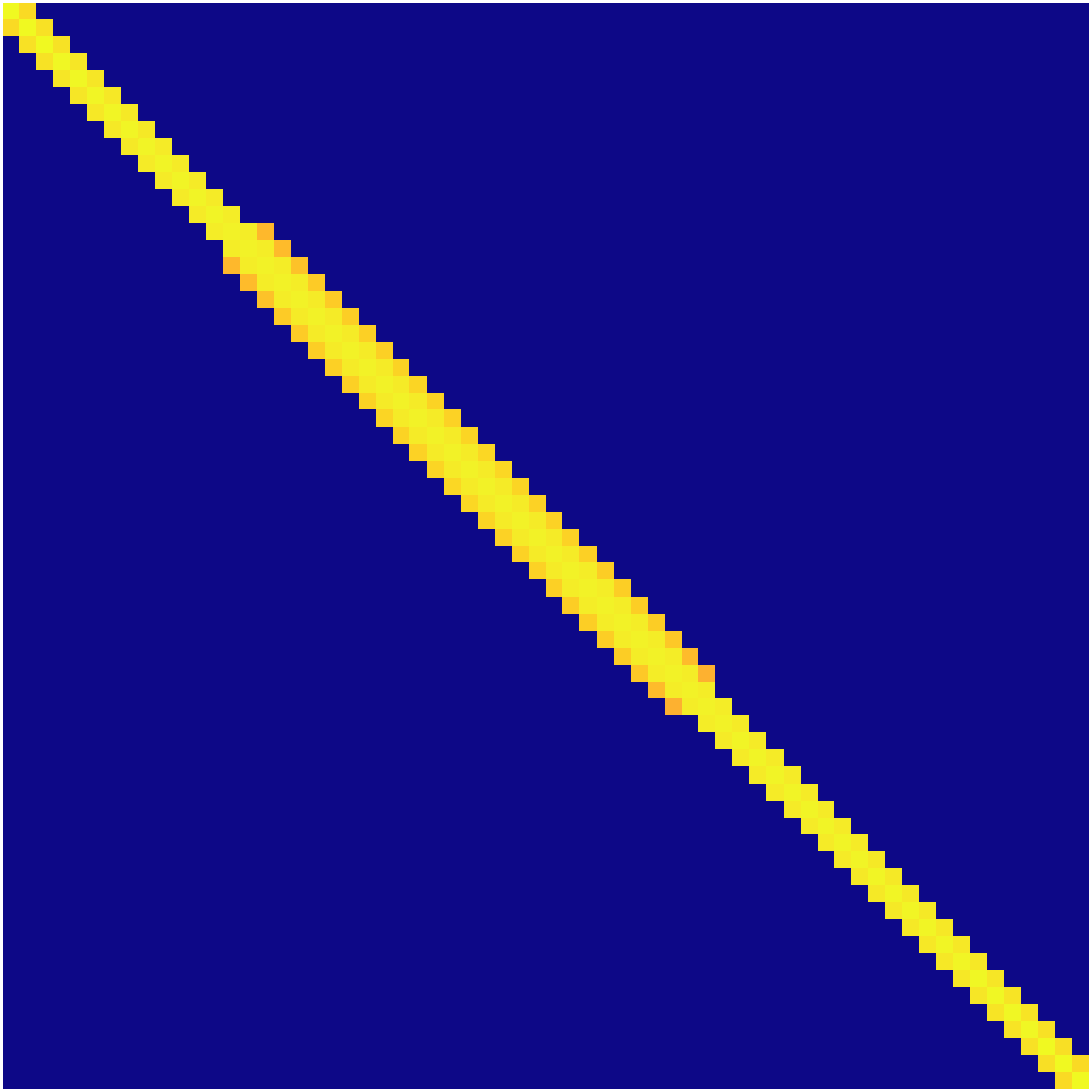}}\vspace{1pt} & \raisebox{-.5\height}{\includegraphics[width=1.7cm]{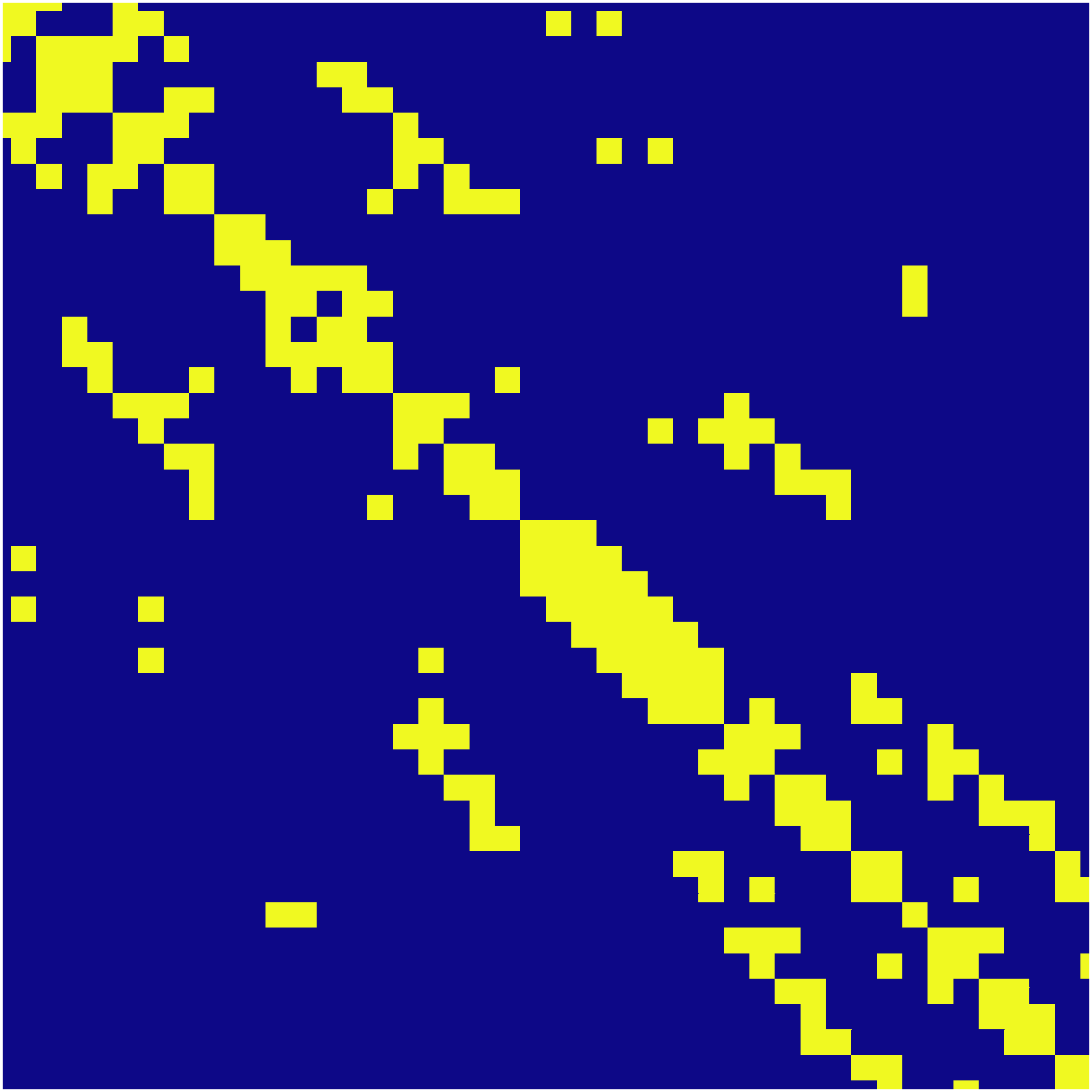}} & \makecell{ Janna/ \\ \quad Geo\_1438 \\ $1437960 \times 1437960$ ($63156690$)} \\
\raisebox{-.5\height}{\includegraphics[width=1.7cm]{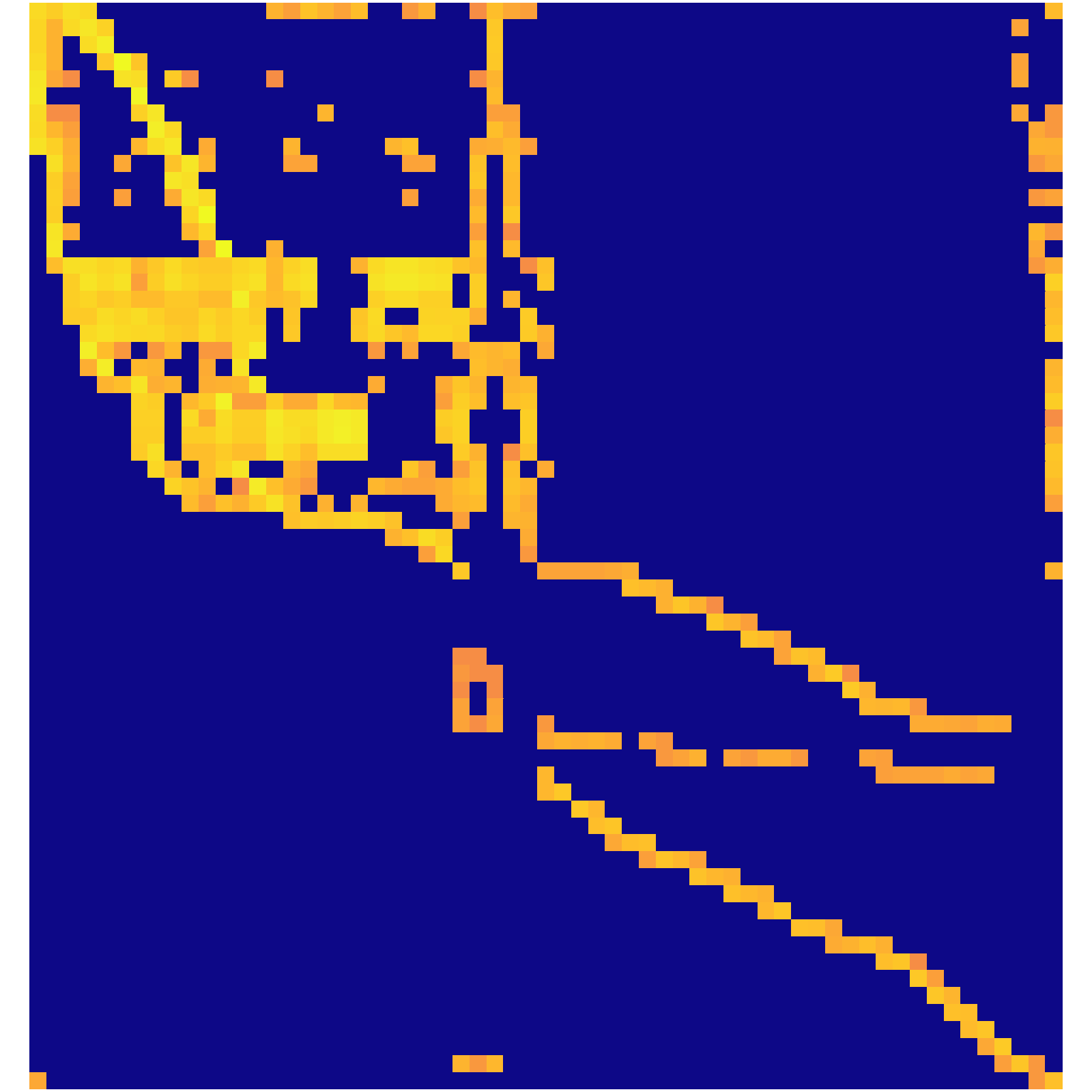}}\vspace{1pt} & \raisebox{-.5\height}{\includegraphics[width=1.7cm]{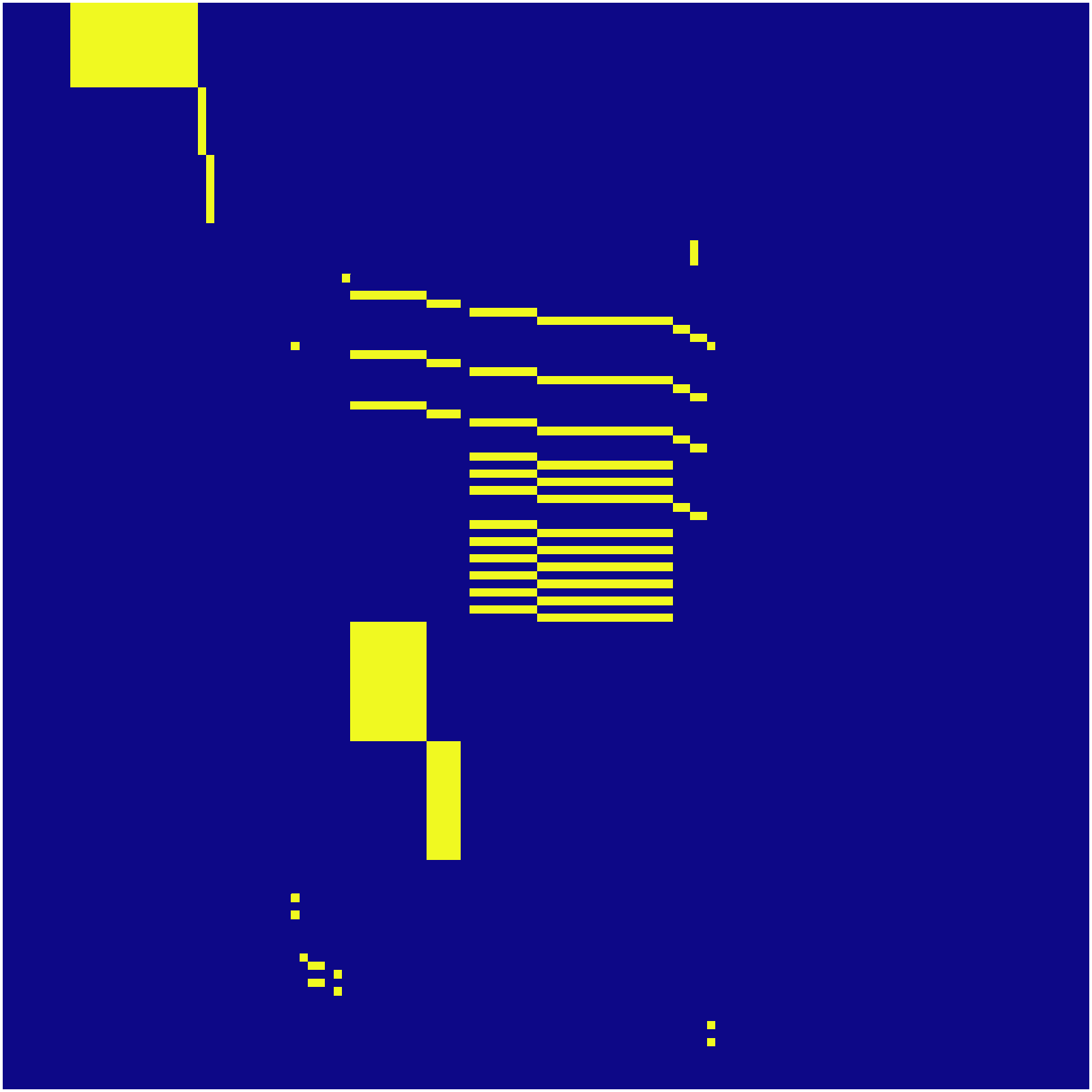}} & \makecell{ LPnetlib/ \\ \quad lpi\_gran \\ $2658 \times 2525$ ($20111$)} \\
\raisebox{-.5\height}{\includegraphics[width=1.7cm]{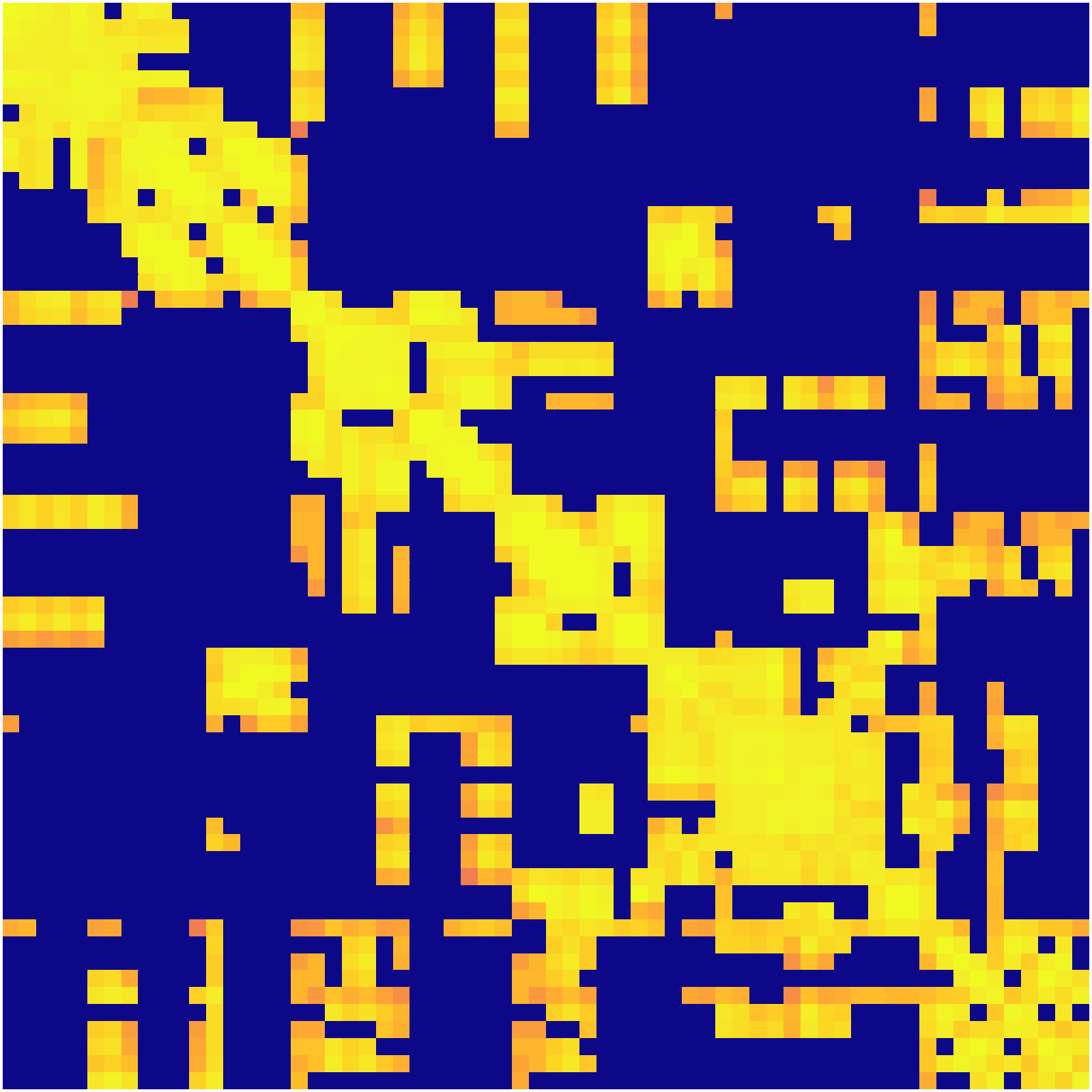}}\vspace{1pt} & \raisebox{-.5\height}{\includegraphics[width=1.7cm]{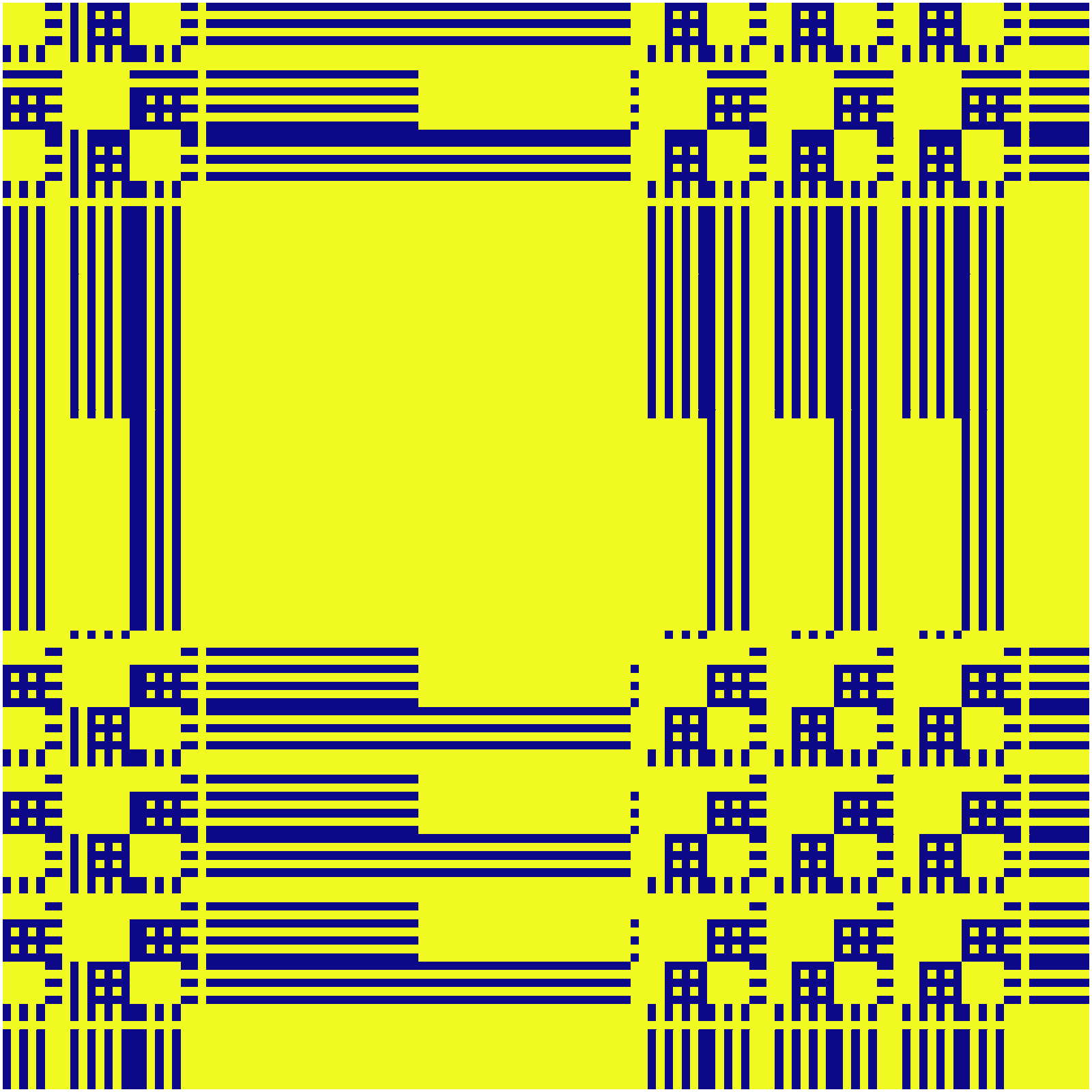}} & \makecell{ Norris/ \\ \quad heart3 \\ $2339 \times 2339$ ($682797$)} \\
\raisebox{-.5\height}{\includegraphics[width=1.7cm]{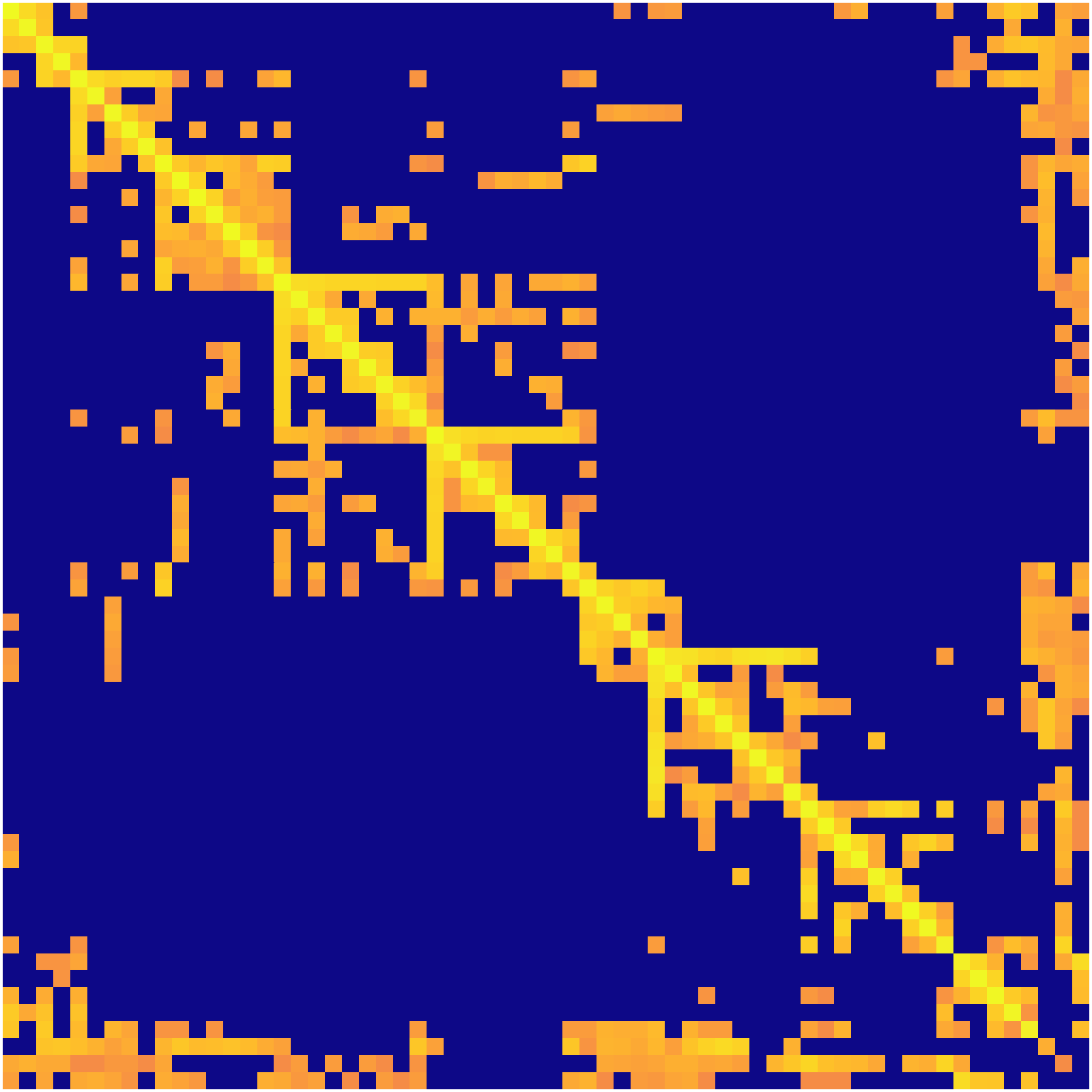}}\vspace{1pt} & \raisebox{-.5\height}{\includegraphics[width=1.7cm]{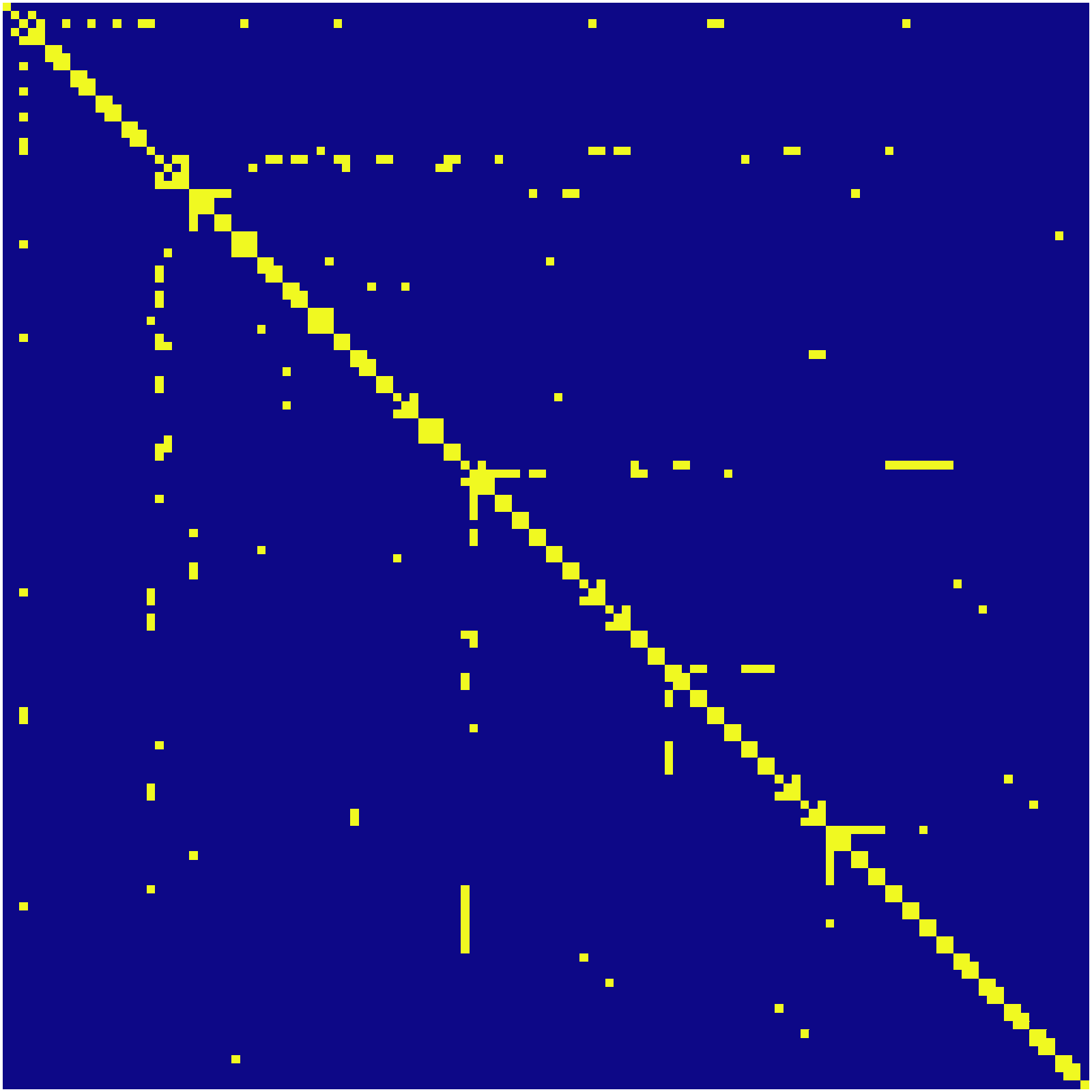}} & \makecell{ Rajat/ \\ \quad rajat26 \\ $51032 \times 51032$ ($249302$)} \\
\raisebox{-.5\height}{\includegraphics[width=1.7cm]{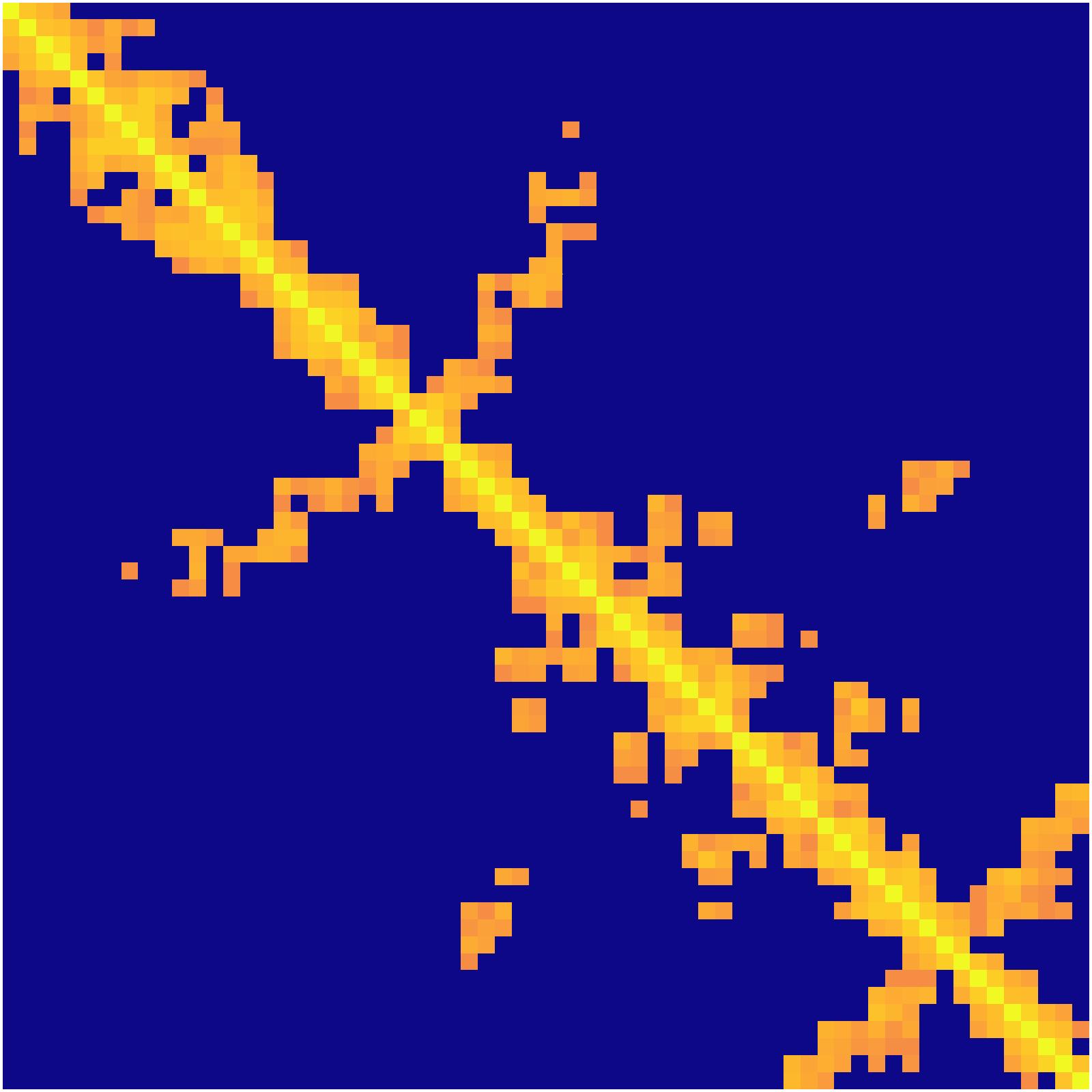}}\vspace{1pt} & \raisebox{-.5\height}{\includegraphics[width=1.7cm]{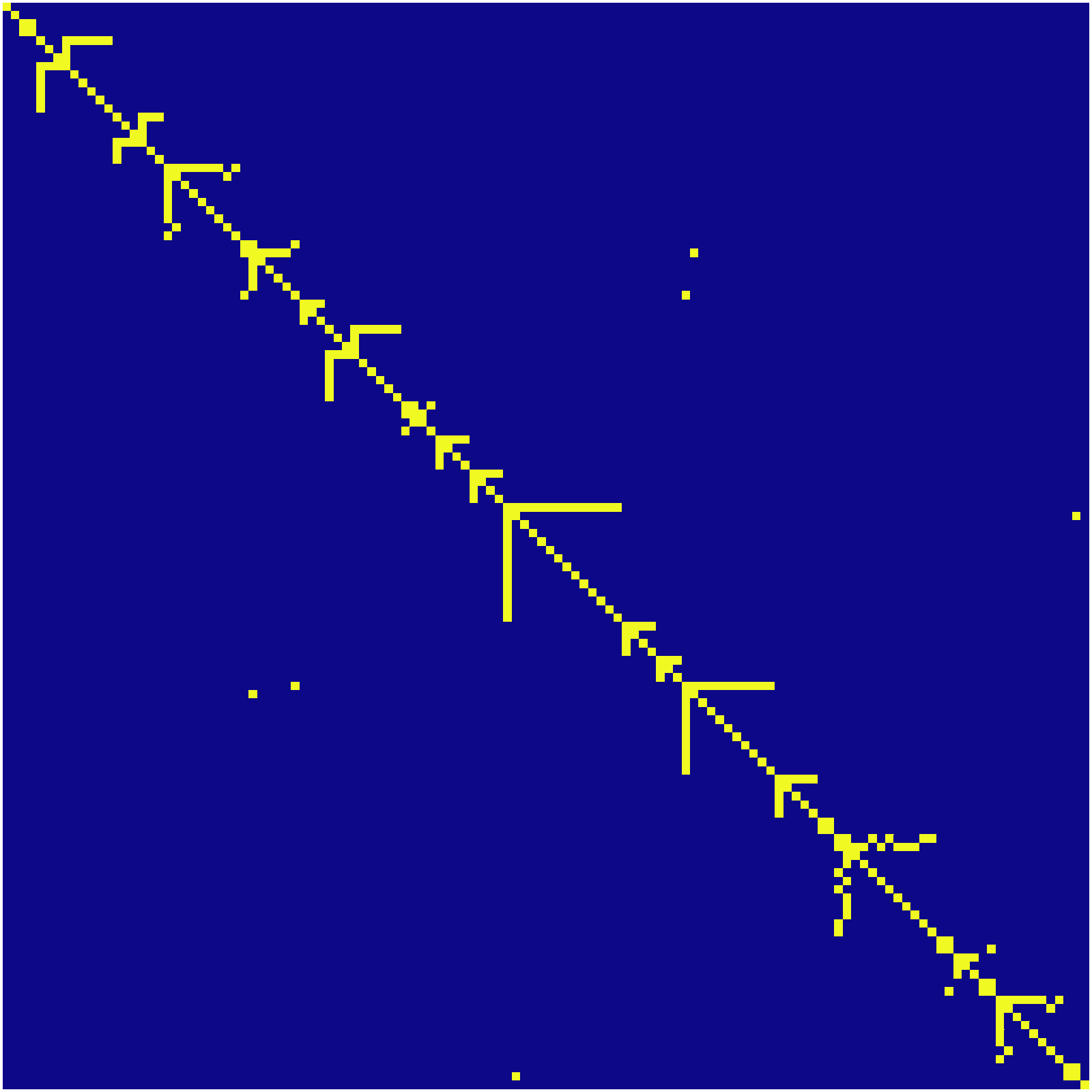}} & \makecell{ TAMU\_SmartGridCenter/ \\ \quad ACTIVSg70K \\ $69999 \times 69999$ ($238627$)} \\
\raisebox{-.5\height}{\includegraphics[width=1.7cm]{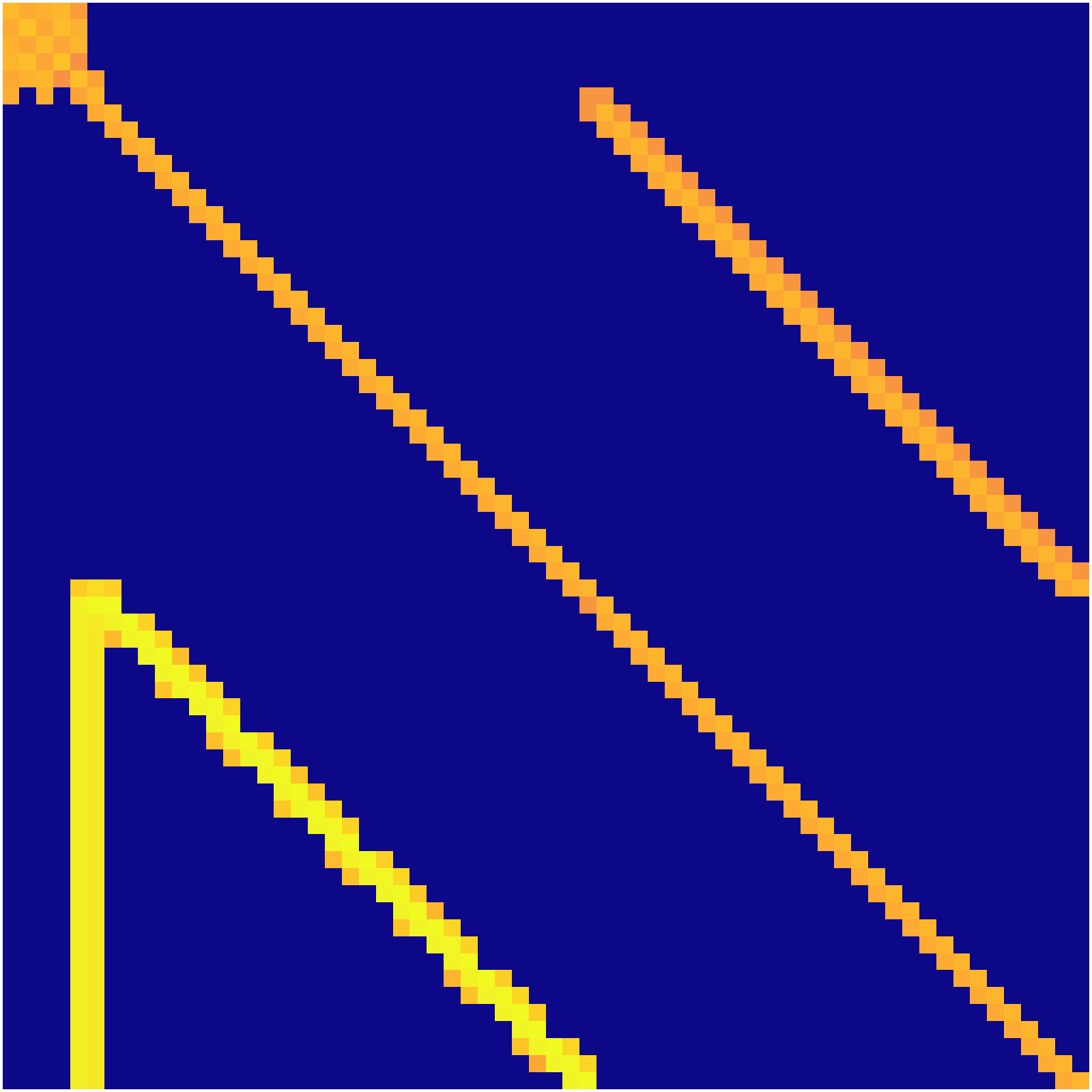}}\vspace{1pt} & \raisebox{-.5\height}{\includegraphics[width=1.7cm]{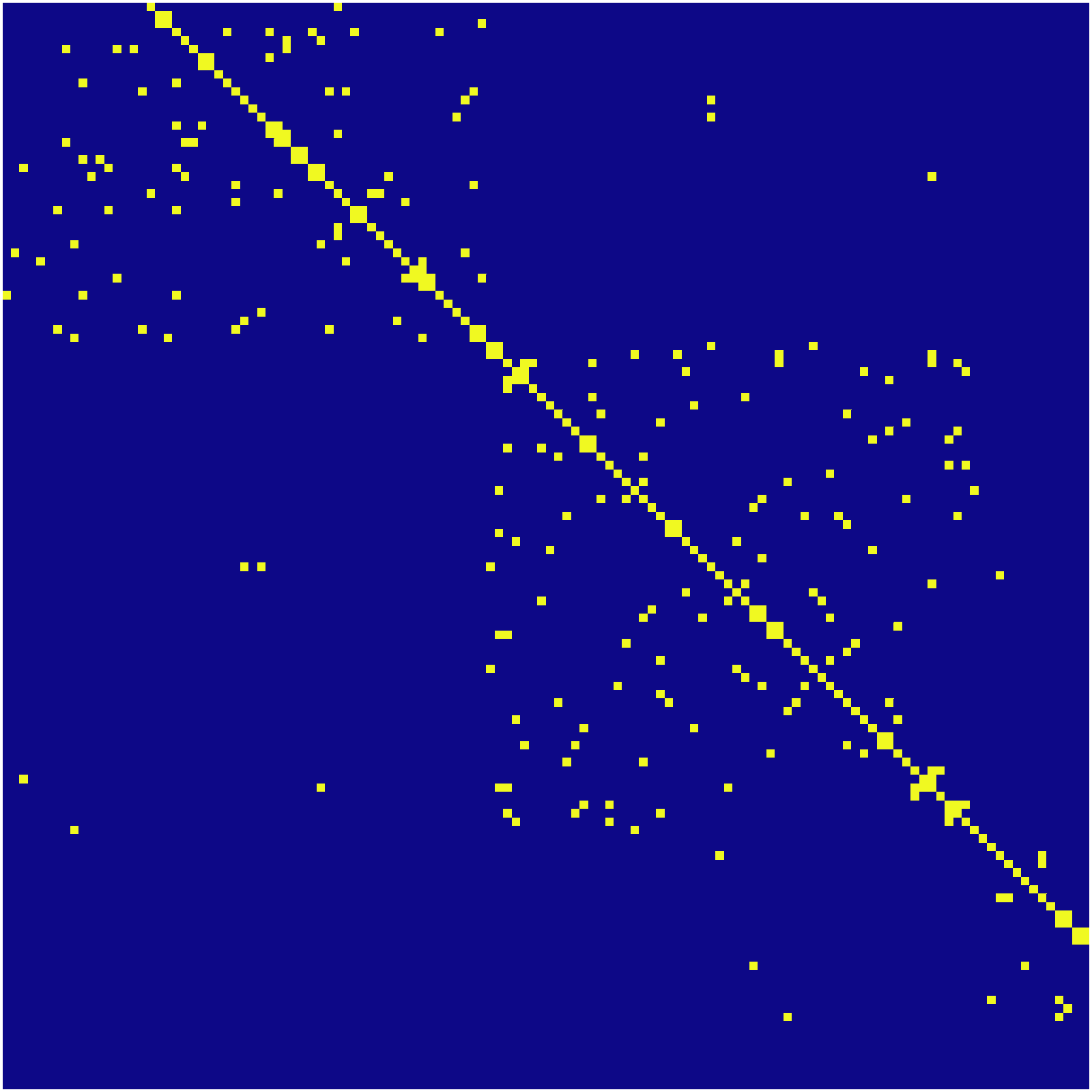}} & \makecell{ TSOPF/ \\ \quad TSOPF\_RS\_b678\_c1 \\ $18696 \times 18696$ ($4396289$)} \\
\end{tabular}

\end{table}

We compare several partitioners and matrix formats. The ``Strict'' label
refers to the \textproc{StrictPartitioner} algorithm. Note that we also use
the specialized conversion routine for the ``Strict 1D-VBR'' case. The
``Overlap ($rho = \rho$).'' label refers to the \textproc{OverlapPartitioner}
(Algorithm \ref{alg:overlappartitioner}) with a setting of $\rho$ that
worked well in our tests. Our \textproc{OptimalPartitioner} algorithm
(Algorithm \ref{alg:optimalpartitioner}) is tested with 3 different cost
models. The ``Min Memory'' label refers to minimizing the footprint of 1D-VBR
\eqref{eq:1dvbrmemory} or VBR \eqref{eq:1dvbrmemory}. The ``Min Compute''
label refers to minimizing the modeled computation time
\eqref{eq:vbrcompute}. The ``Min Blocks'' label refers to minimizing the
number of blocks \eqref{eq:numberofblocks}. When the
\textproc{OptimalPartitioner} is used to partition VBR, we partition the rows
first, then the columns, then the rows again. Further improvement after
continued alternation was observed to be negligible, suggesting either that
the initial partitioning problems are close to optimal, or that the row
partition is highly influential on the column partition, and vice-versa.

%For each partitioner, we measure the time to partition the matrix, convert
%from CSR to 1D-VBR format, and to multiply the matrix by a vector. Since our
%partitioners are intended to accelerate an SpMV kernel, they must have a
%similar runtime. We therefore normalize our runtimes to a standard, unblocked
%CSR SpMV kernel. We measure the memory usage and number of blocks in the
%resulting sparse matrix formats, normalized to the memory usage and number of
%nonzeros in the CSR format.

Since one might use our algorithms in the context of a sparse iterative
solver, where we partition once and multiply several times, using a
partitioner only produces an overall speedup after a certain number of SpMV
executions. If $t_{\text{VBR partition}}$ and $t_{\text{VBR convert}}$
are the measured times to partition and convert, and $t_{\text{VBR
multiply}}$ is the time to multiply once, then if we are to multiply $M$ times,
the total time to perform $M$ multiplications is 
$t_{\text{VBR partition}} + t_{\text{VBR convert}} + M\cdot t_{\text{VBR multiply}}$.
If the time required to multiply in CSR is $t_{\text{CSR multiply}}$, then
partitioning is the faster approach only if one plans to perform $M_{\text{critical}}$ multiplies, where
\begin{equation}\label{eq:1dvbrcriticalpoint}
    M_{\text{critical}} = \frac{t_{\text{VBR partition}} + t_{\text{VBR convert}}}{t_{\text{CSR multiply}} - t_{\text{VBR multiply}}}.
\end{equation}
We refer to $M_{\text{critical}}$ as the \textbf{critical point}.
Figure \ref{fig:profiles} shows performance profiles for all of our
partitioners on the metrics of memory usage, multiplication time, and
critical point, stratified by the floating point precision. Performance
profiles allow us to compare the relative partition quality visually over the
entire test set \cite{dolan_benchmarking_2002}. Table
\ref{tbl:results_spread} shows the distribution of running times, memory
usage, and critical points, all normalized to unpartitioned CSR, over both
precisions.

\begin{figure*}[!ht]
    \includegraphics[width=0.5\linewidth]{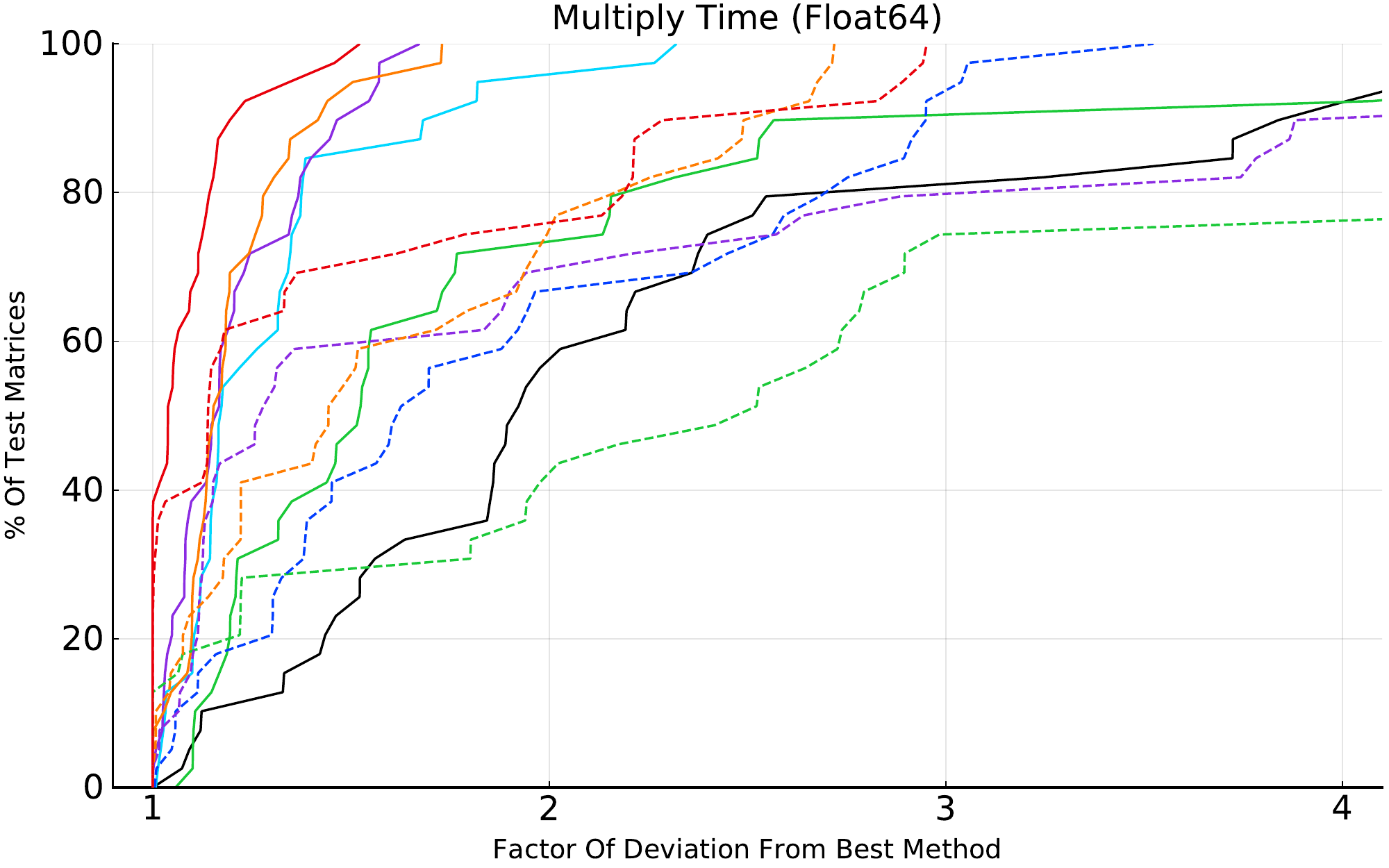}%
    \includegraphics[width=0.5\linewidth]{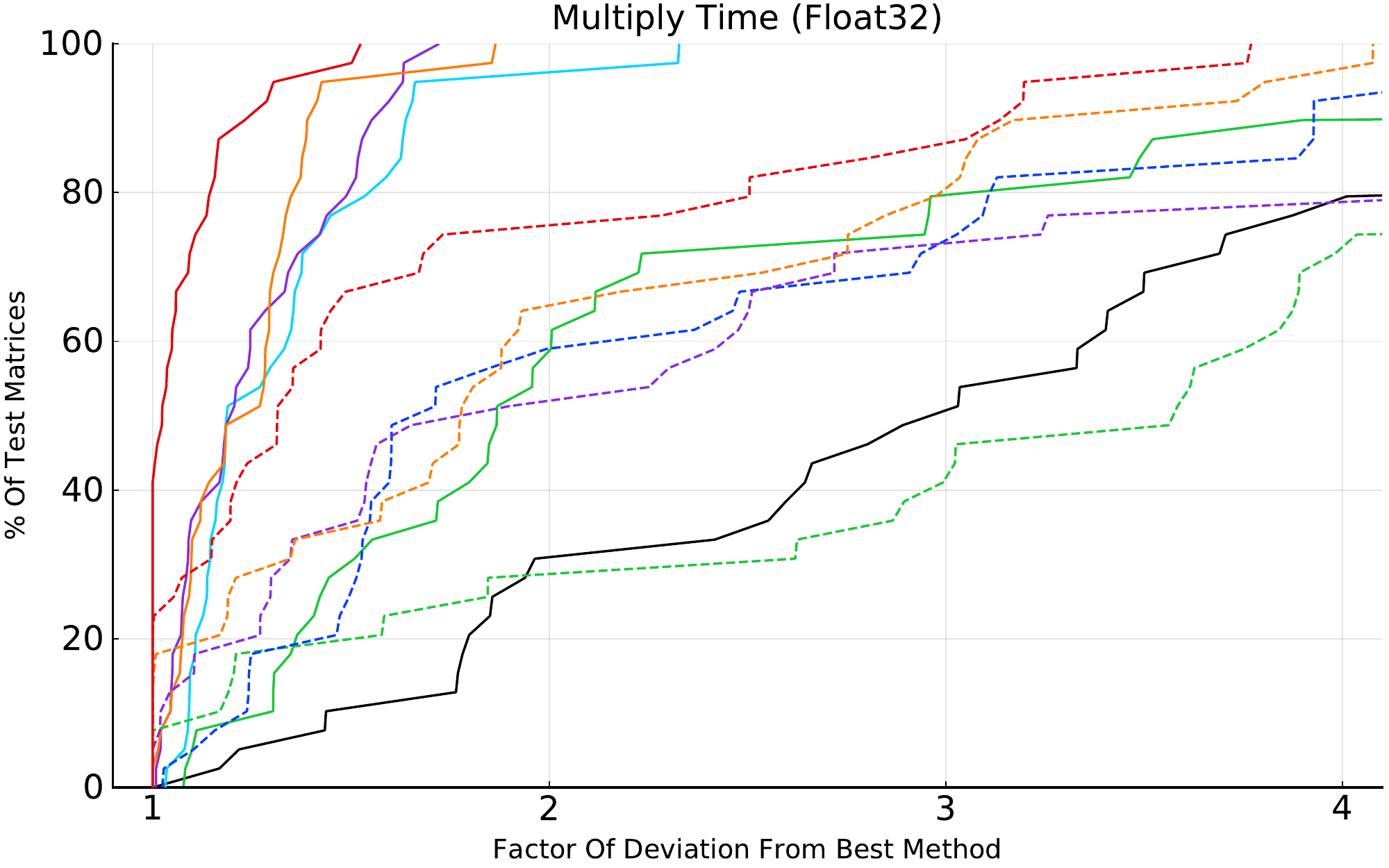}
    \includegraphics[width=0.5\linewidth]{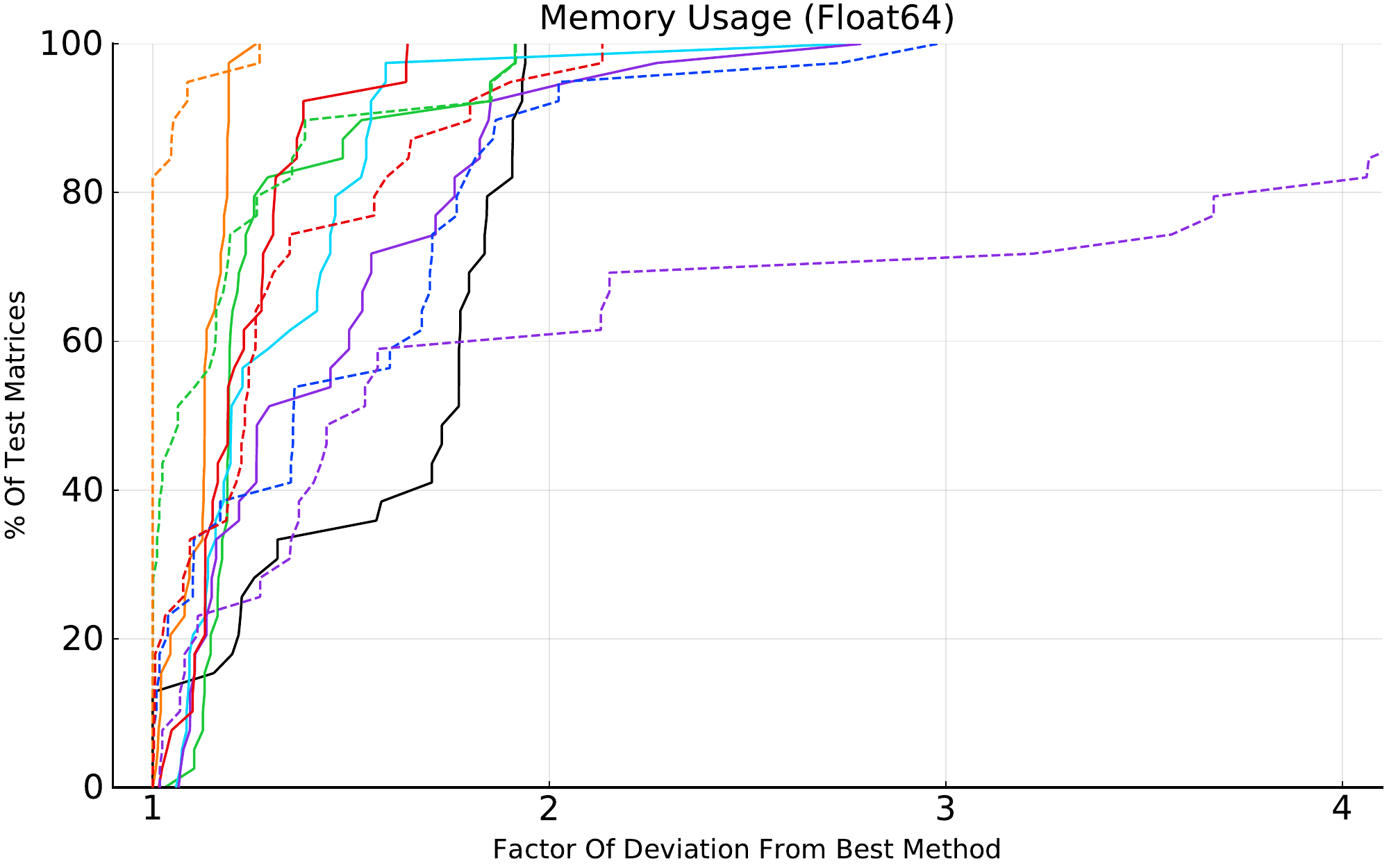}%
    \includegraphics[width=0.5\linewidth]{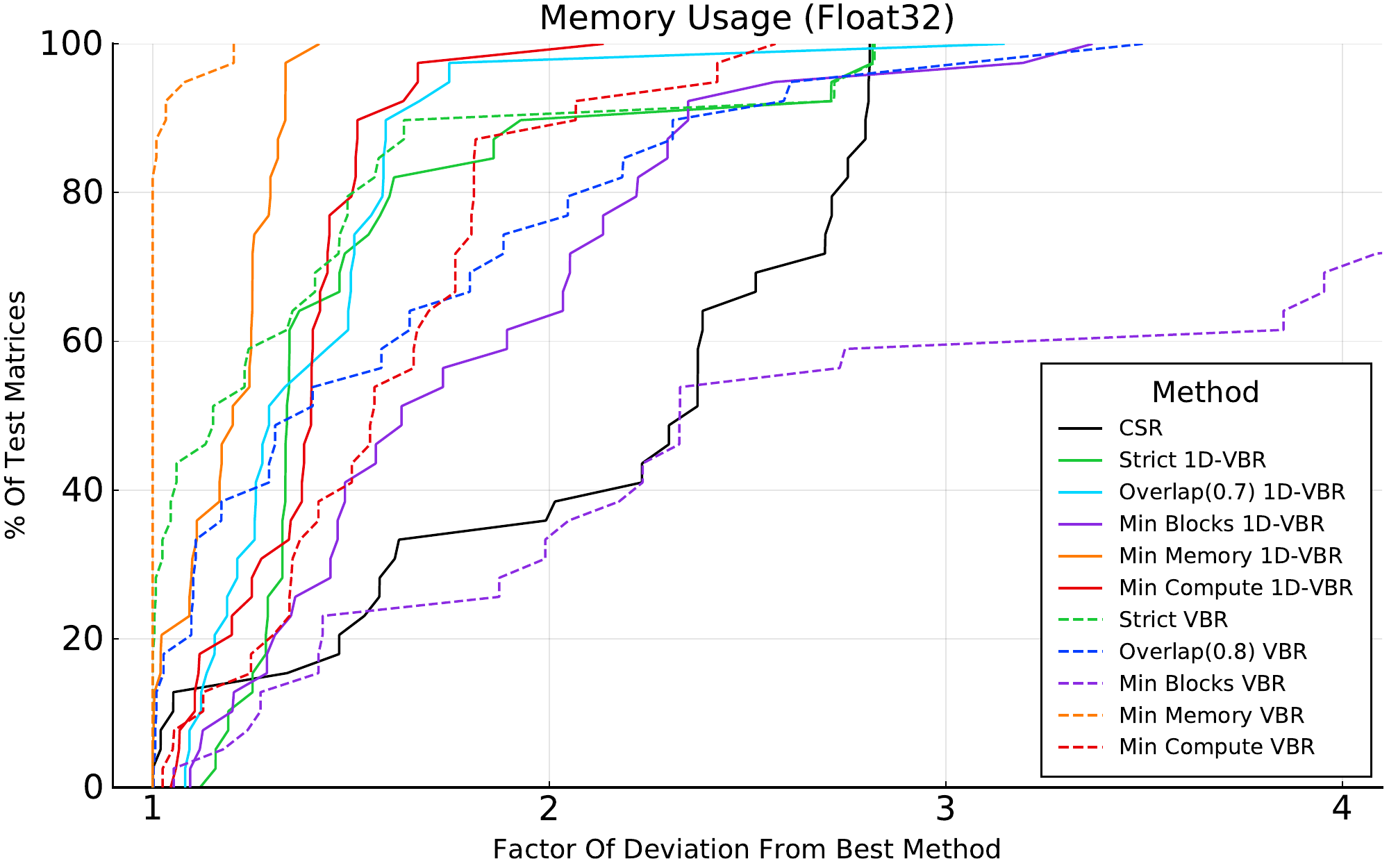}
    \caption{Performance profiles of sparse matrix-vector multiplication
    time, memory usage, and critical points for our partitioners on our
    entire test set. The $x$ axis shows the factor of deviation from the best
    performing partitioner and format, and the $y$ axis shows the percentage
    of test cases achieving such a factor. 
    } \label{fig:profiles}
\end{figure*}

\begin{table*}[htbp!]
    \centering
    \setlength{\tabcolsep}{2pt}\begin{tabular}{rp{12pt}rrrp{12pt}rrrp{12pt}rrr}
\multicolumn{1}{r}{\textbf{Partitioner}} &  & \multicolumn{3}{c}{\textbf{Memory Used}} &  & \multicolumn{3}{c}{\textbf{Multiply Time}} &  & \multicolumn{3}{c}{\textbf{Critical Point}}\\\cline{3-5}\cline{7-9}\cline{11-13}
\multicolumn{1}{r}{} &  & \multicolumn{1}{r}{\textbf{Q. 1}} & \multicolumn{1}{r}{\textbf{Med.}} & \multicolumn{1}{r}{\textbf{Q. 3}} &  & \multicolumn{1}{r}{\textbf{Q. 1}} & \multicolumn{1}{r}{\textbf{Med.}} & \multicolumn{1}{r}{\textbf{Q. 3}} &  & \multicolumn{1}{r}{\textbf{Q. 1}} & \multicolumn{1}{r}{\textbf{Med.}} & \multicolumn{1}{r}{\textbf{Q. 3}}\\
\hline
\textbf{Strict 1D-VBR} &   & $0.6$ & $0.808$ & $1$ &   & $0.573$ & $0.886$ & $1.17$ &   & $6.57$ & $46.1$ & $\infty$\\
\textbf{Overlap(0.9) 1D-VBR} &   & $0.615$ & $0.746$ & $0.918$ &   & $0.338$ & $0.545$ & $0.858$ &   & $10.9$ & $18.4$ & $68.3$\\
\textbf{Overlap(0.8) 1D-VBR} &   & $0.616$ & $0.746$ & $0.918$ &   & $0.338$ & $0.546$ & $0.902$ &   & $11.2$ & $18.4$ & $92.5$\\
\textbf{Overlap(0.7) 1D-VBR} &   & $0.619$ & $0.746$ & $0.918$ &   & $0.331$ & $0.547$ & $0.904$ &   & $10.9$ & $19$ & $69$\\
\textbf{Min Blocks 1D-VBR} &   & $0.606$ & $0.84$ & $1.32$ &   & $0.351$ & $0.493$ & $0.806$ &   & $13.9$ & $19$ & $66.3$\\
\textbf{Min Memory 1D-VBR} &   & $0.562$ & $0.633$ & $0.806$ &   & $0.347$ & $0.535$ & $0.723$ &   & $11.7$ & $17.3$ & $32.2$\\
\textbf{Min Compute 1D-VBR} &   & $0.585$ & $0.664$ & $0.857$ &   & $0.317$ & $0.45$ & $0.589$ &   & $12.4$ & $16.9$ & $29.1$\\
\textbf{Strict VBR} &   & $0.521$ & $0.696$ & $0.983$ &   & $0.543$ & $1.55$ & $2.16$ &   & $29.9$ & $\infty$ & $\infty$\\
\textbf{Overlap(0.9) VBR} &   & $0.575$ & $0.819$ & $1.18$ &   & $0.508$ & $0.769$ & $1.52$ &   & $28.8$ & $71.4$ & $\infty$\\
\textbf{Overlap(0.8) VBR} &   & $0.59$ & $0.834$ & $1.19$ &   & $0.487$ & $0.729$ & $1.56$ &   & $28.7$ & $68$ & $\infty$\\
\textbf{Overlap(0.7) VBR} &   & $0.605$ & $0.847$ & $1.2$ &   & $0.501$ & $0.777$ & $1.55$ &   & $28.5$ & $79.2$ & $\infty$\\
\textbf{Min Blocks VBR} &   & $0.67$ & $1.06$ & $2.8$ &   & $0.411$ & $0.573$ & $1.81$ &   & $29$ & $47.1$ & $\infty$\\
\textbf{Min Memory VBR} &   & $0.432$ & $0.56$ & $0.75$ &   & $0.451$ & $0.6$ & $1.31$ &   & $38.9$ & $62.1$ & $\infty$\\
\textbf{Min Compute VBR} &   & $0.556$ & $0.721$ & $1.06$ &   & $0.33$ & $0.512$ & $1.11$ &   & $36.2$ & $49.7$ & $\infty$\\
\end{tabular}

    \caption{The distribution of the normalized memory usage, multiply time, and
    critical point \eqref{eq:1dvbrcriticalpoint} over each matrix,
    transposition, and precision in our test set of matrices. A critical point
    of $\infty$ indicates that no speedup was observed. All metrics are normalized
    to unpartitioned CSR representation. The ``Q. 1'', ``Med'', and ``Q. 3'',
    columns refer to the least quartile, median, and greatest quartile,
    respectively, of the corresponding distribution. The three partitioners with
    ``Min'' objectives and the 1D-VBR format are contributions of this work.}
    \label{tbl:results_spread}
\end{table*}

Our results show that for both VBR and 1D-VBR, using Algorithm
\ref{alg:optimalpartitioner} to optimize memory usage or multiplication time
consistently produces smaller or faster VBR and 1D-VBR formats than any other
partitioner, respectively. Practitioners using variably blocked formats in
practice can expect to see improved compression or performance simply by
switching to our improved partitioning algorithm to split the matrix, though
they may need to use different cost models depending on their application.

The 1D-VBR format is the more performant choice, while VBR is capable of better
compression.  Using heuristics to partition for 1D-VBR often resulted in faster
matrix multiplication than the best VBR partitions. The superior performance of
1D-VBR is due to implementation differences between the two formats. All blocks
in a 1D-VBR block row are the same size, but each VBR block requires a random
access to the column partition $\Phi$, and a conditional jump to the appropriate
block width. Our memory minimizing partitioner was unmatched by any other
partitioner for compressing either format, but using this partitioner on VBR
format usually yielded the highest compression.

While the heuristics produced quality partitions on some matrices, minimizing
our cost models was effective on the entire test set, evinced by the superior
third quartile of memory usage and multiplication time for the ``Min Memory''
and ``Min Compute'' partitioners. A disadvantage of the overlap heuristic is
the need to set the $\rho$ parameter. We plotted values which worked well in
our tests, but it may be necessary to test several settings to find the best
one, an expensive process \cite{vuduc_fast_2005-1}.

When the precision is reduced, the reference CSR implementation deviates more
from the best (blocked) method. Thus, we can say that the benefits of
blocking are amplified when the precision is halved. Performance improves
because our vector units fit twice the elements, and compression improves
because the 64-bit integer indices become twice as large as the individual
elements, so there is more incentive to merge blocks. We expect these trends
to continue for further reduced precisions, especially when the
index precision remains the same.

Our dynamic programming algorithm was similarly efficient to other
heuristics. The distribution of critical points shows that the benefits of
partitioning outweighed the runtime of partitioning and conversion within a
similar number of multiplications. However, the simplicity of the ``Strict''
heuristic and its specialized conversion to 1D-VBR format made it a practical
choice in cases where the matrix had clear blocks. The critical points also
show the added partitioning and conversion costs and reduced performance of
VBR format, only breaking even after triple the multiplications.

\section{Conclusions}

We present an algorithm for optimally partitioning rows in a VBR matrix when
the columns are fixed. We apply the algorithm to a novel specialization of
VBR where only rows are blocked, 1D-VBR. Our algorithm effectively optimizes
partitions under a diverse family of cost models. We show that minimizing an
empirical cost model for SpMV runtime in 1D-VBR format yields the best
performance, and minimizing a cost model for memory consumption in VBR format
yields the best compression. The benefits of blocking are amplified as the
precision is reduced.

Existing algorithms using variably blocked formats
stand to benefit from employing these techniques to pick better partitions,
without any changes to the algorithms themselves. Since simultaneous
partitioning of rows and columns for VBR format is NP-hard and 1D-VBR is
faster and supports diverse cost models, practitioners using two-dimensional
aligned block algorithms should consider one-dimensional reformulations when
possible.

We hope to see our techniques applied to new cost models and contiguous
partitioning problems, such as block decompositions with sparse blocks
\cite{buluc_representation_2008, buluc_parallel_2009, hong_adaptive_2019,
namashivayam_variable-sized_2021}.

\newpage
\bibliographystyle{unsrt}
\bibliography{VBR}

\begin{thebibliography}{10}

\bibitem{saad_iterative_2003}
Yousef Saad.
\newblock {\em Iterative methods for sparse linear systems}.
\newblock SIAM, Philadelphia, 2nd edition, 2003.

\bibitem{vuduc_oski:_2005}
Richard Vuduc, James~W Demmel, and Katherine~A Yelick.
\newblock {OSKI}: {A} library of automatically tuned sparse matrix kernels.
\newblock {\em Journal of Physics: Conference Series}, 16:521--530, January
  2005.

\bibitem{davis_survey_2016}
Timothy~A. Davis, Sivasankaran Rajamanickam, and Wissam~M. Sid-Lakhdar.
\newblock A survey of direct methods for sparse linear systems.
\newblock {\em Acta Numerica}, 25:383--566, May 2016.
\newblock Publisher: Cambridge University Press.

\bibitem{demmel_supernodal_1999}
James~W. Demmel, Stanley~C. Eisenstat, John~R. Gilbert, Xiaoye~S. Li, and
  Joseph W.~H. Liu.
\newblock A {Supernodal} {Approach} to {Sparse} {Partial} {Pivoting}.
\newblock {\em SIAM Journal on Matrix Analysis and Applications},
  20(3):720--755, January 1999.
\newblock Publisher: Society for Industrial and Applied Mathematics.

\bibitem{saad_finding_2003}
Yousef Saad.
\newblock Finding {Exact} and {Approximate} {Block} {Structures} for {ILU}
  {Preconditioning}.
\newblock {\em SIAM Journal on Scientific Computing}, 24(4):1107--1123, January
  2003.

\bibitem{yamazaki_performance_2020}
Ichitaro Yamazaki, Sivasankaran Rajamanickam, and Nathan Ellingwood.
\newblock Performance {Portable} {Supernode}-based {Sparse} {Triangular}
  {Solver} for {Manycore} {Architectures}.
\newblock In {\em 49th {International} {Conference} on {Parallel} {Processing}
  - {ICPP}}, {ICPP} '20, pages 1--11, New York, NY, USA, August 2020.
  Association for Computing Machinery.

\bibitem{ashcraft_influence_1989}
Cleve Ashcraft and Roger Grimes.
\newblock The influence of relaxed supernode partitions on the multifrontal
  method.
\newblock {\em ACM Transactions on Mathematical Software}, 15(4):291--309,
  December 1989.

\bibitem{kim_task_2016}
Kyungjoo Kim, Sivasankaran Rajamanickam, George~Widgery Stelle, Harold~C.
  Edwards, and Stephen~Lecler Olivier.
\newblock Task {Parallel} {Incomplete} {Cholesky} {Factorization} using {2D}
  {Partitioned}-{Block} {Layout}.
\newblock Technical Report SAND-2016-0637R, Sandia National Lab. (SNL-NM),
  Albuquerque, NM (United States), January 2016.

\bibitem{vuduc_fast_2005}
Richard~W. Vuduc and Hyun-Jin Moon.
\newblock Fast {Sparse} {Matrix}-{Vector} {Multiplication} by {Exploiting}
  {Variable} {Block} {Structure}.
\newblock In Laurence~T. Yang, Omer~F. Rana, Beniamino Di~Martino, and Jack
  Dongarra, editors, {\em High {Performance} {Computing} and {Communications}},
  Lecture {Notes} in {Computer} {Science}, pages 807--816, Berlin, Heidelberg,
  2005. Springer.

\bibitem{karakasis_perfomance_2009}
V.~Karakasis, G.~Goumas, and N.~Koziris.
\newblock Perfomance {Models} for {Blocked} {Sparse} {Matrix}-{Vector}
  {Multiplication} {Kernels}.
\newblock In {\em 2009 {International} {Conference} on {Parallel}
  {Processing}}, pages 356--364, September 2009.

\bibitem{karakasis_comparative_2009}
Vasileios Karakasis, Georgios Goumas, and Nectarios Koziris.
\newblock A {Comparative} {Study} of {Blocking} {Storage} {Methods} for
  {Sparse} {Matrices} on {Multicore} {Architectures}.
\newblock In {\em 2009 {International} {Conference} on {Computational}
  {Science} and {Engineering}}, pages 247--256, Vancouver, BC, Canada, 2009.
  IEEE.

\bibitem{saad_sparskit_1990}
Youcef Saad.
\newblock {SPARSKIT}: {A} basic tool kit for sparse matrix computations.
\newblock Technical report, May 1990.

\bibitem{saad_sparskit_1994}
Youcef Saad.
\newblock {\em {SPARSKIT}: a basic tool kit for sparse matrix computations -
  {Version} 2}.
\newblock 1994.

\bibitem{remington_nist_1996}
Karin Remington and Roldan Pozo.
\newblock {NIST} {Sparse} {BLAS}: {User}{\textquoteright}s {Guide}.
\newblock Technical report, NIST, 1996.

\bibitem{vuduc_automatic_2004}
Richard~W. Vuduc.
\newblock {\em Automatic performance tuning of sparse matrix kernels}.
\newblock PhD thesis, University of California, Berkeley, CA, USA, January
  2004.

\bibitem{noauthor_developer_2020}
Developer {Reference} for {Intel}{\textregistered} {Math} {Kernel} {Library} -
  {Fortran}.
\newblock Technical Report 097, Intel{\textregistered}, 2020.

\bibitem{shantharam_exploiting_2011}
Manu Shantharam, Anirban Chatterjee, and Padma Raghavan.
\newblock Exploiting dense substructures for fast sparse matrix vector
  multiplication.
\newblock {\em The International Journal of High Performance Computing
  Applications}, 25(3):328--341, August 2011.

\bibitem{pinar_improving_1999}
Ali Pinar and Michael~T. Heath.
\newblock Improving {Performance} of {Sparse} {Matrix}-vector {Multiplication}.
\newblock In {\em Proceedings of the 1999 {ACM}/{IEEE} {Conference} on
  {Supercomputing}}, {SC} '99, New York, NY, USA, 1999. ACM.
\newblock event-place: Portland, Oregon, USA.

\bibitem{karp_reducibility_1972}
Richard~M. Karp.
\newblock Reducibility among {Combinatorial} {Problems}.
\newblock In Raymond~E. Miller, James~W. Thatcher, and Jean~D. Bohlinger,
  editors, {\em Complexity of {Computer} {Computations}: {Proceedings} of a
  symposium on the {Complexity} of {Computer} {Computations}, held {March}
  20{\textendash}22, 1972, at the {IBM} {Thomas} {J}. {Watson} {Research}
  {Center}, {Yorktown} {Heights}, {New} {York}, and sponsored by the {Office}
  of {Naval} {Research}, {Mathematics} {Program}, {IBM} {World} {Trade}
  {Corporation}, and the {IBM} {Research} {Mathematical} {Sciences}
  {Department}}, The {IBM} {Research} {Symposia} {Series}, pages 85--103.
  Springer US, Boston, MA, 1972.

\bibitem{papadimitriou_optimization_1991}
Christos~H. Papadimitriou and Mihalis Yannakakis.
\newblock Optimization, approximation, and complexity classes.
\newblock {\em Journal of Computer and System Sciences}, 43(3):425--440,
  December 1991.

\bibitem{buttari_performance_2007}
Alfredo Buttari, Victor Eijkhout, Julien Langou, and Salvatore Filippone.
\newblock Performance {Optimization} and {Modeling} of {Blocked} {Sparse}
  {Kernels}.
\newblock {\em The International Journal of High Performance Computing
  Applications}, 21(4):467--484, November 2007.

\bibitem{grandjean_optimal_2012}
Anael Grandjean, Johannes Langguth, and Bora U{\c c}ar.
\newblock On {Optimal} and {Balanced} {Sparse} {Matrix} {Partitioning}
  {Problems}.
\newblock In {\em 2012 {IEEE} {International} {Conference} on {Cluster}
  {Computing}}, pages 257--265, September 2012.
\newblock ISSN: 2168-9253.

\bibitem{jackson_algorithm_2005}
Brad Jackson, Jeffrey~D. Scargle, David Barnes, Sundararajan Arabhi, Alina Alt,
  Peter Gioumousis, Elyus Gwin, Paungkaew Sangtrakulcharoen, Linda Tan, and
  Tun~Tao Tsai.
\newblock An {Algorithm} for {Optimal} {Partitioning} of {Data} on an
  {Interval}.
\newblock {\em IEEE Signal Processing Letters}, 12(2):105--108, February 2005.
\newblock arXiv: math/0309285.

\bibitem{alpert_multiway_1995}
C.J. Alpert and A.B. Kahng.
\newblock Multiway partitioning via geometric embeddings, orderings, and
  dynamic programming.
\newblock {\em IEEE Transactions on Computer-Aided Design of Integrated
  Circuits and Systems}, 14(11):1342--1358, November 1995.

\bibitem{ziantz_run-time_1994}
Louis~H. Ziantz, Can~C. {\"O}zturan, and Boleslaw~K. Szymanski.
\newblock Run-time optimization of sparse matrix-vector multiplication on
  {SIMD} machines.
\newblock In Costas Halatsis, Dimitrios Maritsas, George Philokyprou, and
  Sergios Theodoridis, editors, {\em {PARLE}'94 {Parallel} {Architectures} and
  {Languages} {Europe}}, Lecture {Notes} in {Computer} {Science}, pages
  313--322, Berlin, Heidelberg, 1994. Springer.

\bibitem{kernighan_optimal_1971}
Brian~W. Kernighan.
\newblock Optimal {Sequential} {Partitions} of {Graphs}.
\newblock {\em Journal of the ACM (JACM)}, 18(1):34--40, January 1971.

\bibitem{kolda_partitioning_1998}
Tamara~G. Kolda.
\newblock Partitioning sparse rectangular matrices for parallel processing.
\newblock In Alfonso Ferreira, Jos{\'e} Rolim, Horst Simon, and Shang-Hua Teng,
  editors, {\em Solving {Irregularly} {Structured} {Problems} in {Parallel}},
  Lecture {Notes} in {Computer} {Science}, pages 68--79, Berlin, Heidelberg,
  1998. Springer.

\bibitem{hendrickson_graph_2000}
Bruce Hendrickson and Tamara~G. Kolda.
\newblock Graph {Partitioning} {Models} for {Parallel} {Computing}.
\newblock {\em Parallel Comput.}, 26(12):1519--1534, November 2000.

\bibitem{yasar_heuristics_2019}
Abdurrahman Ya{\c s}ar and {\"U}mit~V. {\c C}ataly{\"u}rek.
\newblock Heuristics for {Symmetric} {Rectilinear} {Matrix} {Partitioning}.
\newblock {\em arXiv:1909.12209 [cs]}, September 2019.
\newblock arXiv: 1909.12209.

\bibitem{vuduc_performance_2002}
R.~Vuduc, J.W. Demmel, K.A. Yelick, S.~Kamil, R.~Nishtala, and B.~Lee.
\newblock Performance {Optimizations} and {Bounds} for {Sparse}
  {Matrix}-{Vector} {Multiply}.
\newblock pages 26--26. IEEE, 2002.

\bibitem{razzaq_dynb_2017}
Javed Razzaq, Rudolf Berrendorf, Jan~P. Ecker, Soenke Hack, Max Weierstall, and
  Florian Manuss.
\newblock The {DynB} {Sparse} {Matrix} {Format} {Using} {Variable} {Sized} {2D}
  {Blocks} for {Efficient} {Sparse} {Matrix} {Vector} {Multiplications} with
  {General} {Matrix} {Structures}.
\newblock {\em International Journal On Advances in Intelligent Systems}, 10(1
  and 2):48--58, June 2017.

\bibitem{im_optimizing_2000}
Eun-Jin Im.
\newblock {\em Optimizing the {Performance} of {Sparse} {Matrix}-{Vector}
  {Multiplication}}.
\newblock PhD thesis, EECS Department, University of California, Berkeley, June
  2000.

\bibitem{im_optimizing_2001}
Eun-Jin Im and Katherine Yelick.
\newblock Optimizing {Sparse} {Matrix} {Computations} for {Register} {Reuse} in
  {SPARSITY}.
\newblock In {\em Computational {Science} {\textemdash} {ICCS} 2001}, Lecture
  {Notes} in {Computer} {Science}, pages 127--136. Springer, Berlin,
  Heidelberg, May 2001.

\bibitem{im_sparsity:_2004}
Eun-Jin Im, Katherine Yelick, and Richard Vuduc.
\newblock Sparsity: {Optimization} {Framework} for {Sparse} {Matrix} {Kernels}.
\newblock {\em International Journal of High Performance Computing
  Applications}, 18(1):135--158, February 2004.

\bibitem{eberhardt_optimization_2016}
Ryan Eberhardt and Mark Hoemmen.
\newblock Optimization of {Block} {Sparse} {Matrix}-{Vector} {Multiplication}
  on {Shared}-{Memory} {Parallel} {Architectures}.
\newblock In {\em 2016 {IEEE} {International} {Parallel} and {Distributed}
  {Processing} {Symposium} {Workshops} ({IPDPSW})}, pages 663--672, May 2016.

\bibitem{choi_model-driven_2010}
Jee~W. Choi, Amik Singh, and Richard~W. Vuduc.
\newblock Model-driven autotuning of sparse matrix-vector multiply on {GPUs}.
\newblock {\em ACM SIGPLAN Notices}, 45(5):115, May 2010.

\bibitem{lee_performance_2004}
B.C. Lee, R.W. Vuduc, J.W. Demmel, and K.A. Yelick.
\newblock Performance models for evaluation and automatic tuning of symmetric
  sparse matrix-vector multiply.
\newblock pages 169--176 vol.1. IEEE, 2004.

\bibitem{ahrens_fill_2018}
P.~Ahrens, H.~Xu, and N.~Schiefer.
\newblock A {Fill} {Estimation} {Algorithm} for {Sparse} {Matrices} and
  {Tensors} in {Blocked} {Formats}.
\newblock In {\em 2018 {IEEE} {International} {Parallel} and {Distributed}
  {Processing} {Symposium} ({IPDPS})}, pages 546--556, May 2018.

\bibitem{xu_fill_2018}
Helen Xu.
\newblock {\em Fill {Estimation} for {Blocked} {Sparse} {Matrices} and
  {Tensors}}.
\newblock PhD thesis, Department of Electrical Engineering and Computer
  Science, Massachusetts Institute of Technology, June 2018.

\bibitem{ahrens_parallel_2019}
Peter Ahrens.
\newblock {\em A {Parallel} {Fill} {Estimation} {Algorithm} for {Sparse}
  {Matrices} and {Tensors} in {Blocked} {Formats}}.
\newblock Thesis, Massachusetts Institute of Technology, 2019.

\bibitem{vassilevska_finding_2004}
Virginia Vassilevska and Ali Pinar.
\newblock Finding {Nonoverlapping} {Dense} {Blocks} of a {Sparse} {Matrix}.
\newblock February 2004.

\bibitem{chen_efficient_2018}
Xinhai Chen, Peizhen Xie, Lihua Chi, Jie Liu, and Chunye Gong.
\newblock An efficient {SIMD} compression format for sparse matrix-vector
  multiplication.
\newblock {\em Concurrency and Computation: Practice and Experience},
  30(23):e4800, 2018.

\bibitem{nishtala_when_2007}
Rajesh Nishtala, Richard~W. Vuduc, James~W. Demmel, and Katherine~A. Yelick.
\newblock When cache blocking of sparse matrix vector multiply works and why.
\newblock {\em Applicable Algebra in Engineering, Communication and Computing},
  18(3):297--311, May 2007.

\bibitem{yzelman_generalised_2015}
A.~N. Yzelman.
\newblock Generalised {Vectorisation} for {Sparse} {Matrix}: {Vector}
  {Multiplication}.
\newblock In {\em Proceedings of the 5th {Workshop} on {Irregular}
  {Applications}: {Architectures} and {Algorithms}}, {IA}$^{\textrm{3}}$ '15,
  pages 6:1--6:8, New York, NY, USA, 2015. ACM.

\bibitem{gustavson_two_1978}
Fred~G. Gustavson.
\newblock Two {Fast} {Algorithms} for {Sparse} {Matrices}: {Multiplication} and
  {Permuted} {Transposition}.
\newblock {\em ACM Transactions on Mathematical Software (TOMS)},
  4(3):250--269, September 1978.

\bibitem{aydin_distributed_2019}
Kevin Aydin, MohammadHossein Bateni, and Vahab Mirrokni.
\newblock Distributed {Balanced} {Partitioning} via {Linear} {Embedding}
  {\textdagger}.
\newblock {\em Algorithms}, 12(8):162, August 2019.
\newblock Number: 8 Publisher: Multidisciplinary Digital Publishing Institute.

\bibitem{buluc_representation_2008}
Aydin Buluc and John~R. Gilbert.
\newblock On the representation and multiplication of hypersparse matrices.
\newblock In {\em 2008 {IEEE} {International} {Symposium} on {Parallel} and
  {Distributed} {Processing}}, pages 1--11, April 2008.
\newblock ISSN: 1530-2075.

\bibitem{buluc_parallel_2009}
Aydin Bulu{\c c}, Jeremy~T. Fineman, Matteo Frigo, John~R. Gilbert, and
  Charles~E. Leiserson.
\newblock Parallel sparse matrix-vector and matrix-transpose-vector
  multiplication using compressed sparse blocks.
\newblock page 233. ACM Press, 2009.

\bibitem{zhang_making_2017}
Yunming Zhang, Vladimir Kiriansky, Charith Mendis, Saman Amarasinghe, and Matei
  Zaharia.
\newblock Making caches work for graph analytics.
\newblock In {\em 2017 {IEEE} {International} {Conference} on {Big} {Data}
  ({Big} {Data})}, pages 293--302, December 2017.

\bibitem{hong_adaptive_2019}
Changwan Hong, Aravind Sukumaran-Rajam, Israt Nisa, Kunal Singh, and
  P.~Sadayappan.
\newblock Adaptive {Sparse} {Tiling} for {Sparse} {Matrix} {Multiplication}.
\newblock In {\em Proceedings of the 24th {Symposium} on {Principles} and
  {Practice} of {Parallel} {Programming}}, {PPoPP} '19, pages 300--314, New
  York, NY, USA, February 2019. Association for Computing Machinery.

\bibitem{namashivayam_variable-sized_2021}
Naveen Namashivayam, Sanyam Mehta, and Pen-Chung Yew.
\newblock Variable-{Sized} {Blocks} for {Locality}-{Aware} {SpMV}.
\newblock In {\em Proc. of the {Annual} {IEEE}/{ACM} {Int}'l {Symp}. on {Code}
  {Generation} and {Optimization} ({CGO})}, Virtual Event, Seoul, March 2021.

\bibitem{greene_k-way_1991}
W.A. Greene.
\newblock k-way merging and k-ary sorts.
\newblock In {\em [{Proceedings}] 1991 {Symposium} on {Applied} {Computing}},
  pages 197--, April 1991.
\newblock ISSN: null.

\bibitem{bezanson_julia:_2017}
J.~Bezanson, A.~Edelman, S.~Karpinski, and V.~Shah.
\newblock Julia: {A} {Fresh} {Approach} to {Numerical} {Computing}.
\newblock {\em SIAM Review}, 59(1):65--98, January 2017.

\bibitem{noauthor_ieee_2019}
{IEEE} {Standard} for {Floating}-{Point} {Arithmetic}.
\newblock {\em IEEE Std 754-2019 (Revision of IEEE 754-2008)}, pages 1--84,
  July 2019.

\bibitem{erik_schnetter_eschnettsimdjl_2019}
Erik Schnetter, Takafumi Arakaki, Valentin Churavy, Kristoffer Carlsson,
  Nicolau~Leal Werneck, Steve Kelly, Gunnar Farneb{\"a}ck, Miguel
  Raz~Guzm{\'a}n Macedo, Matt Bauman, Kenta Sato, and Elliot Saba.
\newblock eschnett/{SIMD}.jl: v2.8.0, July 2019.

\bibitem{davis_university_2011}
Timothy~A. Davis and Yifan Hu.
\newblock The university of {Florida} sparse matrix collection.
\newblock {\em ACM Transactions on Mathematical Software}, 38(1):1--25,
  November 2011.

\bibitem{dolan_benchmarking_2002}
Elizabeth~D. Dolan and Jorge~J. Mor{\'e}.
\newblock Benchmarking optimization software with performance profiles.
\newblock {\em Mathematical Programming}, 91(2):201--213, January 2002.

\bibitem{vuduc_fast_2005-1}
R.~W. Vuduc and H.~Moon.
\newblock Fast sparse matrix-vector multiplication by exploiting variable block
  structure.
\newblock Technical Report UCRL-TR-213454, Lawrence Livermore National Lab.
  (LLNL), Livermore, CA (United States), July 2005.

\end{thebibliography}

\newpage
\clearpage

\appendix

\section{Finding Optimal VBR Partitions is NP-Hard}\label{app:vbrblockingnphard}

In this section, we show that Problem \ref{prob:vbrblocking}, finding the row
and column partition which maximizes a cost model $f(A, \Pi, \Phi)$ of the VBR
representation of some matrix $A$, is NP-Hard by reduction from the Maximum
Cut problem, one of Karp's 21 NP-Complete problems
\cite{karp_reducibility_1972, papadimitriou_optimization_1991}.

We restate the Maximum Cut problem here for convenience.
\begin{Definition}[Maximum Cut]\label{prob:maxcut}
Given an undirected graph $G = (V, E)$ with $m$ nodes and $n$ edges, split
$V$ into two sets $C_1 \subset V$ and $C_2 \subset V$ where $C_1 \cap C_2 =
\emptyset$ and the number of edges between $C_1$ and $C_2$ is maximized.
\end{Definition}

\begin{theorem}\label{thm:vbrblockingnphard}
Problem \ref{prob:vbrblocking} is NP-Hard for any $u_{\max} \geq 2$ and
$w_{\max} \geq 2$ and cost functions of the form
\begin{equation} \tag{\ref{eq:nphardcost}}
    f(A, \Pi, \Phi) = s\cdot N_{\text{index}} + N_{\text{value}}
\end{equation}
where $s \geq 1$ is constant.
\end{theorem}

\begin{proof}
Assume we are given an instance of Maximum Cut (Problem \ref{prob:maxcut}).
We first define a matrix $A$ in terms of $G$ and $s$, then show a
correspondence between a class of partitions of $A$ and cuts in $G$. Finally,
we show that the $\Pi$ and $\Phi$ which optimize \eqref{eq:nphardcost} correspond
to a maximum cut through $G$.

Let $A$ be an $\mu m \times \mu n$ matrix of zeros and nonzeros, where nonzeros are
represented with $x$. Unless stated otherwise, entries of $A$ are defined to
be zero. Fix an ordering of the edges of $G$, and let $e_j = (i_1, i_2)$
where $i_1 < i_2$ be the $j^{th}$ edge in this ordering of $G$. We will
insert a $\mu \times \mu$ gadget into $A$ at each endpoint of $e_j$, where
\begin{align}
    \mu &= 3 + \mu_1 + (1 + \mu_1)\mu_2 + 2(1 + \mu_1)\mu_3\\
    \mu_1 &= \lfloor s + 1 \rfloor\\
    \mu_2 &= 32\\
    \mu_3 &= \lceil 28s - 10 \rceil
\end{align}
These constants depend on the relative weights that $s$ assigns to each
block and the size of each block. They are larger to make the proof shorter;
making them large allows us to upper bound $\mu_3$ by $28s-9$ when calculating
the cost of each gadget. To give an example of some smaller constants, if $s
= 1$, we can use $\mu_1 = 2, \mu_2 = 3, \mu_3 = 1$.

If we think of $A$ as being tiled with $\mu \times \mu$ tiles, the placement of
these tiles is analogous to the incidence matrix representation of $G$, so that
rows of tiles correspond to vertices and columns of tiles correspond to edges.
We insert the gadget $B_1$ at the intersection of the $i_1^{th}$ tile row and
$j^{th}$ tile column.
\[
B_1 =
\left[\begin{array}{cccccccccc}
 x & & & \cdots & x & \cdots & x & \cdots & & \cdots\\
 & x & & \cdots & x & \cdots & & \cdots & & \cdots\\
 & & x & \cdots & x & \cdots & & \cdots & x & \cdots\\
 \vdots & \vdots & \vdots & & & & & & & \\
 x & x & x & & & & & & \\
 \vdots & \vdots & \vdots & & & & & & & \\
 x & & & & & & & & & \\
 \vdots & \vdots & \vdots & & & & & & & \\
 & & x & & & & & & & \\
 \vdots & \vdots & \vdots & & & & & & & \\
\end{array}\right]
\]
We insert the gadget $B_2$ at the intersection of the $i_2^{th}$ tile row and
$j^{th}$ tile column.
\[
B_2 =
\left[\begin{array}{cccccccccc}
 & & x & \cdots & x & \cdots & x & \cdots & & \cdots\\
 & x & & \cdots & x & \cdots & & \cdots & & \cdots\\
 x & & & \cdots & x & \cdots & & \cdots & x & \cdots\\
 \vdots & \vdots & \vdots & & & & & & & \\
 x & x & x & & & & & & \\
 \vdots & \vdots & \vdots & & & & & & & \\
 x & & & & & & & & & \\
 \vdots & \vdots & \vdots & & & & & & & \\
 & & x & & & & & & & \\
 \vdots & \vdots & \vdots & & & & & & & \\
\end{array}\right]
\]

where the upper left patterns occur once, the patterns second to the right
and the bottom are repeated $\mu_2$ times, and the two rightmost and bottommost
patterns are repeated $\mu_3$ times. All patterns are followed by $\mu_1$ rows or
columns of zeros. Figure \ref{fig:biggadgetmatrix} gives an example of $A$ for
some $G$.

The gadgets are identical except for the upper left $3 \times 3$ pattern.
Thus, the patterns on the top of the gadgets are column-aligned and the
patterns on the left are row-aligned across gadget rows and gadget columns.
We refer to the resulting fully zero $\mu_1$ rows (resp. columns) as
\textbf{filler} rows (resp. columns).

We start by arguing that it is never optimal to produce a partition with a
row part that contains both filler rows and non-filler rows. A symmetric
argument holds for the columns.

First, consider the case where the row part contains filler rows on the top
or bottom. Separating these rows from the part reduces the sizes
$N_{\text{value}}$ of the blocks in that part without changing the number
$N_{\text{index}}$ of blocks, so the part cannot have been optimal.

Second, if the row part does not start or end with a filler row, it must
contain filler rows. Since the filler rows in $A$ occur in contiguous groups
of size $\mu_1$, this part must contain such a group. Consider a block in this
part. If the block contains nonzeros on only one side of the filler rows,
then separating the rows strictly reduces the size of the block without
adding any new blocks. If the block contains nonzeros on both sides of the
filler rows, then removing the rows creates a block, but deletes at least
$\mu_1$ explicitly stored zero values. Since $\mu_1 > s$, separating these filler
rows still reduces the cost of the partition, so it cannot have been optimal.

Therefore, optimal partitions do not merge different patterns together. We
won't concern ourselves with whether the filler rows have been merged
together, since it doesn't change the cost function. Since the patterns on
top consist of only one column, and the patterns on the side consist of only
one row, the only undetermined piece of our optimal partition is the
partition of the first three rows and columns of each gadget row and gadget
column, respectively.

There remains only four cases for the rows. Either each row lies in a
separate part, all rows share a part, the first two rows share a part, or the
last two rows share a part. A symmetric argument holds for the columns. 

Sections \ref{app:vbrblockingnphard:happygadget} and
\ref{app:vbrblockingnphard:sadgadget} exhaustively check that for all cases
where both the first three rows and columns of a gadget
have two parts each ($\Pi, \Phi \in \{[1{:}1, 2{:}3, ...], [1{:}2, 3{:}3, ...]\}$),
\begin{equation}\label{eq:happygadget}
    f(B_1, \Pi, \Phi) \leq 146 + 263s + 112s^2
\end{equation}
and that in all the other cases,
\begin{equation}\label{eq:sadgadget}
    f(B_1, \Pi, \Phi) \geq 147 + 263s + 112s^2
\end{equation}
The exhaustive proofs for $B_2$ are symmetric, so we omit them.

Assume we start with a row and column partition where only the first three
rows and columns of each gadget share parts in the partition. For every row
or column part with three members, we split off one row or column into a
different part. For any case where the first three rows or columns of a
gadget row or column all belong to different parts, we merge two of the rows
or columns. For every gadget in our initial partition whose first three rows
and columns had two parts each, it's blocks will be unchanged. For every other gadget,
the cost will be strictly reduced. Thus, optimal partitions only merge pairs
of rows and columns, and these pairs occur in the first three rows or columns
of each gadget row or column. In this case, the cost of every subassembly is
the same except for the upper left $3 \times 3$ pattern of each gadget.
Therefore, the remainder of the argument focuses on these assemblies.

At this point, we can establish a correspondence between cuts in the graph
and partitions. Let $(C_1, C_2)$ be a cut in the graph. We will define a row
partition $\Pi$ corresponding to this cut. Unless stated otherwise, rows in
this partition are assigned to distinct parts. If a vertex $i$ lies in $C_1$,
then we merge the first and second rows of the corresponding gadget row. If
our vertex $i$ lies in $C_2$, then we merge the second and third rows of the
gadget row. Consider the gadgets corresponding to an edge $e_j = (i_1,
i_2)$. Notice that if vertices $i_1$ and $i_2$ lie in the same part, $C_1$
for example, we have one of the following situations:
\[
\begin{array}{cc}\makecell{v_i \in V_1 \\ \Pi=[1{:}2, 3{:}3]} & 
\left[\begin{array}{ccc}\cline{1-1} \cline{2-2} 
\multicolumn{1}{|c}{x}& \multicolumn{1}{c|}{0}& \multicolumn{1}{c}{}\\
\multicolumn{1}{|c}{0}& \multicolumn{1}{c|}{x}& \multicolumn{1}{c}{}\\\cline{1-1} \cline{2-2} \cline{3-3} 
\multicolumn{1}{c}{}& \multicolumn{1}{c|}{}& \multicolumn{1}{c|}{x}\\\cline{3-3} 
\end{array}\right]
\\ & 
\\ \makecell{v_{i'} \in V_1 \\ \Pi=[1{:}2, 3{:}3]} &
\left[\begin{array}{ccc}\cline{1-1} \cline{2-2} \cline{3-3} 
\multicolumn{1}{|c}{0}& \multicolumn{1}{c|}{0}& \multicolumn{1}{c|}{x}\\
\multicolumn{1}{|c}{0}& \multicolumn{1}{c|}{x}& \multicolumn{1}{c|}{0}\\\cline{1-1} \cline{2-2} \cline{3-3} 
\multicolumn{1}{|c}{x}& \multicolumn{1}{c|}{0}& \multicolumn{1}{c}{}\\\cline{1-1} \cline{2-2} 
\end{array}\right]
\\ & 
\\ & e_j, \Phi=[1{:}2, 3{:}3]
\end{array}
\text{ or }
\begin{array}{c} 
\left[\begin{array}{ccc}\cline{1-1} \cline{2-2} \cline{3-3} 
\multicolumn{1}{|c|}{x}& \multicolumn{1}{c}{0}& \multicolumn{1}{c|}{0}\\
\multicolumn{1}{|c|}{0}& \multicolumn{1}{c}{x}& \multicolumn{1}{c|}{0}\\\cline{1-1} \cline{2-2} \cline{3-3} 
\multicolumn{1}{c|}{}& \multicolumn{1}{c}{0}& \multicolumn{1}{c|}{x}\\\cline{2-2} \cline{3-3} 
\end{array}\right]
\\
\\
\left[\begin{array}{ccc}\cline{2-2} \cline{3-3} 
\multicolumn{1}{c|}{}& \multicolumn{1}{c}{0}& \multicolumn{1}{c|}{x}\\
\multicolumn{1}{c|}{}& \multicolumn{1}{c}{x}& \multicolumn{1}{c|}{0}\\\cline{1-1} \cline{2-2} \cline{3-3} 
\multicolumn{1}{|c|}{x}& \multicolumn{1}{c}{}& \multicolumn{1}{c}{}\\\cline{1-1} 
\end{array}\right]
\\
\\e_j, \Phi=[1{:}1, 2{:}3]
\end{array}
\]

The cost of either arrangement is $f = 13 + 5s$. However, if vertices $i_1$
and $i_2$ lie in different parts, $C_1$ and $C_2$, for instance, the
situation is as follows:
\[
\begin{array}{cc}\makecell{v_i \in V_1 \\ \Pi=[1{:}2, 3{:}3]} & 
\left[\begin{array}{ccc}\cline{1-1} \cline{2-2} 
\multicolumn{1}{|c}{x}& \multicolumn{1}{c|}{0}& \multicolumn{1}{c}{}\\
\multicolumn{1}{|c}{0}& \multicolumn{1}{c|}{x}& \multicolumn{1}{c}{}\\\cline{1-1} \cline{2-2} \cline{3-3} 
\multicolumn{1}{c}{}& \multicolumn{1}{c|}{}& \multicolumn{1}{c|}{x}\\\cline{3-3} 
\end{array}\right]
\\ & 
\\ \makecell{v_{i'} \in V_2 \\ \Pi=[1{:}1, 2{:}3]} &
\left[\begin{array}{ccc}\cline{3-3} 
\multicolumn{1}{c}{}& \multicolumn{1}{c|}{}& \multicolumn{1}{c|}{x}\\\cline{1-1} \cline{2-2} \cline{3-3} 
\multicolumn{1}{|c}{0}& \multicolumn{1}{c|}{x}& \multicolumn{1}{c}{}\\
\multicolumn{1}{|c}{x}& \multicolumn{1}{c|}{0}& \multicolumn{1}{c}{}\\\cline{1-1} \cline{2-2} 
\end{array}\right]
\\ & 
\\ & e_j, \Phi=[1{:}2, 3{:}3]
\end{array}
\text{ or }
\begin{array}{c} 
\left[\begin{array}{ccc}\cline{1-1} \cline{2-2} \cline{3-3} 
\multicolumn{1}{|c|}{x}& \multicolumn{1}{c}{0}& \multicolumn{1}{c|}{0}\\
\multicolumn{1}{|c|}{0}& \multicolumn{1}{c}{x}& \multicolumn{1}{c|}{0}\\\cline{1-1} \cline{2-2} \cline{3-3} 
\multicolumn{1}{c|}{}& \multicolumn{1}{c}{0}& \multicolumn{1}{c|}{x}\\\cline{2-2} \cline{3-3} 
\end{array}\right]
\\
\\
\left[\begin{array}{ccc}\cline{2-2} \cline{3-3} 
\multicolumn{1}{c|}{}& \multicolumn{1}{c}{0}& \multicolumn{1}{c|}{x}\\\cline{1-1} \cline{2-2} \cline{3-3} 
\multicolumn{1}{|c|}{0}& \multicolumn{1}{c}{x}& \multicolumn{1}{c|}{0}\\
\multicolumn{1}{|c|}{x}& \multicolumn{1}{c}{0}& \multicolumn{1}{c|}{0}\\\cline{1-1} \cline{2-2} \cline{3-3} 
\end{array}\right]
\\
\\e_j, \Phi=[1{:}1, 2{:}3]
\end{array}
\]

Since these are the only two gadgets in the column corresponding to this
edge, we are free to choose a column partition. The above partition of
minimal cost has a cost of $f = 10 + 4s$, less than the case where the
vertices shared an edge.

Thus, the cost of an optimal column partition corresponding to the row
partition representation of a cut can be expressed as a constant minus $3 +
s$ times the number of cut edges. Since there is a bijection between cuts and
our ``pairwise'' row partitions (one of which we know to be optimal),
producing an optimal partition of rows and columns is equivalent to finding
the maximum cut in $G$. If we treat $s$ as a constant, our reduction imposes
only a constant factor of overhead in $m$ and $n$, and Problem
\ref{prob:maxcut} is reducible to Problem \ref{prob:vbrblocking} in
polynomial time.
\end{proof}

\begin{figure*}
    \centering
    \[
     G =
    \begin{tikzpicture}[auto, baseline=0cm, node distance=3cm, every loop/.style={},
        thick,main node/.style={circle,draw,font=\sffamily\Large\bfseries}]

    \node[main node] (1) {1};
    \node[main node] (2) [above right of=1] {2};
    \node[main node] (3) [below right of=2] {3};
    \node[main node] (4) [below right of=1] {4};

    \path[every node/.style={font=\sffamily\small}]
    (1) edge node {} (2)
    (1) edge node {} (3)
    (1) edge node {} (4)
    (2) edge node {} (3);

    \draw[dashed] 
    ([yshift=0.75cm]$ (1)!0.5!(2) $ ) -- 
    ([yshift=-0.75cm]$ (1)!0.5!(4) $ );
    \end{tikzpicture}
    \]

    \[
     A = \resizebox{0.93\textwidth}{!}{\input{thebiggadgetmatrix}}
    \]
    \caption{An example of our reduction for a simple graph $G$.  Recall that
    rows of gadgets correspond to vertices, and columns of gadgets correspond to
    edges, so that the gadget positions in $A$ are analogous to an incidence
    matrix (not an adjacency matrix) representation of $G$. The maximum cut in
    $G$ is shown, and the corresponding optimal partition of $A$ under the cost
    function where $s = 1$ is shown.  Notice that gadgets corresponding to edges
    that cross the cut cost less than the gadgets of edges that do not cross the
    cut.}\label{fig:biggadgetmatrix}
\end{figure*}

Notice that our proof of Theorem \ref{thm:vbrblockingnphard} makes no
assumption on the size of $u_{\max}$ or $w_{\max}$, asking only that they are
at least 2 and enforcing pairwise partitions with the cost function. We chose
this proof technique since the cost function is realistic and in some
situations there may be no limit on the sizes of blocks ($w_{\max} = n$).
However, it requires large gadgets when $s$ is large. Alternative gadgets can
be used in situations where $s$ is large by instead using $u_{\max}$ and
$w_{\max}$ to constrain the sizes of blocks and choosing different gadgets.
This leads us to Theorem \ref{thm:vbrblockingnphard2}, which corresponds to
the case where $s$ is large enough that $N_{\text{value}}$ would be
considered negligibly small.

\begin{theorem}\label{thm:vbrblockingnphard2}
Problem \ref{prob:vbrblocking} is NP-Hard for any $u_{\max} \geq 2$ and
$w_{\max} \geq 2$ and the cost function
\begin{equation} \tag{\ref{eq:nphardcost2}}
    f(A, \Pi, \Phi) = N_{\text{index}}.
\end{equation}
\end{theorem}

\begin{proof}
The proof is similar to that of Theorem \ref{thm:vbrblockingnphard}, but
because $u_{\max}$ and $w_{\max}$ are the only factors limiting the size of blocks, we will
use different gadgets. Since the rest of the proof is so similar, we only
describe the new gadgets.

Let $B_1$ be a $2u_{\max} + 1 \times 2w_{\max} + 1$
matrix whose upper left $u_{\max} + 1 \times w_{\max} + 1$ subregion is nonzero,
save for the upper right and lower left entries of the subregion, and zero
everywhere else. Let $B_2$ be the same as $B_1$, except $B_2$'s upper left
and lower right entries of the dense subregion are zero and the upper right
and lower left entries of the upper left subregion are nonzero.
For example, if $u_{\max} = w_{\max} = 3$, our gadgets are defined as:
\[
B_1 =
\left[\begin{array}{ccccc}
x & x & x & & \cdots \\
x & x & x & x & \cdots\\
x & x & x & x & \cdots\\
  & x & x & x & \cdots\\
\vdots & \vdots & \vdots & \vdots & \\
\end{array}\right], ~~~
B_2 =
\left[\begin{array}{ccccc}
  & x & x & x & \cdots \\
x & x & x & x & \cdots\\
x & x & x & x & \cdots\\
x & x & x &   & \cdots\\
\vdots & \vdots & \vdots & \vdots & \\
\end{array}\right]
\]

Our cost function asks us to minimize the total number of blocks. Since the
dense region is of size $u_{\max} + 1 \times w_{\max} + 1$, each gadget must
contribute at least three blocks to this cost function. This can be achieved
with one horizontal and one vertical split in the dense region that isolates
one of the zeros on the corners. Any other decomposition with a single
horizontal and single vertical split produces four blocks. Thus, when
vertices $i_1$ and $i_2$ lie on the same side of the cut, this corresponds to
one of the following situations (we keep our example value of $u_{\max} = w_{\max} =
3$):
\[
\begin{array}{cc}\makecell{v_i \in V_1, \Pi \\ = [1{:}3, 4{:}4]} & 
\left[\begin{array}{cccc}\cline{1-1} \cline{2-2} \cline{3-3} \cline{4-4} 
\multicolumn{1}{|c}{x}& \multicolumn{1}{c}{x}& \multicolumn{1}{c|}{x}& \multicolumn{1}{c|}{0}\\
\multicolumn{1}{|c}{x}& \multicolumn{1}{c}{x}& \multicolumn{1}{c|}{x}& \multicolumn{1}{c|}{x}\\
\multicolumn{1}{|c}{x}& \multicolumn{1}{c}{x}& \multicolumn{1}{c|}{x}& \multicolumn{1}{c|}{x}\\\cline{1-1} \cline{2-2} \cline{3-3} \cline{4-4} 
\multicolumn{1}{|c}{0}& \multicolumn{1}{c}{x}& \multicolumn{1}{c|}{x}& \multicolumn{1}{c|}{x}\\\cline{1-1} \cline{2-2} \cline{3-3} \cline{4-4} 
\end{array}\right]
\\ & 
\\ \makecell{v_{i'} \in V_1, \Pi \\ = [1{:}3, 4{:}4]}&
\left[\begin{array}{cccc}\cline{1-1} \cline{2-2} \cline{3-3} \cline{4-4} 
\multicolumn{1}{|c}{0}& \multicolumn{1}{c}{x}& \multicolumn{1}{c|}{x}& \multicolumn{1}{c|}{x}\\
\multicolumn{1}{|c}{x}& \multicolumn{1}{c}{x}& \multicolumn{1}{c|}{x}& \multicolumn{1}{c|}{x}\\
\multicolumn{1}{|c}{x}& \multicolumn{1}{c}{x}& \multicolumn{1}{c|}{x}& \multicolumn{1}{c|}{x}\\\cline{1-1} \cline{2-2} \cline{3-3} \cline{4-4} 
\multicolumn{1}{|c}{x}& \multicolumn{1}{c}{x}& \multicolumn{1}{c|}{x}& \multicolumn{1}{c}{}\\\cline{1-1} \cline{2-2} \cline{3-3} 
\end{array}\right]
\\ & 
\\ & e_j, \Phi = [1{:}3, 4{:}4]
\end{array}
\text{ or }
\begin{array}{c} 
\left[\begin{array}{cccc}\cline{1-1} \cline{2-2} \cline{3-3} \cline{4-4} 
\multicolumn{1}{|c|}{x}& \multicolumn{1}{c}{x}& \multicolumn{1}{c}{x}& \multicolumn{1}{c|}{0}\\
\multicolumn{1}{|c|}{x}& \multicolumn{1}{c}{x}& \multicolumn{1}{c}{x}& \multicolumn{1}{c|}{x}\\
\multicolumn{1}{|c|}{x}& \multicolumn{1}{c}{x}& \multicolumn{1}{c}{x}& \multicolumn{1}{c|}{x}\\\cline{1-1} \cline{2-2} \cline{3-3} \cline{4-4} 
\multicolumn{1}{c|}{}& \multicolumn{1}{c}{x}& \multicolumn{1}{c}{x}& \multicolumn{1}{c|}{x}\\\cline{2-2} \cline{3-3} \cline{4-4} 
\end{array}\right]
\\
\\
\left[\begin{array}{cccc}\cline{1-1} \cline{2-2} \cline{3-3} \cline{4-4} 
\multicolumn{1}{|c|}{0}& \multicolumn{1}{c}{x}& \multicolumn{1}{c}{x}& \multicolumn{1}{c|}{x}\\
\multicolumn{1}{|c|}{x}& \multicolumn{1}{c}{x}& \multicolumn{1}{c}{x}& \multicolumn{1}{c|}{x}\\
\multicolumn{1}{|c|}{x}& \multicolumn{1}{c}{x}& \multicolumn{1}{c}{x}& \multicolumn{1}{c|}{x}\\\cline{1-1} \cline{2-2} \cline{3-3} \cline{4-4} 
\multicolumn{1}{|c|}{x}& \multicolumn{1}{c}{x}& \multicolumn{1}{c}{x}& \multicolumn{1}{c|}{0}\\\cline{1-1} \cline{2-2} \cline{3-3} \cline{4-4} 
\end{array}\right]
\\
\\e_j, \Phi = [1{:}1, 2{:}4]
\end{array}
\]

When vertices $i_1$ and $i_2$ like
on different sides of the cut, we have:
\[
\begin{array}{cc}\makecell{v_i \in V_1, \Pi \\ = [1{:}3, 4{:}4]} & 
\left[\begin{array}{cccc}\cline{1-1} \cline{2-2} \cline{3-3} \cline{4-4} 
\multicolumn{1}{|c}{x}& \multicolumn{1}{c}{x}& \multicolumn{1}{c|}{x}& \multicolumn{1}{c|}{0}\\
\multicolumn{1}{|c}{x}& \multicolumn{1}{c}{x}& \multicolumn{1}{c|}{x}& \multicolumn{1}{c|}{x}\\
\multicolumn{1}{|c}{x}& \multicolumn{1}{c}{x}& \multicolumn{1}{c|}{x}& \multicolumn{1}{c|}{x}\\\cline{1-1} \cline{2-2} \cline{3-3} \cline{4-4} 
\multicolumn{1}{|c}{0}& \multicolumn{1}{c}{x}& \multicolumn{1}{c|}{x}& \multicolumn{1}{c|}{x}\\\cline{1-1} \cline{2-2} \cline{3-3} \cline{4-4} 
\end{array}\right]
\\ & 
\\ \makecell{v_{i'} \in V_2, \Pi \\ = [1{:}1, 2{:}4]} &
\left[\begin{array}{cccc}\cline{1-1} \cline{2-2} \cline{3-3} \cline{4-4} 
\multicolumn{1}{|c}{0}& \multicolumn{1}{c}{x}& \multicolumn{1}{c|}{x}& \multicolumn{1}{c|}{x}\\\cline{1-1} \cline{2-2} \cline{3-3} \cline{4-4} 
\multicolumn{1}{|c}{x}& \multicolumn{1}{c}{x}& \multicolumn{1}{c|}{x}& \multicolumn{1}{c|}{x}\\
\multicolumn{1}{|c}{x}& \multicolumn{1}{c}{x}& \multicolumn{1}{c|}{x}& \multicolumn{1}{c|}{x}\\
\multicolumn{1}{|c}{x}& \multicolumn{1}{c}{x}& \multicolumn{1}{c|}{x}& \multicolumn{1}{c|}{0}\\\cline{1-1} \cline{2-2} \cline{3-3} \cline{4-4} 
\end{array}\right]
\\ & 
\\ & e_j, \Phi = [1{:}3, 4{:}4]
\end{array}
\text{ or }
\begin{array}{c} 
\left[\begin{array}{cccc}\cline{1-1} \cline{2-2} \cline{3-3} \cline{4-4} 
\multicolumn{1}{|c|}{x}& \multicolumn{1}{c}{x}& \multicolumn{1}{c}{x}& \multicolumn{1}{c|}{0}\\
\multicolumn{1}{|c|}{x}& \multicolumn{1}{c}{x}& \multicolumn{1}{c}{x}& \multicolumn{1}{c|}{x}\\
\multicolumn{1}{|c|}{x}& \multicolumn{1}{c}{x}& \multicolumn{1}{c}{x}& \multicolumn{1}{c|}{x}\\\cline{1-1} \cline{2-2} \cline{3-3} \cline{4-4} 
\multicolumn{1}{c|}{}& \multicolumn{1}{c}{x}& \multicolumn{1}{c}{x}& \multicolumn{1}{c|}{x}\\\cline{2-2} \cline{3-3} \cline{4-4} 
\end{array}\right]
\\
\\
\left[\begin{array}{cccc}\cline{2-2} \cline{3-3} \cline{4-4} 
\multicolumn{1}{c|}{}& \multicolumn{1}{c}{x}& \multicolumn{1}{c}{x}& \multicolumn{1}{c|}{x}\\\cline{1-1} \cline{2-2} \cline{3-3} \cline{4-4} 
\multicolumn{1}{|c|}{x}& \multicolumn{1}{c}{x}& \multicolumn{1}{c}{x}& \multicolumn{1}{c|}{x}\\
\multicolumn{1}{|c|}{x}& \multicolumn{1}{c}{x}& \multicolumn{1}{c}{x}& \multicolumn{1}{c|}{x}\\
\multicolumn{1}{|c|}{x}& \multicolumn{1}{c}{x}& \multicolumn{1}{c}{x}& \multicolumn{1}{c|}{0}\\\cline{1-1} \cline{2-2} \cline{3-3} \cline{4-4} 
\end{array}\right]
\\
\\e_j, \Phi = [1{:}1, 2{:}4]
\end{array}
\]

By choosing the correct column partition, the case where vertices $i_1$ and
$i_2$ lie on different sides of the cut can be made to use only 6 blocks.
When these vertices are on the same side of the cut we require 7
blocks.
\end{proof}

Corollary \ref{cor:symvbrblockingnphard} addresses the hardness of the
symmetric case, when the row and column partitions are constrained to be the
same.

\begin{corollary}\label{cor:symvbrblockingnphard}
    Even when the row and column partitions are constrained to be the same,
    Problem \ref{prob:vbrblocking} is NP-Hard for any $u_{\max} \geq 2$,
    $w_{\max} \geq 2$ and the cost functions of Theorems
    \ref{thm:vbrblockingnphard} or \ref{thm:vbrblockingnphard2}.
\end{corollary}

\begin{proof}
    We proceed by reducing the asymmetric case to the symmetric one. Consider
    an instance of Problem \ref{prob:vbrblocking} for an $m \times n$ matrix
    $A$. We then form the $m + n \times m + n$ matrix
    \[
        B = \left[\begin{array}{cc} 0 & A \\ A^{\intercal} & 0 \end{array}\right].
    \]
    We then solve our symmetric problem on $B$. Note that we can match or improve the
    cost of any partition with a group which spans both sides of our block
    matrix by simply splitting that group along the boundary, since no
    nonzeros are shared between the two sides of the boundary. We set $\Pi$
    to be the partition of the first $m$ rows/columns of $B$, and $\Phi$ to
    be the partition of the last $n$ rows/columns of $B$. Note that the cost
    of our symmetric partition on $B$ will be twice that of the resulting
    cost of our partition of $A$, and any partition of $A$ can induce a
    corresponding partition of $B$. Thus, a solution to our symmetric problem
    on $B$ would solve the original problem of partitioning $A$, and the symmetric
    case is NP-hard as well.
\end{proof}

\subsection{Exhaustively Checking \eqref{eq:happygadget}}\label{app:vbrblockingnphard:happygadget}
\paragraph{$\Pi = [1{:}1, 2{:}3, ...], \Phi = [1{:}1, 2{:}3, ...]$}
\[
\left[\begin{array}{cccccccccc}\cline{1-1} \cline{5-5} \cline{7-7} 
\multicolumn{1}{|c|}{x}& \multicolumn{1}{c}{}& \multicolumn{1}{c}{}& \multicolumn{1}{c|}{\cdots}& \multicolumn{1}{c|}{x}& \multicolumn{1}{c|}{\cdots}& \multicolumn{1}{c|}{x}& \multicolumn{1}{c}{\cdots}& \multicolumn{1}{c}{}& \multicolumn{1}{c}{\cdots}\\\cline{1-1} \cline{2-2} \cline{3-3} \cline{5-5} \cline{7-7} \cline{9-9} 
\multicolumn{1}{c|}{}& \multicolumn{1}{c}{x}& \multicolumn{1}{c|}{0}& \multicolumn{1}{c|}{\cdots}& \multicolumn{1}{c|}{x}& \multicolumn{1}{c}{\cdots}& \multicolumn{1}{c}{}& \multicolumn{1}{c|}{\cdots}& \multicolumn{1}{c|}{0}& \multicolumn{1}{c}{\cdots}\\
\multicolumn{1}{c|}{}& \multicolumn{1}{c}{0}& \multicolumn{1}{c|}{x}& \multicolumn{1}{c|}{\cdots}& \multicolumn{1}{c|}{x}& \multicolumn{1}{c}{\cdots}& \multicolumn{1}{c}{}& \multicolumn{1}{c|}{\cdots}& \multicolumn{1}{c|}{x}& \multicolumn{1}{c}{\cdots}\\\cline{2-2} \cline{3-3} \cline{5-5} \cline{9-9} 
\multicolumn{1}{c}{\vdots}& \multicolumn{1}{c}{\vdots}& \multicolumn{1}{c}{\vdots}& \multicolumn{1}{c}{}& \multicolumn{1}{c}{}& \multicolumn{1}{c}{}& \multicolumn{1}{c}{}& \multicolumn{1}{c}{}& \multicolumn{1}{c}{}& \multicolumn{1}{c}{}\\\cline{1-1} \cline{2-2} \cline{3-3} 
\multicolumn{1}{|c|}{x}& \multicolumn{1}{c}{x}& \multicolumn{1}{c|}{x}& \multicolumn{1}{c}{}& \multicolumn{1}{c}{}& \multicolumn{1}{c}{}& \multicolumn{1}{c}{}& \multicolumn{1}{c}{}& \multicolumn{1}{c}{}& \multicolumn{1}{c}{}\\\cline{1-1} \cline{2-2} \cline{3-3} 
\multicolumn{1}{c}{\vdots}& \multicolumn{1}{c}{\vdots}& \multicolumn{1}{c}{\vdots}& \multicolumn{1}{c}{}& \multicolumn{1}{c}{}& \multicolumn{1}{c}{}& \multicolumn{1}{c}{}& \multicolumn{1}{c}{}& \multicolumn{1}{c}{}& \multicolumn{1}{c}{}\\\cline{1-1} 
\multicolumn{1}{|c|}{x}& \multicolumn{1}{c}{}& \multicolumn{1}{c}{}& \multicolumn{1}{c}{}& \multicolumn{1}{c}{}& \multicolumn{1}{c}{}& \multicolumn{1}{c}{}& \multicolumn{1}{c}{}& \multicolumn{1}{c}{}& \multicolumn{1}{c}{}\\\cline{1-1} 
\multicolumn{1}{c}{\vdots}& \multicolumn{1}{c}{\vdots}& \multicolumn{1}{c}{\vdots}& \multicolumn{1}{c}{}& \multicolumn{1}{c}{}& \multicolumn{1}{c}{}& \multicolumn{1}{c}{}& \multicolumn{1}{c}{}& \multicolumn{1}{c}{}& \multicolumn{1}{c}{}\\\cline{2-2} \cline{3-3} 
\multicolumn{1}{c|}{}& \multicolumn{1}{c}{0}& \multicolumn{1}{c|}{x}& \multicolumn{1}{c}{}& \multicolumn{1}{c}{}& \multicolumn{1}{c}{}& \multicolumn{1}{c}{}& \multicolumn{1}{c}{}& \multicolumn{1}{c}{}& \multicolumn{1}{c}{}\\\cline{2-2} \cline{3-3} 
\multicolumn{1}{c}{\vdots}& \multicolumn{1}{c}{\vdots}& \multicolumn{1}{c}{\vdots}& \multicolumn{1}{c}{}& \multicolumn{1}{c}{}& \multicolumn{1}{c}{}& \multicolumn{1}{c}{}& \multicolumn{1}{c}{}& \multicolumn{1}{c}{}& \multicolumn{1}{c}{}\\
\end{array}\right]
\]
\begin{align*}
f(B_1, \Pi, \Phi) &= N_{\text{value}} + N_{\text{index}}s\\
&= (5 + 6\mu_2 + 6\mu_3) + (2 + 4\mu_2 + 4\mu_3)s\\
&= 197 + 6\lceil 28s - 10\rceil + (130 + 4\lceil 28s - 10\rceil)s\\
&\leq 143 + 262s + 112s^2
\\&\leq 146 + 263s + 112s^2
\end{align*}
\paragraph{$\Pi = [1{:}1, 2{:}3, ...], \Phi = [1{:}2, 3{:}3, ...]$}
\[
\left[\begin{array}{cccccccccc}\cline{1-1} \cline{2-2} \cline{5-5} \cline{7-7} 
\multicolumn{1}{|c}{x}& \multicolumn{1}{c|}{0}& \multicolumn{1}{c}{}& \multicolumn{1}{c|}{\cdots}& \multicolumn{1}{c|}{x}& \multicolumn{1}{c|}{\cdots}& \multicolumn{1}{c|}{x}& \multicolumn{1}{c}{\cdots}& \multicolumn{1}{c}{}& \multicolumn{1}{c}{\cdots}\\\cline{1-1} \cline{2-2} \cline{3-3} \cline{5-5} \cline{7-7} \cline{9-9} 
\multicolumn{1}{|c}{0}& \multicolumn{1}{c|}{x}& \multicolumn{1}{c|}{0}& \multicolumn{1}{c|}{\cdots}& \multicolumn{1}{c|}{x}& \multicolumn{1}{c}{\cdots}& \multicolumn{1}{c}{}& \multicolumn{1}{c|}{\cdots}& \multicolumn{1}{c|}{0}& \multicolumn{1}{c}{\cdots}\\
\multicolumn{1}{|c}{0}& \multicolumn{1}{c|}{0}& \multicolumn{1}{c|}{x}& \multicolumn{1}{c|}{\cdots}& \multicolumn{1}{c|}{x}& \multicolumn{1}{c}{\cdots}& \multicolumn{1}{c}{}& \multicolumn{1}{c|}{\cdots}& \multicolumn{1}{c|}{x}& \multicolumn{1}{c}{\cdots}\\\cline{1-1} \cline{2-2} \cline{3-3} \cline{5-5} \cline{9-9} 
\multicolumn{1}{c}{\vdots}& \multicolumn{1}{c}{\vdots}& \multicolumn{1}{c}{\vdots}& \multicolumn{1}{c}{}& \multicolumn{1}{c}{}& \multicolumn{1}{c}{}& \multicolumn{1}{c}{}& \multicolumn{1}{c}{}& \multicolumn{1}{c}{}& \multicolumn{1}{c}{}\\\cline{1-1} \cline{2-2} \cline{3-3} 
\multicolumn{1}{|c}{x}& \multicolumn{1}{c|}{x}& \multicolumn{1}{c|}{x}& \multicolumn{1}{c}{}& \multicolumn{1}{c}{}& \multicolumn{1}{c}{}& \multicolumn{1}{c}{}& \multicolumn{1}{c}{}& \multicolumn{1}{c}{}& \multicolumn{1}{c}{}\\\cline{1-1} \cline{2-2} \cline{3-3} 
\multicolumn{1}{c}{\vdots}& \multicolumn{1}{c}{\vdots}& \multicolumn{1}{c}{\vdots}& \multicolumn{1}{c}{}& \multicolumn{1}{c}{}& \multicolumn{1}{c}{}& \multicolumn{1}{c}{}& \multicolumn{1}{c}{}& \multicolumn{1}{c}{}& \multicolumn{1}{c}{}\\\cline{1-1} \cline{2-2} 
\multicolumn{1}{|c}{x}& \multicolumn{1}{c|}{0}& \multicolumn{1}{c}{}& \multicolumn{1}{c}{}& \multicolumn{1}{c}{}& \multicolumn{1}{c}{}& \multicolumn{1}{c}{}& \multicolumn{1}{c}{}& \multicolumn{1}{c}{}& \multicolumn{1}{c}{}\\\cline{1-1} \cline{2-2} 
\multicolumn{1}{c}{\vdots}& \multicolumn{1}{c}{\vdots}& \multicolumn{1}{c}{\vdots}& \multicolumn{1}{c}{}& \multicolumn{1}{c}{}& \multicolumn{1}{c}{}& \multicolumn{1}{c}{}& \multicolumn{1}{c}{}& \multicolumn{1}{c}{}& \multicolumn{1}{c}{}\\\cline{3-3} 
\multicolumn{1}{c}{}& \multicolumn{1}{c|}{}& \multicolumn{1}{c|}{x}& \multicolumn{1}{c}{}& \multicolumn{1}{c}{}& \multicolumn{1}{c}{}& \multicolumn{1}{c}{}& \multicolumn{1}{c}{}& \multicolumn{1}{c}{}& \multicolumn{1}{c}{}\\\cline{3-3} 
\multicolumn{1}{c}{\vdots}& \multicolumn{1}{c}{\vdots}& \multicolumn{1}{c}{\vdots}& \multicolumn{1}{c}{}& \multicolumn{1}{c}{}& \multicolumn{1}{c}{}& \multicolumn{1}{c}{}& \multicolumn{1}{c}{}& \multicolumn{1}{c}{}& \multicolumn{1}{c}{}\\
\end{array}\right]
\]
\begin{align*}
f(B_1, \Pi, \Phi) &= N_{\text{value}} + N_{\text{index}}s\\
&= (8 + 6\mu_2 + 6\mu_3) + (3 + 4\mu_2 + 4\mu_3)s\\
&= 200 + 6\lceil 28s - 10\rceil + (131 + 4\lceil 28s - 10\rceil)s\\
&\leq 146 + 263s + 112s^2
\end{align*}
\paragraph{$\Pi = [1{:}2, 3{:}3, ...], \Phi = [1{:}1, 2{:}3, ...]$}
\[
\left[\begin{array}{cccccccccc}\cline{1-1} \cline{2-2} \cline{3-3} \cline{5-5} \cline{7-7} 
\multicolumn{1}{|c|}{x}& \multicolumn{1}{c}{0}& \multicolumn{1}{c|}{0}& \multicolumn{1}{c|}{\cdots}& \multicolumn{1}{c|}{x}& \multicolumn{1}{c|}{\cdots}& \multicolumn{1}{c|}{x}& \multicolumn{1}{c}{\cdots}& \multicolumn{1}{c}{}& \multicolumn{1}{c}{\cdots}\\
\multicolumn{1}{|c|}{0}& \multicolumn{1}{c}{x}& \multicolumn{1}{c|}{0}& \multicolumn{1}{c|}{\cdots}& \multicolumn{1}{c|}{x}& \multicolumn{1}{c|}{\cdots}& \multicolumn{1}{c|}{0}& \multicolumn{1}{c}{\cdots}& \multicolumn{1}{c}{}& \multicolumn{1}{c}{\cdots}\\\cline{1-1} \cline{2-2} \cline{3-3} \cline{5-5} \cline{7-7} \cline{9-9} 
\multicolumn{1}{c|}{}& \multicolumn{1}{c}{0}& \multicolumn{1}{c|}{x}& \multicolumn{1}{c|}{\cdots}& \multicolumn{1}{c|}{x}& \multicolumn{1}{c}{\cdots}& \multicolumn{1}{c}{}& \multicolumn{1}{c|}{\cdots}& \multicolumn{1}{c|}{x}& \multicolumn{1}{c}{\cdots}\\\cline{2-2} \cline{3-3} \cline{5-5} \cline{9-9} 
\multicolumn{1}{c}{\vdots}& \multicolumn{1}{c}{\vdots}& \multicolumn{1}{c}{\vdots}& \multicolumn{1}{c}{}& \multicolumn{1}{c}{}& \multicolumn{1}{c}{}& \multicolumn{1}{c}{}& \multicolumn{1}{c}{}& \multicolumn{1}{c}{}& \multicolumn{1}{c}{}\\\cline{1-1} \cline{2-2} \cline{3-3} 
\multicolumn{1}{|c|}{x}& \multicolumn{1}{c}{x}& \multicolumn{1}{c|}{x}& \multicolumn{1}{c}{}& \multicolumn{1}{c}{}& \multicolumn{1}{c}{}& \multicolumn{1}{c}{}& \multicolumn{1}{c}{}& \multicolumn{1}{c}{}& \multicolumn{1}{c}{}\\\cline{1-1} \cline{2-2} \cline{3-3} 
\multicolumn{1}{c}{\vdots}& \multicolumn{1}{c}{\vdots}& \multicolumn{1}{c}{\vdots}& \multicolumn{1}{c}{}& \multicolumn{1}{c}{}& \multicolumn{1}{c}{}& \multicolumn{1}{c}{}& \multicolumn{1}{c}{}& \multicolumn{1}{c}{}& \multicolumn{1}{c}{}\\\cline{1-1} 
\multicolumn{1}{|c|}{x}& \multicolumn{1}{c}{}& \multicolumn{1}{c}{}& \multicolumn{1}{c}{}& \multicolumn{1}{c}{}& \multicolumn{1}{c}{}& \multicolumn{1}{c}{}& \multicolumn{1}{c}{}& \multicolumn{1}{c}{}& \multicolumn{1}{c}{}\\\cline{1-1} 
\multicolumn{1}{c}{\vdots}& \multicolumn{1}{c}{\vdots}& \multicolumn{1}{c}{\vdots}& \multicolumn{1}{c}{}& \multicolumn{1}{c}{}& \multicolumn{1}{c}{}& \multicolumn{1}{c}{}& \multicolumn{1}{c}{}& \multicolumn{1}{c}{}& \multicolumn{1}{c}{}\\\cline{2-2} \cline{3-3} 
\multicolumn{1}{c|}{}& \multicolumn{1}{c}{0}& \multicolumn{1}{c|}{x}& \multicolumn{1}{c}{}& \multicolumn{1}{c}{}& \multicolumn{1}{c}{}& \multicolumn{1}{c}{}& \multicolumn{1}{c}{}& \multicolumn{1}{c}{}& \multicolumn{1}{c}{}\\\cline{2-2} \cline{3-3} 
\multicolumn{1}{c}{\vdots}& \multicolumn{1}{c}{\vdots}& \multicolumn{1}{c}{\vdots}& \multicolumn{1}{c}{}& \multicolumn{1}{c}{}& \multicolumn{1}{c}{}& \multicolumn{1}{c}{}& \multicolumn{1}{c}{}& \multicolumn{1}{c}{}& \multicolumn{1}{c}{}\\
\end{array}\right]
\]
\begin{align*}
f(B_1, \Pi, \Phi) &= N_{\text{value}} + N_{\text{index}}s\\
&= (8 + 6\mu_2 + 6\mu_3) + (3 + 4\mu_2 + 4\mu_3)s\\
&= 200 + 6\lceil 28s - 10\rceil + (131 + 4\lceil 28s - 10\rceil)s\\
&\leq 146 + 263s + 112s^2
\end{align*}
\paragraph{$\Pi = [1{:}2, 3{:}3, ...], \Phi = [1{:}2, 3{:}3, ...]$}
\[
\left[\begin{array}{cccccccccc}\cline{1-1} \cline{2-2} \cline{5-5} \cline{7-7} 
\multicolumn{1}{|c}{x}& \multicolumn{1}{c|}{0}& \multicolumn{1}{c}{}& \multicolumn{1}{c|}{\cdots}& \multicolumn{1}{c|}{x}& \multicolumn{1}{c|}{\cdots}& \multicolumn{1}{c|}{x}& \multicolumn{1}{c}{\cdots}& \multicolumn{1}{c}{}& \multicolumn{1}{c}{\cdots}\\
\multicolumn{1}{|c}{0}& \multicolumn{1}{c|}{x}& \multicolumn{1}{c}{}& \multicolumn{1}{c|}{\cdots}& \multicolumn{1}{c|}{x}& \multicolumn{1}{c|}{\cdots}& \multicolumn{1}{c|}{0}& \multicolumn{1}{c}{\cdots}& \multicolumn{1}{c}{}& \multicolumn{1}{c}{\cdots}\\\cline{1-1} \cline{2-2} \cline{3-3} \cline{5-5} \cline{7-7} \cline{9-9} 
\multicolumn{1}{c}{}& \multicolumn{1}{c|}{}& \multicolumn{1}{c|}{x}& \multicolumn{1}{c|}{\cdots}& \multicolumn{1}{c|}{x}& \multicolumn{1}{c}{\cdots}& \multicolumn{1}{c}{}& \multicolumn{1}{c|}{\cdots}& \multicolumn{1}{c|}{x}& \multicolumn{1}{c}{\cdots}\\\cline{3-3} \cline{5-5} \cline{9-9} 
\multicolumn{1}{c}{\vdots}& \multicolumn{1}{c}{\vdots}& \multicolumn{1}{c}{\vdots}& \multicolumn{1}{c}{}& \multicolumn{1}{c}{}& \multicolumn{1}{c}{}& \multicolumn{1}{c}{}& \multicolumn{1}{c}{}& \multicolumn{1}{c}{}& \multicolumn{1}{c}{}\\\cline{1-1} \cline{2-2} \cline{3-3} 
\multicolumn{1}{|c}{x}& \multicolumn{1}{c|}{x}& \multicolumn{1}{c|}{x}& \multicolumn{1}{c}{}& \multicolumn{1}{c}{}& \multicolumn{1}{c}{}& \multicolumn{1}{c}{}& \multicolumn{1}{c}{}& \multicolumn{1}{c}{}& \multicolumn{1}{c}{}\\\cline{1-1} \cline{2-2} \cline{3-3} 
\multicolumn{1}{c}{\vdots}& \multicolumn{1}{c}{\vdots}& \multicolumn{1}{c}{\vdots}& \multicolumn{1}{c}{}& \multicolumn{1}{c}{}& \multicolumn{1}{c}{}& \multicolumn{1}{c}{}& \multicolumn{1}{c}{}& \multicolumn{1}{c}{}& \multicolumn{1}{c}{}\\\cline{1-1} \cline{2-2} 
\multicolumn{1}{|c}{x}& \multicolumn{1}{c|}{0}& \multicolumn{1}{c}{}& \multicolumn{1}{c}{}& \multicolumn{1}{c}{}& \multicolumn{1}{c}{}& \multicolumn{1}{c}{}& \multicolumn{1}{c}{}& \multicolumn{1}{c}{}& \multicolumn{1}{c}{}\\\cline{1-1} \cline{2-2} 
\multicolumn{1}{c}{\vdots}& \multicolumn{1}{c}{\vdots}& \multicolumn{1}{c}{\vdots}& \multicolumn{1}{c}{}& \multicolumn{1}{c}{}& \multicolumn{1}{c}{}& \multicolumn{1}{c}{}& \multicolumn{1}{c}{}& \multicolumn{1}{c}{}& \multicolumn{1}{c}{}\\\cline{3-3} 
\multicolumn{1}{c}{}& \multicolumn{1}{c|}{}& \multicolumn{1}{c|}{x}& \multicolumn{1}{c}{}& \multicolumn{1}{c}{}& \multicolumn{1}{c}{}& \multicolumn{1}{c}{}& \multicolumn{1}{c}{}& \multicolumn{1}{c}{}& \multicolumn{1}{c}{}\\\cline{3-3} 
\multicolumn{1}{c}{\vdots}& \multicolumn{1}{c}{\vdots}& \multicolumn{1}{c}{\vdots}& \multicolumn{1}{c}{}& \multicolumn{1}{c}{}& \multicolumn{1}{c}{}& \multicolumn{1}{c}{}& \multicolumn{1}{c}{}& \multicolumn{1}{c}{}& \multicolumn{1}{c}{}\\
\end{array}\right]
\]
\begin{align*}
f(B_1, \Pi, \Phi) &= N_{\text{value}} + N_{\text{index}}s\\
&= (5 + 6\mu_2 + 6\mu_3) + (2 + 4\mu_2 + 4\mu_3)s\\
&= 197 + 6\lceil 28s - 10\rceil + (130 + 4\lceil 28s - 10\rceil)s\\
&\leq 143 + 262s + 112s^2
\\&\leq 146 + 263s + 112s^2
\end{align*}
\subsection{Exhaustively Checking \eqref{eq:sadgadget}}\label{app:vbrblockingnphard:sadgadget}
\paragraph{$\Pi = [1{:}3, ...], \Phi = [1{:}3, ...]$}
\[
\left[\begin{array}{cccccccccc}\cline{1-1} \cline{2-2} \cline{3-3} \cline{5-5} \cline{7-7} \cline{9-9} 
\multicolumn{1}{|c}{x}& \multicolumn{1}{c}{0}& \multicolumn{1}{c|}{0}& \multicolumn{1}{c|}{\cdots}& \multicolumn{1}{c|}{x}& \multicolumn{1}{c|}{\cdots}& \multicolumn{1}{c|}{x}& \multicolumn{1}{c|}{\cdots}& \multicolumn{1}{c|}{0}& \multicolumn{1}{c}{\cdots}\\
\multicolumn{1}{|c}{0}& \multicolumn{1}{c}{x}& \multicolumn{1}{c|}{0}& \multicolumn{1}{c|}{\cdots}& \multicolumn{1}{c|}{x}& \multicolumn{1}{c|}{\cdots}& \multicolumn{1}{c|}{0}& \multicolumn{1}{c|}{\cdots}& \multicolumn{1}{c|}{0}& \multicolumn{1}{c}{\cdots}\\
\multicolumn{1}{|c}{0}& \multicolumn{1}{c}{0}& \multicolumn{1}{c|}{x}& \multicolumn{1}{c|}{\cdots}& \multicolumn{1}{c|}{x}& \multicolumn{1}{c|}{\cdots}& \multicolumn{1}{c|}{0}& \multicolumn{1}{c|}{\cdots}& \multicolumn{1}{c|}{x}& \multicolumn{1}{c}{\cdots}\\\cline{1-1} \cline{2-2} \cline{3-3} \cline{5-5} \cline{7-7} \cline{9-9} 
\multicolumn{1}{c}{\vdots}& \multicolumn{1}{c}{\vdots}& \multicolumn{1}{c}{\vdots}& \multicolumn{1}{c}{}& \multicolumn{1}{c}{}& \multicolumn{1}{c}{}& \multicolumn{1}{c}{}& \multicolumn{1}{c}{}& \multicolumn{1}{c}{}& \multicolumn{1}{c}{}\\\cline{1-1} \cline{2-2} \cline{3-3} 
\multicolumn{1}{|c}{x}& \multicolumn{1}{c}{x}& \multicolumn{1}{c|}{x}& \multicolumn{1}{c}{}& \multicolumn{1}{c}{}& \multicolumn{1}{c}{}& \multicolumn{1}{c}{}& \multicolumn{1}{c}{}& \multicolumn{1}{c}{}& \multicolumn{1}{c}{}\\\cline{1-1} \cline{2-2} \cline{3-3} 
\multicolumn{1}{c}{\vdots}& \multicolumn{1}{c}{\vdots}& \multicolumn{1}{c}{\vdots}& \multicolumn{1}{c}{}& \multicolumn{1}{c}{}& \multicolumn{1}{c}{}& \multicolumn{1}{c}{}& \multicolumn{1}{c}{}& \multicolumn{1}{c}{}& \multicolumn{1}{c}{}\\\cline{1-1} \cline{2-2} \cline{3-3} 
\multicolumn{1}{|c}{x}& \multicolumn{1}{c}{0}& \multicolumn{1}{c|}{0}& \multicolumn{1}{c}{}& \multicolumn{1}{c}{}& \multicolumn{1}{c}{}& \multicolumn{1}{c}{}& \multicolumn{1}{c}{}& \multicolumn{1}{c}{}& \multicolumn{1}{c}{}\\\cline{1-1} \cline{2-2} \cline{3-3} 
\multicolumn{1}{c}{\vdots}& \multicolumn{1}{c}{\vdots}& \multicolumn{1}{c}{\vdots}& \multicolumn{1}{c}{}& \multicolumn{1}{c}{}& \multicolumn{1}{c}{}& \multicolumn{1}{c}{}& \multicolumn{1}{c}{}& \multicolumn{1}{c}{}& \multicolumn{1}{c}{}\\\cline{1-1} \cline{2-2} \cline{3-3} 
\multicolumn{1}{|c}{0}& \multicolumn{1}{c}{0}& \multicolumn{1}{c|}{x}& \multicolumn{1}{c}{}& \multicolumn{1}{c}{}& \multicolumn{1}{c}{}& \multicolumn{1}{c}{}& \multicolumn{1}{c}{}& \multicolumn{1}{c}{}& \multicolumn{1}{c}{}\\\cline{1-1} \cline{2-2} \cline{3-3} 
\multicolumn{1}{c}{\vdots}& \multicolumn{1}{c}{\vdots}& \multicolumn{1}{c}{\vdots}& \multicolumn{1}{c}{}& \multicolumn{1}{c}{}& \multicolumn{1}{c}{}& \multicolumn{1}{c}{}& \multicolumn{1}{c}{}& \multicolumn{1}{c}{}& \multicolumn{1}{c}{}\\
\end{array}\right]
\]
\begin{align*}
f(B_1, \Pi, \Phi) &= N_{\text{value}} + N_{\text{index}}s\\
&= (9 + 6\mu_2 + 12\mu_3) + (1 + 2\mu_2 + 4\mu_3)s\\
&= 201 + 12\lceil 28s - 10\rceil + (65 + 4\lceil 28s - 10\rceil)s\\
&\geq 81 + 361s + 112s^2
\\&\geq 147 + 263s + 112s^2
\end{align*}
\paragraph{$\Pi = [1{:}3, ...], \Phi = [1{:}1, 2{:}3, ...]$}
\[
\left[\begin{array}{cccccccccc}\cline{1-1} \cline{2-2} \cline{3-3} \cline{5-5} \cline{7-7} \cline{9-9} 
\multicolumn{1}{|c|}{x}& \multicolumn{1}{c}{0}& \multicolumn{1}{c|}{0}& \multicolumn{1}{c|}{\cdots}& \multicolumn{1}{c|}{x}& \multicolumn{1}{c|}{\cdots}& \multicolumn{1}{c|}{x}& \multicolumn{1}{c|}{\cdots}& \multicolumn{1}{c|}{0}& \multicolumn{1}{c}{\cdots}\\
\multicolumn{1}{|c|}{0}& \multicolumn{1}{c}{x}& \multicolumn{1}{c|}{0}& \multicolumn{1}{c|}{\cdots}& \multicolumn{1}{c|}{x}& \multicolumn{1}{c|}{\cdots}& \multicolumn{1}{c|}{0}& \multicolumn{1}{c|}{\cdots}& \multicolumn{1}{c|}{0}& \multicolumn{1}{c}{\cdots}\\
\multicolumn{1}{|c|}{0}& \multicolumn{1}{c}{0}& \multicolumn{1}{c|}{x}& \multicolumn{1}{c|}{\cdots}& \multicolumn{1}{c|}{x}& \multicolumn{1}{c|}{\cdots}& \multicolumn{1}{c|}{0}& \multicolumn{1}{c|}{\cdots}& \multicolumn{1}{c|}{x}& \multicolumn{1}{c}{\cdots}\\\cline{1-1} \cline{2-2} \cline{3-3} \cline{5-5} \cline{7-7} \cline{9-9} 
\multicolumn{1}{c}{\vdots}& \multicolumn{1}{c}{\vdots}& \multicolumn{1}{c}{\vdots}& \multicolumn{1}{c}{}& \multicolumn{1}{c}{}& \multicolumn{1}{c}{}& \multicolumn{1}{c}{}& \multicolumn{1}{c}{}& \multicolumn{1}{c}{}& \multicolumn{1}{c}{}\\\cline{1-1} \cline{2-2} \cline{3-3} 
\multicolumn{1}{|c|}{x}& \multicolumn{1}{c}{x}& \multicolumn{1}{c|}{x}& \multicolumn{1}{c}{}& \multicolumn{1}{c}{}& \multicolumn{1}{c}{}& \multicolumn{1}{c}{}& \multicolumn{1}{c}{}& \multicolumn{1}{c}{}& \multicolumn{1}{c}{}\\\cline{1-1} \cline{2-2} \cline{3-3} 
\multicolumn{1}{c}{\vdots}& \multicolumn{1}{c}{\vdots}& \multicolumn{1}{c}{\vdots}& \multicolumn{1}{c}{}& \multicolumn{1}{c}{}& \multicolumn{1}{c}{}& \multicolumn{1}{c}{}& \multicolumn{1}{c}{}& \multicolumn{1}{c}{}& \multicolumn{1}{c}{}\\\cline{1-1} 
\multicolumn{1}{|c|}{x}& \multicolumn{1}{c}{}& \multicolumn{1}{c}{}& \multicolumn{1}{c}{}& \multicolumn{1}{c}{}& \multicolumn{1}{c}{}& \multicolumn{1}{c}{}& \multicolumn{1}{c}{}& \multicolumn{1}{c}{}& \multicolumn{1}{c}{}\\\cline{1-1} 
\multicolumn{1}{c}{\vdots}& \multicolumn{1}{c}{\vdots}& \multicolumn{1}{c}{\vdots}& \multicolumn{1}{c}{}& \multicolumn{1}{c}{}& \multicolumn{1}{c}{}& \multicolumn{1}{c}{}& \multicolumn{1}{c}{}& \multicolumn{1}{c}{}& \multicolumn{1}{c}{}\\\cline{2-2} \cline{3-3} 
\multicolumn{1}{c|}{}& \multicolumn{1}{c}{0}& \multicolumn{1}{c|}{x}& \multicolumn{1}{c}{}& \multicolumn{1}{c}{}& \multicolumn{1}{c}{}& \multicolumn{1}{c}{}& \multicolumn{1}{c}{}& \multicolumn{1}{c}{}& \multicolumn{1}{c}{}\\\cline{2-2} \cline{3-3} 
\multicolumn{1}{c}{\vdots}& \multicolumn{1}{c}{\vdots}& \multicolumn{1}{c}{\vdots}& \multicolumn{1}{c}{}& \multicolumn{1}{c}{}& \multicolumn{1}{c}{}& \multicolumn{1}{c}{}& \multicolumn{1}{c}{}& \multicolumn{1}{c}{}& \multicolumn{1}{c}{}\\
\end{array}\right]
\]
\begin{align*}
f(B_1, \Pi, \Phi) &= N_{\text{value}} + N_{\text{index}}s\\
&= (9 + 6\mu_2 + 9\mu_3) + (2 + 3\mu_2 + 4\mu_3)s\\
&= 201 + 9\lceil 28s - 10\rceil + (98 + 4\lceil 28s - 10\rceil)s\\
&\geq 111 + 310s + 112s^2
\\&\geq 147 + 263s + 112s^2
\end{align*}
\paragraph{$\Pi = [1{:}3, ...], \Phi = [1{:}2, 3{:}3, ...]$}
\[
\left[\begin{array}{cccccccccc}\cline{1-1} \cline{2-2} \cline{3-3} \cline{5-5} \cline{7-7} \cline{9-9} 
\multicolumn{1}{|c}{x}& \multicolumn{1}{c|}{0}& \multicolumn{1}{c|}{0}& \multicolumn{1}{c|}{\cdots}& \multicolumn{1}{c|}{x}& \multicolumn{1}{c|}{\cdots}& \multicolumn{1}{c|}{x}& \multicolumn{1}{c|}{\cdots}& \multicolumn{1}{c|}{0}& \multicolumn{1}{c}{\cdots}\\
\multicolumn{1}{|c}{0}& \multicolumn{1}{c|}{x}& \multicolumn{1}{c|}{0}& \multicolumn{1}{c|}{\cdots}& \multicolumn{1}{c|}{x}& \multicolumn{1}{c|}{\cdots}& \multicolumn{1}{c|}{0}& \multicolumn{1}{c|}{\cdots}& \multicolumn{1}{c|}{0}& \multicolumn{1}{c}{\cdots}\\
\multicolumn{1}{|c}{0}& \multicolumn{1}{c|}{0}& \multicolumn{1}{c|}{x}& \multicolumn{1}{c|}{\cdots}& \multicolumn{1}{c|}{x}& \multicolumn{1}{c|}{\cdots}& \multicolumn{1}{c|}{0}& \multicolumn{1}{c|}{\cdots}& \multicolumn{1}{c|}{x}& \multicolumn{1}{c}{\cdots}\\\cline{1-1} \cline{2-2} \cline{3-3} \cline{5-5} \cline{7-7} \cline{9-9} 
\multicolumn{1}{c}{\vdots}& \multicolumn{1}{c}{\vdots}& \multicolumn{1}{c}{\vdots}& \multicolumn{1}{c}{}& \multicolumn{1}{c}{}& \multicolumn{1}{c}{}& \multicolumn{1}{c}{}& \multicolumn{1}{c}{}& \multicolumn{1}{c}{}& \multicolumn{1}{c}{}\\\cline{1-1} \cline{2-2} \cline{3-3} 
\multicolumn{1}{|c}{x}& \multicolumn{1}{c|}{x}& \multicolumn{1}{c|}{x}& \multicolumn{1}{c}{}& \multicolumn{1}{c}{}& \multicolumn{1}{c}{}& \multicolumn{1}{c}{}& \multicolumn{1}{c}{}& \multicolumn{1}{c}{}& \multicolumn{1}{c}{}\\\cline{1-1} \cline{2-2} \cline{3-3} 
\multicolumn{1}{c}{\vdots}& \multicolumn{1}{c}{\vdots}& \multicolumn{1}{c}{\vdots}& \multicolumn{1}{c}{}& \multicolumn{1}{c}{}& \multicolumn{1}{c}{}& \multicolumn{1}{c}{}& \multicolumn{1}{c}{}& \multicolumn{1}{c}{}& \multicolumn{1}{c}{}\\\cline{1-1} \cline{2-2} 
\multicolumn{1}{|c}{x}& \multicolumn{1}{c|}{0}& \multicolumn{1}{c}{}& \multicolumn{1}{c}{}& \multicolumn{1}{c}{}& \multicolumn{1}{c}{}& \multicolumn{1}{c}{}& \multicolumn{1}{c}{}& \multicolumn{1}{c}{}& \multicolumn{1}{c}{}\\\cline{1-1} \cline{2-2} 
\multicolumn{1}{c}{\vdots}& \multicolumn{1}{c}{\vdots}& \multicolumn{1}{c}{\vdots}& \multicolumn{1}{c}{}& \multicolumn{1}{c}{}& \multicolumn{1}{c}{}& \multicolumn{1}{c}{}& \multicolumn{1}{c}{}& \multicolumn{1}{c}{}& \multicolumn{1}{c}{}\\\cline{3-3} 
\multicolumn{1}{c}{}& \multicolumn{1}{c|}{}& \multicolumn{1}{c|}{x}& \multicolumn{1}{c}{}& \multicolumn{1}{c}{}& \multicolumn{1}{c}{}& \multicolumn{1}{c}{}& \multicolumn{1}{c}{}& \multicolumn{1}{c}{}& \multicolumn{1}{c}{}\\\cline{3-3} 
\multicolumn{1}{c}{\vdots}& \multicolumn{1}{c}{\vdots}& \multicolumn{1}{c}{\vdots}& \multicolumn{1}{c}{}& \multicolumn{1}{c}{}& \multicolumn{1}{c}{}& \multicolumn{1}{c}{}& \multicolumn{1}{c}{}& \multicolumn{1}{c}{}& \multicolumn{1}{c}{}\\
\end{array}\right]
\]
\begin{align*}
f(B_1, \Pi, \Phi) &= N_{\text{value}} + N_{\text{index}}s\\
&= (9 + 6\mu_2 + 9\mu_3) + (2 + 3\mu_2 + 4\mu_3)s\\
&= 201 + 9\lceil 28s - 10\rceil + (98 + 4\lceil 28s - 10\rceil)s\\
&\geq 111 + 310s + 112s^2
\\&\geq 147 + 263s + 112s^2
\end{align*}
\paragraph{$\Pi = [1{:}3, ...], \Phi = [1{:}1, 2{:}2, 3{:}3, ...]$}
\[
\left[\begin{array}{cccccccccc}\cline{1-1} \cline{2-2} \cline{3-3} \cline{5-5} \cline{7-7} \cline{9-9} 
\multicolumn{1}{|c|}{x}& \multicolumn{1}{c|}{0}& \multicolumn{1}{c|}{0}& \multicolumn{1}{c|}{\cdots}& \multicolumn{1}{c|}{x}& \multicolumn{1}{c|}{\cdots}& \multicolumn{1}{c|}{x}& \multicolumn{1}{c|}{\cdots}& \multicolumn{1}{c|}{0}& \multicolumn{1}{c}{\cdots}\\
\multicolumn{1}{|c|}{0}& \multicolumn{1}{c|}{x}& \multicolumn{1}{c|}{0}& \multicolumn{1}{c|}{\cdots}& \multicolumn{1}{c|}{x}& \multicolumn{1}{c|}{\cdots}& \multicolumn{1}{c|}{0}& \multicolumn{1}{c|}{\cdots}& \multicolumn{1}{c|}{0}& \multicolumn{1}{c}{\cdots}\\
\multicolumn{1}{|c|}{0}& \multicolumn{1}{c|}{0}& \multicolumn{1}{c|}{x}& \multicolumn{1}{c|}{\cdots}& \multicolumn{1}{c|}{x}& \multicolumn{1}{c|}{\cdots}& \multicolumn{1}{c|}{0}& \multicolumn{1}{c|}{\cdots}& \multicolumn{1}{c|}{x}& \multicolumn{1}{c}{\cdots}\\\cline{1-1} \cline{2-2} \cline{3-3} \cline{5-5} \cline{7-7} \cline{9-9} 
\multicolumn{1}{c}{\vdots}& \multicolumn{1}{c}{\vdots}& \multicolumn{1}{c}{\vdots}& \multicolumn{1}{c}{}& \multicolumn{1}{c}{}& \multicolumn{1}{c}{}& \multicolumn{1}{c}{}& \multicolumn{1}{c}{}& \multicolumn{1}{c}{}& \multicolumn{1}{c}{}\\\cline{1-1} \cline{2-2} \cline{3-3} 
\multicolumn{1}{|c|}{x}& \multicolumn{1}{c|}{x}& \multicolumn{1}{c|}{x}& \multicolumn{1}{c}{}& \multicolumn{1}{c}{}& \multicolumn{1}{c}{}& \multicolumn{1}{c}{}& \multicolumn{1}{c}{}& \multicolumn{1}{c}{}& \multicolumn{1}{c}{}\\\cline{1-1} \cline{2-2} \cline{3-3} 
\multicolumn{1}{c}{\vdots}& \multicolumn{1}{c}{\vdots}& \multicolumn{1}{c}{\vdots}& \multicolumn{1}{c}{}& \multicolumn{1}{c}{}& \multicolumn{1}{c}{}& \multicolumn{1}{c}{}& \multicolumn{1}{c}{}& \multicolumn{1}{c}{}& \multicolumn{1}{c}{}\\\cline{1-1} 
\multicolumn{1}{|c|}{x}& \multicolumn{1}{c}{}& \multicolumn{1}{c}{}& \multicolumn{1}{c}{}& \multicolumn{1}{c}{}& \multicolumn{1}{c}{}& \multicolumn{1}{c}{}& \multicolumn{1}{c}{}& \multicolumn{1}{c}{}& \multicolumn{1}{c}{}\\\cline{1-1} 
\multicolumn{1}{c}{\vdots}& \multicolumn{1}{c}{\vdots}& \multicolumn{1}{c}{\vdots}& \multicolumn{1}{c}{}& \multicolumn{1}{c}{}& \multicolumn{1}{c}{}& \multicolumn{1}{c}{}& \multicolumn{1}{c}{}& \multicolumn{1}{c}{}& \multicolumn{1}{c}{}\\\cline{3-3} 
\multicolumn{1}{c}{}& \multicolumn{1}{c|}{}& \multicolumn{1}{c|}{x}& \multicolumn{1}{c}{}& \multicolumn{1}{c}{}& \multicolumn{1}{c}{}& \multicolumn{1}{c}{}& \multicolumn{1}{c}{}& \multicolumn{1}{c}{}& \multicolumn{1}{c}{}\\\cline{3-3} 
\multicolumn{1}{c}{\vdots}& \multicolumn{1}{c}{\vdots}& \multicolumn{1}{c}{\vdots}& \multicolumn{1}{c}{}& \multicolumn{1}{c}{}& \multicolumn{1}{c}{}& \multicolumn{1}{c}{}& \multicolumn{1}{c}{}& \multicolumn{1}{c}{}& \multicolumn{1}{c}{}\\
\end{array}\right]
\]
\begin{align*}
f(B_1, \Pi, \Phi) &= N_{\text{value}} + N_{\text{index}}s\\
&= (9 + 6\mu_2 + 8\mu_3) + (3 + 4\mu_2 + 4\mu_3)s\\
&= 201 + 8\lceil 28s - 10\rceil + (131 + 4\lceil 28s - 10\rceil)s\\
&\geq 121 + 315s + 112s^2
\\&\geq 147 + 263s + 112s^2
\end{align*}
\paragraph{$\Pi = [1{:}1, 2{:}3, ...], \Phi = [1{:}3, ...]$}
\[
\left[\begin{array}{cccccccccc}\cline{1-1} \cline{2-2} \cline{3-3} \cline{5-5} \cline{7-7} 
\multicolumn{1}{|c}{x}& \multicolumn{1}{c}{0}& \multicolumn{1}{c|}{0}& \multicolumn{1}{c|}{\cdots}& \multicolumn{1}{c|}{x}& \multicolumn{1}{c|}{\cdots}& \multicolumn{1}{c|}{x}& \multicolumn{1}{c}{\cdots}& \multicolumn{1}{c}{}& \multicolumn{1}{c}{\cdots}\\\cline{1-1} \cline{2-2} \cline{3-3} \cline{5-5} \cline{7-7} \cline{9-9} 
\multicolumn{1}{|c}{0}& \multicolumn{1}{c}{x}& \multicolumn{1}{c|}{0}& \multicolumn{1}{c|}{\cdots}& \multicolumn{1}{c|}{x}& \multicolumn{1}{c}{\cdots}& \multicolumn{1}{c}{}& \multicolumn{1}{c|}{\cdots}& \multicolumn{1}{c|}{0}& \multicolumn{1}{c}{\cdots}\\
\multicolumn{1}{|c}{0}& \multicolumn{1}{c}{0}& \multicolumn{1}{c|}{x}& \multicolumn{1}{c|}{\cdots}& \multicolumn{1}{c|}{x}& \multicolumn{1}{c}{\cdots}& \multicolumn{1}{c}{}& \multicolumn{1}{c|}{\cdots}& \multicolumn{1}{c|}{x}& \multicolumn{1}{c}{\cdots}\\\cline{1-1} \cline{2-2} \cline{3-3} \cline{5-5} \cline{9-9} 
\multicolumn{1}{c}{\vdots}& \multicolumn{1}{c}{\vdots}& \multicolumn{1}{c}{\vdots}& \multicolumn{1}{c}{}& \multicolumn{1}{c}{}& \multicolumn{1}{c}{}& \multicolumn{1}{c}{}& \multicolumn{1}{c}{}& \multicolumn{1}{c}{}& \multicolumn{1}{c}{}\\\cline{1-1} \cline{2-2} \cline{3-3} 
\multicolumn{1}{|c}{x}& \multicolumn{1}{c}{x}& \multicolumn{1}{c|}{x}& \multicolumn{1}{c}{}& \multicolumn{1}{c}{}& \multicolumn{1}{c}{}& \multicolumn{1}{c}{}& \multicolumn{1}{c}{}& \multicolumn{1}{c}{}& \multicolumn{1}{c}{}\\\cline{1-1} \cline{2-2} \cline{3-3} 
\multicolumn{1}{c}{\vdots}& \multicolumn{1}{c}{\vdots}& \multicolumn{1}{c}{\vdots}& \multicolumn{1}{c}{}& \multicolumn{1}{c}{}& \multicolumn{1}{c}{}& \multicolumn{1}{c}{}& \multicolumn{1}{c}{}& \multicolumn{1}{c}{}& \multicolumn{1}{c}{}\\\cline{1-1} \cline{2-2} \cline{3-3} 
\multicolumn{1}{|c}{x}& \multicolumn{1}{c}{0}& \multicolumn{1}{c|}{0}& \multicolumn{1}{c}{}& \multicolumn{1}{c}{}& \multicolumn{1}{c}{}& \multicolumn{1}{c}{}& \multicolumn{1}{c}{}& \multicolumn{1}{c}{}& \multicolumn{1}{c}{}\\\cline{1-1} \cline{2-2} \cline{3-3} 
\multicolumn{1}{c}{\vdots}& \multicolumn{1}{c}{\vdots}& \multicolumn{1}{c}{\vdots}& \multicolumn{1}{c}{}& \multicolumn{1}{c}{}& \multicolumn{1}{c}{}& \multicolumn{1}{c}{}& \multicolumn{1}{c}{}& \multicolumn{1}{c}{}& \multicolumn{1}{c}{}\\\cline{1-1} \cline{2-2} \cline{3-3} 
\multicolumn{1}{|c}{0}& \multicolumn{1}{c}{0}& \multicolumn{1}{c|}{x}& \multicolumn{1}{c}{}& \multicolumn{1}{c}{}& \multicolumn{1}{c}{}& \multicolumn{1}{c}{}& \multicolumn{1}{c}{}& \multicolumn{1}{c}{}& \multicolumn{1}{c}{}\\\cline{1-1} \cline{2-2} \cline{3-3} 
\multicolumn{1}{c}{\vdots}& \multicolumn{1}{c}{\vdots}& \multicolumn{1}{c}{\vdots}& \multicolumn{1}{c}{}& \multicolumn{1}{c}{}& \multicolumn{1}{c}{}& \multicolumn{1}{c}{}& \multicolumn{1}{c}{}& \multicolumn{1}{c}{}& \multicolumn{1}{c}{}\\
\end{array}\right]
\]
\begin{align*}
f(B_1, \Pi, \Phi) &= N_{\text{value}} + N_{\text{index}}s\\
&= (9 + 6\mu_2 + 9\mu_3) + (2 + 3\mu_2 + 4\mu_3)s\\
&= 201 + 9\lceil 28s - 10\rceil + (98 + 4\lceil 28s - 10\rceil)s\\
&\geq 111 + 310s + 112s^2
\\&\geq 147 + 263s + 112s^2
\end{align*}
\paragraph{$\Pi = [1{:}1, 2{:}3, ...], \Phi = [1{:}1, 2{:}2, 3{:}3, ...]$}
\[
\left[\begin{array}{cccccccccc}\cline{1-1} \cline{5-5} \cline{7-7} 
\multicolumn{1}{|c|}{x}& \multicolumn{1}{c}{}& \multicolumn{1}{c}{}& \multicolumn{1}{c|}{\cdots}& \multicolumn{1}{c|}{x}& \multicolumn{1}{c|}{\cdots}& \multicolumn{1}{c|}{x}& \multicolumn{1}{c}{\cdots}& \multicolumn{1}{c}{}& \multicolumn{1}{c}{\cdots}\\\cline{1-1} \cline{2-2} \cline{3-3} \cline{5-5} \cline{7-7} \cline{9-9} 
\multicolumn{1}{c|}{}& \multicolumn{1}{c|}{x}& \multicolumn{1}{c|}{0}& \multicolumn{1}{c|}{\cdots}& \multicolumn{1}{c|}{x}& \multicolumn{1}{c}{\cdots}& \multicolumn{1}{c}{}& \multicolumn{1}{c|}{\cdots}& \multicolumn{1}{c|}{0}& \multicolumn{1}{c}{\cdots}\\
\multicolumn{1}{c|}{}& \multicolumn{1}{c|}{0}& \multicolumn{1}{c|}{x}& \multicolumn{1}{c|}{\cdots}& \multicolumn{1}{c|}{x}& \multicolumn{1}{c}{\cdots}& \multicolumn{1}{c}{}& \multicolumn{1}{c|}{\cdots}& \multicolumn{1}{c|}{x}& \multicolumn{1}{c}{\cdots}\\\cline{2-2} \cline{3-3} \cline{5-5} \cline{9-9} 
\multicolumn{1}{c}{\vdots}& \multicolumn{1}{c}{\vdots}& \multicolumn{1}{c}{\vdots}& \multicolumn{1}{c}{}& \multicolumn{1}{c}{}& \multicolumn{1}{c}{}& \multicolumn{1}{c}{}& \multicolumn{1}{c}{}& \multicolumn{1}{c}{}& \multicolumn{1}{c}{}\\\cline{1-1} \cline{2-2} \cline{3-3} 
\multicolumn{1}{|c|}{x}& \multicolumn{1}{c|}{x}& \multicolumn{1}{c|}{x}& \multicolumn{1}{c}{}& \multicolumn{1}{c}{}& \multicolumn{1}{c}{}& \multicolumn{1}{c}{}& \multicolumn{1}{c}{}& \multicolumn{1}{c}{}& \multicolumn{1}{c}{}\\\cline{1-1} \cline{2-2} \cline{3-3} 
\multicolumn{1}{c}{\vdots}& \multicolumn{1}{c}{\vdots}& \multicolumn{1}{c}{\vdots}& \multicolumn{1}{c}{}& \multicolumn{1}{c}{}& \multicolumn{1}{c}{}& \multicolumn{1}{c}{}& \multicolumn{1}{c}{}& \multicolumn{1}{c}{}& \multicolumn{1}{c}{}\\\cline{1-1} 
\multicolumn{1}{|c|}{x}& \multicolumn{1}{c}{}& \multicolumn{1}{c}{}& \multicolumn{1}{c}{}& \multicolumn{1}{c}{}& \multicolumn{1}{c}{}& \multicolumn{1}{c}{}& \multicolumn{1}{c}{}& \multicolumn{1}{c}{}& \multicolumn{1}{c}{}\\\cline{1-1} 
\multicolumn{1}{c}{\vdots}& \multicolumn{1}{c}{\vdots}& \multicolumn{1}{c}{\vdots}& \multicolumn{1}{c}{}& \multicolumn{1}{c}{}& \multicolumn{1}{c}{}& \multicolumn{1}{c}{}& \multicolumn{1}{c}{}& \multicolumn{1}{c}{}& \multicolumn{1}{c}{}\\\cline{3-3} 
\multicolumn{1}{c}{}& \multicolumn{1}{c|}{}& \multicolumn{1}{c|}{x}& \multicolumn{1}{c}{}& \multicolumn{1}{c}{}& \multicolumn{1}{c}{}& \multicolumn{1}{c}{}& \multicolumn{1}{c}{}& \multicolumn{1}{c}{}& \multicolumn{1}{c}{}\\\cline{3-3} 
\multicolumn{1}{c}{\vdots}& \multicolumn{1}{c}{\vdots}& \multicolumn{1}{c}{\vdots}& \multicolumn{1}{c}{}& \multicolumn{1}{c}{}& \multicolumn{1}{c}{}& \multicolumn{1}{c}{}& \multicolumn{1}{c}{}& \multicolumn{1}{c}{}& \multicolumn{1}{c}{}\\
\end{array}\right]
\]
\begin{align*}
f(B_1, \Pi, \Phi) &= N_{\text{value}} + N_{\text{index}}s\\
&= (5 + 6\mu_2 + 5\mu_3) + (3 + 5\mu_2 + 4\mu_3)s\\
&= 197 + 5\lceil 28s - 10\rceil + (163 + 4\lceil 28s - 10\rceil)s\\
&\geq 147 + 263s + 112s^2
\end{align*}
\paragraph{$\Pi = [1{:}2, 3{:}3, ...], \Phi = [1{:}3, ...]$}
\[
\left[\begin{array}{cccccccccc}\cline{1-1} \cline{2-2} \cline{3-3} \cline{5-5} \cline{7-7} 
\multicolumn{1}{|c}{x}& \multicolumn{1}{c}{0}& \multicolumn{1}{c|}{0}& \multicolumn{1}{c|}{\cdots}& \multicolumn{1}{c|}{x}& \multicolumn{1}{c|}{\cdots}& \multicolumn{1}{c|}{x}& \multicolumn{1}{c}{\cdots}& \multicolumn{1}{c}{}& \multicolumn{1}{c}{\cdots}\\
\multicolumn{1}{|c}{0}& \multicolumn{1}{c}{x}& \multicolumn{1}{c|}{0}& \multicolumn{1}{c|}{\cdots}& \multicolumn{1}{c|}{x}& \multicolumn{1}{c|}{\cdots}& \multicolumn{1}{c|}{0}& \multicolumn{1}{c}{\cdots}& \multicolumn{1}{c}{}& \multicolumn{1}{c}{\cdots}\\\cline{1-1} \cline{2-2} \cline{3-3} \cline{5-5} \cline{7-7} \cline{9-9} 
\multicolumn{1}{|c}{0}& \multicolumn{1}{c}{0}& \multicolumn{1}{c|}{x}& \multicolumn{1}{c|}{\cdots}& \multicolumn{1}{c|}{x}& \multicolumn{1}{c}{\cdots}& \multicolumn{1}{c}{}& \multicolumn{1}{c|}{\cdots}& \multicolumn{1}{c|}{x}& \multicolumn{1}{c}{\cdots}\\\cline{1-1} \cline{2-2} \cline{3-3} \cline{5-5} \cline{9-9} 
\multicolumn{1}{c}{\vdots}& \multicolumn{1}{c}{\vdots}& \multicolumn{1}{c}{\vdots}& \multicolumn{1}{c}{}& \multicolumn{1}{c}{}& \multicolumn{1}{c}{}& \multicolumn{1}{c}{}& \multicolumn{1}{c}{}& \multicolumn{1}{c}{}& \multicolumn{1}{c}{}\\\cline{1-1} \cline{2-2} \cline{3-3} 
\multicolumn{1}{|c}{x}& \multicolumn{1}{c}{x}& \multicolumn{1}{c|}{x}& \multicolumn{1}{c}{}& \multicolumn{1}{c}{}& \multicolumn{1}{c}{}& \multicolumn{1}{c}{}& \multicolumn{1}{c}{}& \multicolumn{1}{c}{}& \multicolumn{1}{c}{}\\\cline{1-1} \cline{2-2} \cline{3-3} 
\multicolumn{1}{c}{\vdots}& \multicolumn{1}{c}{\vdots}& \multicolumn{1}{c}{\vdots}& \multicolumn{1}{c}{}& \multicolumn{1}{c}{}& \multicolumn{1}{c}{}& \multicolumn{1}{c}{}& \multicolumn{1}{c}{}& \multicolumn{1}{c}{}& \multicolumn{1}{c}{}\\\cline{1-1} \cline{2-2} \cline{3-3} 
\multicolumn{1}{|c}{x}& \multicolumn{1}{c}{0}& \multicolumn{1}{c|}{0}& \multicolumn{1}{c}{}& \multicolumn{1}{c}{}& \multicolumn{1}{c}{}& \multicolumn{1}{c}{}& \multicolumn{1}{c}{}& \multicolumn{1}{c}{}& \multicolumn{1}{c}{}\\\cline{1-1} \cline{2-2} \cline{3-3} 
\multicolumn{1}{c}{\vdots}& \multicolumn{1}{c}{\vdots}& \multicolumn{1}{c}{\vdots}& \multicolumn{1}{c}{}& \multicolumn{1}{c}{}& \multicolumn{1}{c}{}& \multicolumn{1}{c}{}& \multicolumn{1}{c}{}& \multicolumn{1}{c}{}& \multicolumn{1}{c}{}\\\cline{1-1} \cline{2-2} \cline{3-3} 
\multicolumn{1}{|c}{0}& \multicolumn{1}{c}{0}& \multicolumn{1}{c|}{x}& \multicolumn{1}{c}{}& \multicolumn{1}{c}{}& \multicolumn{1}{c}{}& \multicolumn{1}{c}{}& \multicolumn{1}{c}{}& \multicolumn{1}{c}{}& \multicolumn{1}{c}{}\\\cline{1-1} \cline{2-2} \cline{3-3} 
\multicolumn{1}{c}{\vdots}& \multicolumn{1}{c}{\vdots}& \multicolumn{1}{c}{\vdots}& \multicolumn{1}{c}{}& \multicolumn{1}{c}{}& \multicolumn{1}{c}{}& \multicolumn{1}{c}{}& \multicolumn{1}{c}{}& \multicolumn{1}{c}{}& \multicolumn{1}{c}{}\\
\end{array}\right]
\]
\begin{align*}
f(B_1, \Pi, \Phi) &= N_{\text{value}} + N_{\text{index}}s\\
&= (9 + 6\mu_2 + 9\mu_3) + (2 + 3\mu_2 + 4\mu_3)s\\
&= 201 + 9\lceil 28s - 10\rceil + (98 + 4\lceil 28s - 10\rceil)s\\
&\geq 111 + 310s + 112s^2
\\&\geq 147 + 263s + 112s^2
\end{align*}
\paragraph{$\Pi = [1{:}2, 3{:}3, ...], \Phi = [1{:}1, 2{:}2, 3{:}3, ...]$}
\[
\left[\begin{array}{cccccccccc}\cline{1-1} \cline{2-2} \cline{5-5} \cline{7-7} 
\multicolumn{1}{|c|}{x}& \multicolumn{1}{c|}{0}& \multicolumn{1}{c}{}& \multicolumn{1}{c|}{\cdots}& \multicolumn{1}{c|}{x}& \multicolumn{1}{c|}{\cdots}& \multicolumn{1}{c|}{x}& \multicolumn{1}{c}{\cdots}& \multicolumn{1}{c}{}& \multicolumn{1}{c}{\cdots}\\
\multicolumn{1}{|c|}{0}& \multicolumn{1}{c|}{x}& \multicolumn{1}{c}{}& \multicolumn{1}{c|}{\cdots}& \multicolumn{1}{c|}{x}& \multicolumn{1}{c|}{\cdots}& \multicolumn{1}{c|}{0}& \multicolumn{1}{c}{\cdots}& \multicolumn{1}{c}{}& \multicolumn{1}{c}{\cdots}\\\cline{1-1} \cline{2-2} \cline{3-3} \cline{5-5} \cline{7-7} \cline{9-9} 
\multicolumn{1}{c}{}& \multicolumn{1}{c|}{}& \multicolumn{1}{c|}{x}& \multicolumn{1}{c|}{\cdots}& \multicolumn{1}{c|}{x}& \multicolumn{1}{c}{\cdots}& \multicolumn{1}{c}{}& \multicolumn{1}{c|}{\cdots}& \multicolumn{1}{c|}{x}& \multicolumn{1}{c}{\cdots}\\\cline{3-3} \cline{5-5} \cline{9-9} 
\multicolumn{1}{c}{\vdots}& \multicolumn{1}{c}{\vdots}& \multicolumn{1}{c}{\vdots}& \multicolumn{1}{c}{}& \multicolumn{1}{c}{}& \multicolumn{1}{c}{}& \multicolumn{1}{c}{}& \multicolumn{1}{c}{}& \multicolumn{1}{c}{}& \multicolumn{1}{c}{}\\\cline{1-1} \cline{2-2} \cline{3-3} 
\multicolumn{1}{|c|}{x}& \multicolumn{1}{c|}{x}& \multicolumn{1}{c|}{x}& \multicolumn{1}{c}{}& \multicolumn{1}{c}{}& \multicolumn{1}{c}{}& \multicolumn{1}{c}{}& \multicolumn{1}{c}{}& \multicolumn{1}{c}{}& \multicolumn{1}{c}{}\\\cline{1-1} \cline{2-2} \cline{3-3} 
\multicolumn{1}{c}{\vdots}& \multicolumn{1}{c}{\vdots}& \multicolumn{1}{c}{\vdots}& \multicolumn{1}{c}{}& \multicolumn{1}{c}{}& \multicolumn{1}{c}{}& \multicolumn{1}{c}{}& \multicolumn{1}{c}{}& \multicolumn{1}{c}{}& \multicolumn{1}{c}{}\\\cline{1-1} 
\multicolumn{1}{|c|}{x}& \multicolumn{1}{c}{}& \multicolumn{1}{c}{}& \multicolumn{1}{c}{}& \multicolumn{1}{c}{}& \multicolumn{1}{c}{}& \multicolumn{1}{c}{}& \multicolumn{1}{c}{}& \multicolumn{1}{c}{}& \multicolumn{1}{c}{}\\\cline{1-1} 
\multicolumn{1}{c}{\vdots}& \multicolumn{1}{c}{\vdots}& \multicolumn{1}{c}{\vdots}& \multicolumn{1}{c}{}& \multicolumn{1}{c}{}& \multicolumn{1}{c}{}& \multicolumn{1}{c}{}& \multicolumn{1}{c}{}& \multicolumn{1}{c}{}& \multicolumn{1}{c}{}\\\cline{3-3} 
\multicolumn{1}{c}{}& \multicolumn{1}{c|}{}& \multicolumn{1}{c|}{x}& \multicolumn{1}{c}{}& \multicolumn{1}{c}{}& \multicolumn{1}{c}{}& \multicolumn{1}{c}{}& \multicolumn{1}{c}{}& \multicolumn{1}{c}{}& \multicolumn{1}{c}{}\\\cline{3-3} 
\multicolumn{1}{c}{\vdots}& \multicolumn{1}{c}{\vdots}& \multicolumn{1}{c}{\vdots}& \multicolumn{1}{c}{}& \multicolumn{1}{c}{}& \multicolumn{1}{c}{}& \multicolumn{1}{c}{}& \multicolumn{1}{c}{}& \multicolumn{1}{c}{}& \multicolumn{1}{c}{}\\
\end{array}\right]
\]
\begin{align*}
f(B_1, \Pi, \Phi) &= N_{\text{value}} + N_{\text{index}}s\\
&= (5 + 6\mu_2 + 5\mu_3) + (3 + 5\mu_2 + 4\mu_3)s\\
&= 197 + 5\lceil 28s - 10\rceil + (163 + 4\lceil 28s - 10\rceil)s\\
&\geq 147 + 263s + 112s^2
\end{align*}
\paragraph{$\Pi = [1{:}1, 2{:}2, 3{:}3, ...], \Phi = [1{:}3, ...]$}
\[
\left[\begin{array}{cccccccccc}\cline{1-1} \cline{2-2} \cline{3-3} \cline{5-5} \cline{7-7} 
\multicolumn{1}{|c}{x}& \multicolumn{1}{c}{0}& \multicolumn{1}{c|}{0}& \multicolumn{1}{c|}{\cdots}& \multicolumn{1}{c|}{x}& \multicolumn{1}{c|}{\cdots}& \multicolumn{1}{c|}{x}& \multicolumn{1}{c}{\cdots}& \multicolumn{1}{c}{}& \multicolumn{1}{c}{\cdots}\\\cline{1-1} \cline{2-2} \cline{3-3} \cline{5-5} \cline{7-7} 
\multicolumn{1}{|c}{0}& \multicolumn{1}{c}{x}& \multicolumn{1}{c|}{0}& \multicolumn{1}{c|}{\cdots}& \multicolumn{1}{c|}{x}& \multicolumn{1}{c}{\cdots}& \multicolumn{1}{c}{}& \multicolumn{1}{c}{\cdots}& \multicolumn{1}{c}{}& \multicolumn{1}{c}{\cdots}\\\cline{1-1} \cline{2-2} \cline{3-3} \cline{5-5} \cline{9-9} 
\multicolumn{1}{|c}{0}& \multicolumn{1}{c}{0}& \multicolumn{1}{c|}{x}& \multicolumn{1}{c|}{\cdots}& \multicolumn{1}{c|}{x}& \multicolumn{1}{c}{\cdots}& \multicolumn{1}{c}{}& \multicolumn{1}{c|}{\cdots}& \multicolumn{1}{c|}{x}& \multicolumn{1}{c}{\cdots}\\\cline{1-1} \cline{2-2} \cline{3-3} \cline{5-5} \cline{9-9} 
\multicolumn{1}{c}{\vdots}& \multicolumn{1}{c}{\vdots}& \multicolumn{1}{c}{\vdots}& \multicolumn{1}{c}{}& \multicolumn{1}{c}{}& \multicolumn{1}{c}{}& \multicolumn{1}{c}{}& \multicolumn{1}{c}{}& \multicolumn{1}{c}{}& \multicolumn{1}{c}{}\\\cline{1-1} \cline{2-2} \cline{3-3} 
\multicolumn{1}{|c}{x}& \multicolumn{1}{c}{x}& \multicolumn{1}{c|}{x}& \multicolumn{1}{c}{}& \multicolumn{1}{c}{}& \multicolumn{1}{c}{}& \multicolumn{1}{c}{}& \multicolumn{1}{c}{}& \multicolumn{1}{c}{}& \multicolumn{1}{c}{}\\\cline{1-1} \cline{2-2} \cline{3-3} 
\multicolumn{1}{c}{\vdots}& \multicolumn{1}{c}{\vdots}& \multicolumn{1}{c}{\vdots}& \multicolumn{1}{c}{}& \multicolumn{1}{c}{}& \multicolumn{1}{c}{}& \multicolumn{1}{c}{}& \multicolumn{1}{c}{}& \multicolumn{1}{c}{}& \multicolumn{1}{c}{}\\\cline{1-1} \cline{2-2} \cline{3-3} 
\multicolumn{1}{|c}{x}& \multicolumn{1}{c}{0}& \multicolumn{1}{c|}{0}& \multicolumn{1}{c}{}& \multicolumn{1}{c}{}& \multicolumn{1}{c}{}& \multicolumn{1}{c}{}& \multicolumn{1}{c}{}& \multicolumn{1}{c}{}& \multicolumn{1}{c}{}\\\cline{1-1} \cline{2-2} \cline{3-3} 
\multicolumn{1}{c}{\vdots}& \multicolumn{1}{c}{\vdots}& \multicolumn{1}{c}{\vdots}& \multicolumn{1}{c}{}& \multicolumn{1}{c}{}& \multicolumn{1}{c}{}& \multicolumn{1}{c}{}& \multicolumn{1}{c}{}& \multicolumn{1}{c}{}& \multicolumn{1}{c}{}\\\cline{1-1} \cline{2-2} \cline{3-3} 
\multicolumn{1}{|c}{0}& \multicolumn{1}{c}{0}& \multicolumn{1}{c|}{x}& \multicolumn{1}{c}{}& \multicolumn{1}{c}{}& \multicolumn{1}{c}{}& \multicolumn{1}{c}{}& \multicolumn{1}{c}{}& \multicolumn{1}{c}{}& \multicolumn{1}{c}{}\\\cline{1-1} \cline{2-2} \cline{3-3} 
\multicolumn{1}{c}{\vdots}& \multicolumn{1}{c}{\vdots}& \multicolumn{1}{c}{\vdots}& \multicolumn{1}{c}{}& \multicolumn{1}{c}{}& \multicolumn{1}{c}{}& \multicolumn{1}{c}{}& \multicolumn{1}{c}{}& \multicolumn{1}{c}{}& \multicolumn{1}{c}{}\\
\end{array}\right]
\]
\begin{align*}
f(B_1, \Pi, \Phi) &= N_{\text{value}} + N_{\text{index}}s\\
&= (9 + 6\mu_2 + 8\mu_3) + (3 + 4\mu_2 + 4\mu_3)s\\
&= 201 + 8\lceil 28s - 10\rceil + (131 + 4\lceil 28s - 10\rceil)s\\
&\geq 121 + 315s + 112s^2
\\&\geq 147 + 263s + 112s^2
\end{align*}
\paragraph{$\Pi = [1{:}1, 2{:}2, 3{:}3, ...], \Phi = [1{:}1, 2{:}3, ...]$}
\[
\left[\begin{array}{cccccccccc}\cline{1-1} \cline{5-5} \cline{7-7} 
\multicolumn{1}{|c|}{x}& \multicolumn{1}{c}{}& \multicolumn{1}{c}{}& \multicolumn{1}{c|}{\cdots}& \multicolumn{1}{c|}{x}& \multicolumn{1}{c|}{\cdots}& \multicolumn{1}{c|}{x}& \multicolumn{1}{c}{\cdots}& \multicolumn{1}{c}{}& \multicolumn{1}{c}{\cdots}\\\cline{1-1} \cline{2-2} \cline{3-3} \cline{5-5} \cline{7-7} 
\multicolumn{1}{c|}{}& \multicolumn{1}{c}{x}& \multicolumn{1}{c|}{0}& \multicolumn{1}{c|}{\cdots}& \multicolumn{1}{c|}{x}& \multicolumn{1}{c}{\cdots}& \multicolumn{1}{c}{}& \multicolumn{1}{c}{\cdots}& \multicolumn{1}{c}{}& \multicolumn{1}{c}{\cdots}\\\cline{2-2} \cline{3-3} \cline{5-5} \cline{9-9} 
\multicolumn{1}{c|}{}& \multicolumn{1}{c}{0}& \multicolumn{1}{c|}{x}& \multicolumn{1}{c|}{\cdots}& \multicolumn{1}{c|}{x}& \multicolumn{1}{c}{\cdots}& \multicolumn{1}{c}{}& \multicolumn{1}{c|}{\cdots}& \multicolumn{1}{c|}{x}& \multicolumn{1}{c}{\cdots}\\\cline{2-2} \cline{3-3} \cline{5-5} \cline{9-9} 
\multicolumn{1}{c}{\vdots}& \multicolumn{1}{c}{\vdots}& \multicolumn{1}{c}{\vdots}& \multicolumn{1}{c}{}& \multicolumn{1}{c}{}& \multicolumn{1}{c}{}& \multicolumn{1}{c}{}& \multicolumn{1}{c}{}& \multicolumn{1}{c}{}& \multicolumn{1}{c}{}\\\cline{1-1} \cline{2-2} \cline{3-3} 
\multicolumn{1}{|c|}{x}& \multicolumn{1}{c}{x}& \multicolumn{1}{c|}{x}& \multicolumn{1}{c}{}& \multicolumn{1}{c}{}& \multicolumn{1}{c}{}& \multicolumn{1}{c}{}& \multicolumn{1}{c}{}& \multicolumn{1}{c}{}& \multicolumn{1}{c}{}\\\cline{1-1} \cline{2-2} \cline{3-3} 
\multicolumn{1}{c}{\vdots}& \multicolumn{1}{c}{\vdots}& \multicolumn{1}{c}{\vdots}& \multicolumn{1}{c}{}& \multicolumn{1}{c}{}& \multicolumn{1}{c}{}& \multicolumn{1}{c}{}& \multicolumn{1}{c}{}& \multicolumn{1}{c}{}& \multicolumn{1}{c}{}\\\cline{1-1} 
\multicolumn{1}{|c|}{x}& \multicolumn{1}{c}{}& \multicolumn{1}{c}{}& \multicolumn{1}{c}{}& \multicolumn{1}{c}{}& \multicolumn{1}{c}{}& \multicolumn{1}{c}{}& \multicolumn{1}{c}{}& \multicolumn{1}{c}{}& \multicolumn{1}{c}{}\\\cline{1-1} 
\multicolumn{1}{c}{\vdots}& \multicolumn{1}{c}{\vdots}& \multicolumn{1}{c}{\vdots}& \multicolumn{1}{c}{}& \multicolumn{1}{c}{}& \multicolumn{1}{c}{}& \multicolumn{1}{c}{}& \multicolumn{1}{c}{}& \multicolumn{1}{c}{}& \multicolumn{1}{c}{}\\\cline{2-2} \cline{3-3} 
\multicolumn{1}{c|}{}& \multicolumn{1}{c}{0}& \multicolumn{1}{c|}{x}& \multicolumn{1}{c}{}& \multicolumn{1}{c}{}& \multicolumn{1}{c}{}& \multicolumn{1}{c}{}& \multicolumn{1}{c}{}& \multicolumn{1}{c}{}& \multicolumn{1}{c}{}\\\cline{2-2} \cline{3-3} 
\multicolumn{1}{c}{\vdots}& \multicolumn{1}{c}{\vdots}& \multicolumn{1}{c}{\vdots}& \multicolumn{1}{c}{}& \multicolumn{1}{c}{}& \multicolumn{1}{c}{}& \multicolumn{1}{c}{}& \multicolumn{1}{c}{}& \multicolumn{1}{c}{}& \multicolumn{1}{c}{}\\
\end{array}\right]
\]
\begin{align*}
f(B_1, \Pi, \Phi) &= N_{\text{value}} + N_{\text{index}}s\\
&= (5 + 6\mu_2 + 5\mu_3) + (3 + 5\mu_2 + 4\mu_3)s\\
&= 197 + 5\lceil 28s - 10\rceil + (163 + 4\lceil 28s - 10\rceil)s\\
&\geq 147 + 263s + 112s^2
\end{align*}
\paragraph{$\Pi = [1{:}1, 2{:}2, 3{:}3, ...], \Phi = [1{:}2, 3{:}3, ...]$}
\[
\left[\begin{array}{cccccccccc}\cline{1-1} \cline{2-2} \cline{5-5} \cline{7-7} 
\multicolumn{1}{|c}{x}& \multicolumn{1}{c|}{0}& \multicolumn{1}{c}{}& \multicolumn{1}{c|}{\cdots}& \multicolumn{1}{c|}{x}& \multicolumn{1}{c|}{\cdots}& \multicolumn{1}{c|}{x}& \multicolumn{1}{c}{\cdots}& \multicolumn{1}{c}{}& \multicolumn{1}{c}{\cdots}\\\cline{1-1} \cline{2-2} \cline{5-5} \cline{7-7} 
\multicolumn{1}{|c}{0}& \multicolumn{1}{c|}{x}& \multicolumn{1}{c}{}& \multicolumn{1}{c|}{\cdots}& \multicolumn{1}{c|}{x}& \multicolumn{1}{c}{\cdots}& \multicolumn{1}{c}{}& \multicolumn{1}{c}{\cdots}& \multicolumn{1}{c}{}& \multicolumn{1}{c}{\cdots}\\\cline{1-1} \cline{2-2} \cline{3-3} \cline{5-5} \cline{9-9} 
\multicolumn{1}{c}{}& \multicolumn{1}{c|}{}& \multicolumn{1}{c|}{x}& \multicolumn{1}{c|}{\cdots}& \multicolumn{1}{c|}{x}& \multicolumn{1}{c}{\cdots}& \multicolumn{1}{c}{}& \multicolumn{1}{c|}{\cdots}& \multicolumn{1}{c|}{x}& \multicolumn{1}{c}{\cdots}\\\cline{3-3} \cline{5-5} \cline{9-9} 
\multicolumn{1}{c}{\vdots}& \multicolumn{1}{c}{\vdots}& \multicolumn{1}{c}{\vdots}& \multicolumn{1}{c}{}& \multicolumn{1}{c}{}& \multicolumn{1}{c}{}& \multicolumn{1}{c}{}& \multicolumn{1}{c}{}& \multicolumn{1}{c}{}& \multicolumn{1}{c}{}\\\cline{1-1} \cline{2-2} \cline{3-3} 
\multicolumn{1}{|c}{x}& \multicolumn{1}{c|}{x}& \multicolumn{1}{c|}{x}& \multicolumn{1}{c}{}& \multicolumn{1}{c}{}& \multicolumn{1}{c}{}& \multicolumn{1}{c}{}& \multicolumn{1}{c}{}& \multicolumn{1}{c}{}& \multicolumn{1}{c}{}\\\cline{1-1} \cline{2-2} \cline{3-3} 
\multicolumn{1}{c}{\vdots}& \multicolumn{1}{c}{\vdots}& \multicolumn{1}{c}{\vdots}& \multicolumn{1}{c}{}& \multicolumn{1}{c}{}& \multicolumn{1}{c}{}& \multicolumn{1}{c}{}& \multicolumn{1}{c}{}& \multicolumn{1}{c}{}& \multicolumn{1}{c}{}\\\cline{1-1} \cline{2-2} 
\multicolumn{1}{|c}{x}& \multicolumn{1}{c|}{0}& \multicolumn{1}{c}{}& \multicolumn{1}{c}{}& \multicolumn{1}{c}{}& \multicolumn{1}{c}{}& \multicolumn{1}{c}{}& \multicolumn{1}{c}{}& \multicolumn{1}{c}{}& \multicolumn{1}{c}{}\\\cline{1-1} \cline{2-2} 
\multicolumn{1}{c}{\vdots}& \multicolumn{1}{c}{\vdots}& \multicolumn{1}{c}{\vdots}& \multicolumn{1}{c}{}& \multicolumn{1}{c}{}& \multicolumn{1}{c}{}& \multicolumn{1}{c}{}& \multicolumn{1}{c}{}& \multicolumn{1}{c}{}& \multicolumn{1}{c}{}\\\cline{3-3} 
\multicolumn{1}{c}{}& \multicolumn{1}{c|}{}& \multicolumn{1}{c|}{x}& \multicolumn{1}{c}{}& \multicolumn{1}{c}{}& \multicolumn{1}{c}{}& \multicolumn{1}{c}{}& \multicolumn{1}{c}{}& \multicolumn{1}{c}{}& \multicolumn{1}{c}{}\\\cline{3-3} 
\multicolumn{1}{c}{\vdots}& \multicolumn{1}{c}{\vdots}& \multicolumn{1}{c}{\vdots}& \multicolumn{1}{c}{}& \multicolumn{1}{c}{}& \multicolumn{1}{c}{}& \multicolumn{1}{c}{}& \multicolumn{1}{c}{}& \multicolumn{1}{c}{}& \multicolumn{1}{c}{}\\
\end{array}\right]
\]
\begin{align*}
f(B_1, \Pi, \Phi) &= N_{\text{value}} + N_{\text{index}}s\\
&= (5 + 6\mu_2 + 5\mu_3) + (3 + 5\mu_2 + 4\mu_3)s\\
&= 197 + 5\lceil 28s - 10\rceil + (163 + 4\lceil 28s - 10\rceil)s\\
&\geq 147 + 263s + 112s^2
\end{align*}
\paragraph{$\Pi = [1{:}1, 2{:}2, 3{:}3, ...], \Phi = [1{:}1, 2{:}2, 3{:}3, ...]$}
\[
\left[\begin{array}{cccccccccc}\cline{1-1} \cline{5-5} \cline{7-7} 
\multicolumn{1}{|c|}{x}& \multicolumn{1}{c}{}& \multicolumn{1}{c}{}& \multicolumn{1}{c|}{\cdots}& \multicolumn{1}{c|}{x}& \multicolumn{1}{c|}{\cdots}& \multicolumn{1}{c|}{x}& \multicolumn{1}{c}{\cdots}& \multicolumn{1}{c}{}& \multicolumn{1}{c}{\cdots}\\\cline{1-1} \cline{2-2} \cline{5-5} \cline{7-7} 
\multicolumn{1}{c|}{}& \multicolumn{1}{c|}{x}& \multicolumn{1}{c}{}& \multicolumn{1}{c|}{\cdots}& \multicolumn{1}{c|}{x}& \multicolumn{1}{c}{\cdots}& \multicolumn{1}{c}{}& \multicolumn{1}{c}{\cdots}& \multicolumn{1}{c}{}& \multicolumn{1}{c}{\cdots}\\\cline{2-2} \cline{3-3} \cline{5-5} \cline{9-9} 
\multicolumn{1}{c}{}& \multicolumn{1}{c|}{}& \multicolumn{1}{c|}{x}& \multicolumn{1}{c|}{\cdots}& \multicolumn{1}{c|}{x}& \multicolumn{1}{c}{\cdots}& \multicolumn{1}{c}{}& \multicolumn{1}{c|}{\cdots}& \multicolumn{1}{c|}{x}& \multicolumn{1}{c}{\cdots}\\\cline{3-3} \cline{5-5} \cline{9-9} 
\multicolumn{1}{c}{\vdots}& \multicolumn{1}{c}{\vdots}& \multicolumn{1}{c}{\vdots}& \multicolumn{1}{c}{}& \multicolumn{1}{c}{}& \multicolumn{1}{c}{}& \multicolumn{1}{c}{}& \multicolumn{1}{c}{}& \multicolumn{1}{c}{}& \multicolumn{1}{c}{}\\\cline{1-1} \cline{2-2} \cline{3-3} 
\multicolumn{1}{|c|}{x}& \multicolumn{1}{c|}{x}& \multicolumn{1}{c|}{x}& \multicolumn{1}{c}{}& \multicolumn{1}{c}{}& \multicolumn{1}{c}{}& \multicolumn{1}{c}{}& \multicolumn{1}{c}{}& \multicolumn{1}{c}{}& \multicolumn{1}{c}{}\\\cline{1-1} \cline{2-2} \cline{3-3} 
\multicolumn{1}{c}{\vdots}& \multicolumn{1}{c}{\vdots}& \multicolumn{1}{c}{\vdots}& \multicolumn{1}{c}{}& \multicolumn{1}{c}{}& \multicolumn{1}{c}{}& \multicolumn{1}{c}{}& \multicolumn{1}{c}{}& \multicolumn{1}{c}{}& \multicolumn{1}{c}{}\\\cline{1-1} 
\multicolumn{1}{|c|}{x}& \multicolumn{1}{c}{}& \multicolumn{1}{c}{}& \multicolumn{1}{c}{}& \multicolumn{1}{c}{}& \multicolumn{1}{c}{}& \multicolumn{1}{c}{}& \multicolumn{1}{c}{}& \multicolumn{1}{c}{}& \multicolumn{1}{c}{}\\\cline{1-1} 
\multicolumn{1}{c}{\vdots}& \multicolumn{1}{c}{\vdots}& \multicolumn{1}{c}{\vdots}& \multicolumn{1}{c}{}& \multicolumn{1}{c}{}& \multicolumn{1}{c}{}& \multicolumn{1}{c}{}& \multicolumn{1}{c}{}& \multicolumn{1}{c}{}& \multicolumn{1}{c}{}\\\cline{3-3} 
\multicolumn{1}{c}{}& \multicolumn{1}{c|}{}& \multicolumn{1}{c|}{x}& \multicolumn{1}{c}{}& \multicolumn{1}{c}{}& \multicolumn{1}{c}{}& \multicolumn{1}{c}{}& \multicolumn{1}{c}{}& \multicolumn{1}{c}{}& \multicolumn{1}{c}{}\\\cline{3-3} 
\multicolumn{1}{c}{\vdots}& \multicolumn{1}{c}{\vdots}& \multicolumn{1}{c}{\vdots}& \multicolumn{1}{c}{}& \multicolumn{1}{c}{}& \multicolumn{1}{c}{}& \multicolumn{1}{c}{}& \multicolumn{1}{c}{}& \multicolumn{1}{c}{}& \multicolumn{1}{c}{}\\
\end{array}\right]
\]
\begin{align*}
f(B_1, \Pi, \Phi) &= N_{\text{value}} + N_{\text{index}}s\\
&= (9 + 6\mu_2 + 4\mu_3) + (3 + 6\mu_2 + 4\mu_3)s\\
&= 201 + 4\lceil 28s - 10\rceil + (195 + 4\lceil 28s - 10\rceil)s\\
&\geq 161 + 267s + 112s^2
\\&\geq 147 + 263s + 112s^2
\end{align*}

\section{Improved Overlap Heuristic}\label{app:overlappartitioner}
Here, we optimize the overlap heuristic of \cite{vuduc_oski:_2005,
noauthor_developer_2020} to minimize random access to the column hash table
$h$. Instead of checking whether $h_j$ has been set to true, we can check
whether $h_j = i$. When we start a new group, we will use a different value
of $i$ and avoid reinitialization of $h$. When checking if row $i'$ should be
added to the current part, we can also add $v_{i'}$ to the hash table to
avoid reading $v_{i'}$ twice. However, this may overwrite the entries which
contain $i$ with $i'$, so if $h_j = i$, we can indicate that it was $i$ when
we overwrote it by negating $i'$ before writing to $h$. Thus, our hash table
can check the similarity between the two rows and add a new row at the same
time, reducing the number of random accesses to a minimum (once per nonzero
element). Our improved overlap pseudocode is presented in Algorithm
\ref{alg:overlappartitioner}.

\begin{algorithm}\label{alg:overlappartitioner}
    Given an overlap similarity $\rho$, partition the rows of $m \times
    n$ matrix $A$ producing no part with more than $u_{\max}$ rows. Return
    $\Pi$, $pos$, and $ofs$.
    \begin{algorithmic}[1]
        \small
        \Require $m > 1$, $u_{\max} > 1$, $0 < \rho \leq 1$.
        \Function{OverlapPartitioner}{$A$, $\rho$}
            \State{\textrm{Allocate length-$(m + 1)$ vector $spl_\Pi$}}
            \State{\textrm{Allocate length-$n$ vector $h$ initialized to $0$}}
            \State{$spl_\Pi[1] \gets 1$}
            \State{$i \gets 1$}
            \State{$K \gets 0$}
            \State{$d \gets |v_1|$}
            \For{$i' \gets 2 \textbf{ to } m$}
                \State{$d' \gets d$}
                \State{$c \gets 0$}
                \For{$j \gets v_{i'}$ in ascending order}
                    \If{$h[j] = \pm i$}
                        \State{$c \gets c + 1$}
                        \State{$h[j] \gets -i'$}
                    \ElsIf{$i < h[j]$}
                        \State{$h[j] \gets i'$}
                    \ElsIf{$h[j] < -i$}
                        \State{$c \gets c + 1$}
                        \State{$h[j] \gets -i'$}
                    \Else
                        \State{$d' \gets d' + 1$}
                        \State{$h[j] \gets i'$}
                    \EndIf
                \EndFor
                \State{$u \gets i' - i$}
                \If{$u = u_{\max}$ \textbf{or} $c < \rho \cdot \min(|v_i|, |v_i'|)$}
                    \State{$K \gets K + 1$} \Comment{Start a new partition.}
                    \State{$spl_\Pi[K + 1] \gets i'$}\label{alg:overlappartitioner:current}
                    \State{$i \gets i'$}
                    \State{$d \gets |v_{i'}|$}
                \Else
                    \State{$d \gets d'$} \Comment{Expand current partition.}
                \EndIf
                \State{$K \gets K + 1$}
                \State{$u \gets (m + 1) - i$}
                \State{$spl_\Pi[K + 1] \gets m + 1$}\label{alg:overlappartitioner:last}
            \EndFor
            \State{\Return $spl_\Pi[1:K + 1]$}
        \EndFunction
    \end{algorithmic}
\end{algorithm}

\end{document}